\newif\ifcomments\commentstrue
\newif\iflong\longfalse
\newif\iftwocolumns\twocolumnsfalse

\documentclass{lmcs} 


\usepackage{hyperref}

\usepackage{macros/macros}
\usepackage{macros/laetitia_macro}

\usepackage{thmtools} 
\usepackage{thm-restate}
\usepackage{wrapfig}

\begin{document}

\title{On the Impact of the Communication Model on Realisability}

\author[C.~Di Giusto]{Cinzia {Di Giusto}\lmcsorcid{0000-0003-1563-6581}}[a]

\author[E.~Lozes]{Etienne Lozes\lmcsorcid{0000-0001-8505-585X}}

 \author[P.~Urso]{Pascal Urso}

\address{Université Côte d’Azur, CNRS, I3S, France}	
\email{cinzia.di-giusto, etienne.lozes, pascal.urso@univ-cotedazr.fr}  

\begin{abstract}
	% !TEX root =  ../main.tex
%!TEX spellcheck = en_GB

%TITLE: Realisability and Complementability of Multiparty Session Types 
%
%
%
%
%Multiparty session types (MPST) are a type-based approach for specifying message-passing protocols.
%MPST rely on the notion of global type specifying the global behaviour and local types which are the projections of the global behaviour onto each local participant. An essential property of global types is realisability, i.e., whether the composition of the local behaviours conforms to those specified by the global type. 
%MPST realisability has been studied extensively but almost solely for peer-to-peer communications. 
%In this paper, we generalise the realisability problem to other communication models, like bag, causaly ordered, or synchronous communications.
%First, we show that if a global type is realisable in an arbitrary communication model,
%then it is realisable in the synchronous model. Second, we show that if a global type is realisable in the synchronous model, then it is complementable, in the sense that there exists a global type that describes the complementary behaviour of the original global type.
%Third, we give an algorithm to decide whether a complementable global type, given with an explicit complement, is 
%realisable in a given communication model. The algorithm is PSPACE in the size of the global type and its complement, when the communication model is fixed. Finally, we propose a complementation construction for global types with sender driven choice with a linear increase in the size of the global type.

Multiparty Session Types (MPST) provide a type-theoretic foundation for specifying and verifying communication protocols in distributed systems. MPST rely on the notion of global type which specifies the global behaviour and local types, which are the projections of the global behaviour onto each local participant. A central notion in MPST is realisability — whether local implementations derived from a global specification correctly realise the intended protocol under a given communication model. While realisability has been extensively studied under peer-to-peer semantics, it remains poorly understood in alternative communication models such as bag-based, causally ordered, or synchronous communications.

In this paper, we develop a unified framework for reasoning about realisability and subtyping across a spectrum of communication models. We show that the communication model does not impact the notion of subtyping, but that it impacts the notion of realisability. We introduce several decision procedures for subtyping checking and realisability checking with complexities ranging from NLOGSPACE to EXPSPACE depending on the assumptions made on the global types, in particular depending on their complementability and the size of a given complement. 
\end{abstract}

\maketitle

%\todo[inline]{for final submission put everything in a single file!}

\section{Introduction}
	% !TEX root =  ../main.tex
%!TEX spellcheck = en_GB

The design and verification of communication protocols in concurrent and distributed systems is a central concern in programming languages. A prominent formalism is Multiparty Session Types (MPST)—a type-theoretic framework that enables specification and static verification of structured interactions among multiple communicating entities~\cite{DBLP:conf/concur/VasconcelosH93,HondaYC08,10.1145/2873052,DBLP:conf/icdcit/YoshidaG20}. MPST  define a global communication protocol that can be projected into local types for each participant, ensuring communication safety and deadlock-freedom.

In this work, we focus on the \emph{realisability problem}, a foundational question in the theory of session types,  that concerns the soundness of the projection mechanism:
the realisability problem asks whether the local types obtained by projecting a global communication specification can be safely used for restricting processes interactions to the original global protocol. 

This problem is central to the goal of \emph{correct-by-construction} protocol design, as it ensures that local implementations derived from a specification do not introduce communication mismatches, deadlocks, or unintended behaviours. In the literature on MPST, this challenge has appeared under various names, including \emph{implementability} and \emph{projectability}, reflecting different  on the same core issue. In this work, we adopt the term \emph{realisability}, following the terminology introduced by Alur et al.~\cite{DBLP:journals/tcs/AlurEY05}, to emphasise the connection with automata-theoretic approaches. 

The goal of this paper is to tackle the realisability problem for a generalised class of global types, where protocols are specified as \emph{sets of message sequence charts (MSCs)} rather than as a single linearised control flow. These sets (in the style of Alur's work) are specified via MSC graphs (deterministic finite set automata over the language of synchronous exchanges of messages). This representation naturally captures concurrency, and partial orderings of events—features common in real-world protocols. Our framework is parametric in the \emph{communication model}, allowing us to reason uniformly across both synchronous and asynchronous settings, from the fully asynchronous model (bag), passing through FIFO models (peer-to-peer, mailbox), to fully synchronous one. Moreover, we integrate a form of semantic \emph{subtyping}, enabling the refinement of global behaviours.

Our main contribution is the identification of a family of communication models, that we call causally-closed for which we can reduce realisability to realisability in the synchronous model, provided that the behaviours of the projection are quasi-synchronous, orphan and deadlock free. Intuitively, a communication model $\acommunicationmodel$ is causally-closed if for each execution in $\acommunicationmodel$ all causally equivalent (up-to permutations of causally independent actions) executions are also contained in $\acommunicationmodel$. Moreover, an execution is quasi-synchronous if it is causally equivalent to an execution where all sent messages can be immediately received.
Now, it is known that realisability in the synchronous model is generally undecidable. 
A key insight of our work is the connection between \emph{realisability} and a property we call \emph{complementability}: the ability to construct a global type that captures all behaviours that are complementary (i.e., not allowed) with respect to a given global specification. 
It turns out that if the complement of a global type is given (or can be computed) then realisability in the synchronous model is decidable.

In the rest of the introduction we give three examples to motivate our choices and results.
%\todo[inline]{parler plus de model de communication}

\subsection{ Global Types as Automata for Real-World Protocols}
As mentioned above, we adopt the view of \emph{global types as automata} where transitions capture exchanges between participants. This automaton-based perspective supports precise reasoning about the evolution of protocol states, accommodating branching, concurrency, and synchronisation.
It is more general than the standard process-algebra approach as it admits: 1) mixed choice 2) non sender-driven choice 3) some form of parallel behaviour. 

Take for instance a simplified variant of the MQTT publish/subscribe protocol. Two clients, $c_1$ and $c_2$, communicate through a multi-threaded broker. 
Each client may independently send a publish ($p$) message to their corresponding broker thread ($b$ or $b'$), which forwards to the other client which itself responds with an acknowledgment ($a$). 
When a client has finished to publish it send an end message to the broker. 
When all clients have finished, the broker sends a final acknowledgement and the protocol terminates. 
This global behaviour is depicted in Figure \ref{fig:mqtt-globaltype}. 
For readability reasons, we borrow for this example a BPMN inspired notation where nodes labelled with $\|$ to represent thread creation and thread joins, whereas $+$ nodes represent choices or merge of distinct branches. A similar notation has been used in~\cite{DBLP:conf/popl/LangeTY15}. We can compile this formalism to finite state automata (when the Petri net underlying the BPMN-like process is 1-safe, or more generally bounded). We illustrate the compilation to a finite state automaton on Figure~\ref{fig:mqtt-subtype}.

\begin{figure}
    \begin{tikzpicture}[>=stealth, node distance=2cm, auto]
  \node[noeud, initial above, initial text={}, initial distance=3mm, minimum size=1.5em] (q0) {$\|$};

  % États pour publication c1 
  \node[noeud, minimum size=1.5em] (q1) [below left=1em and 4em of q0] {$+$};
  \node[bloc, minimum size=1.5em] (q3) [below left =1.5em and 4em of q1] {$\begin{array}{c}\gtlabel{c_1}{b}{p};\\\gtlabel{b}{c_2}{p};\\\gtlabel{c_2}{b}{a}\end{array}$};
  \node[bloc, minimum size=1.5em] (q4) [below =1.5em of q1] {$\gtlabel{c_1}{b}{end}$};

  % États pour publication c2
  \node[noeud, minimum size=1.5em] (q2) [below right=1em and 4em of q0] {$+$};
  \node[bloc, minimum size=1.5em] (q5) [below right =1.5em and 4em of q2] {$\begin{array}{c}\gtlabel{c_2}{b'}{p};\\\gtlabel{b'}{c_1}{p};\\\gtlabel{c_1}{b'}{a}\end{array}$};
  \node[bloc, minimum size=1.5em] (q6) [below =1.5em of q2] {$\gtlabel{c_2}{b'}{end}$};
            
  % États pour end
  \node[noeud, minimum size=1.5em] (q7) [below=7em  of q0] {$\|$};
  \node[bloc, minimum size=1.5em] (q8a) [below=3.5em  of q4] {$\gtlabel{b}{c_1}{end}$};
  \node[bloc, minimum size=1.5em] (q8b) [below=3.5em  of q6] {$\gtlabel{b'}{c_2}{end}$};
  \node[noeud, accepting, minimum size=1.5em] (q9) [below=3em of q7] {};

  \draw[->] (q0) -| (q1);
  \draw[->] (q0) -| (q2);
  
  % Boucle de publication de c1
  \draw[->] (q1.south west) -| (q3.north);
  \draw[->] (q3.south) --++ (0,-1.5em) --++ (-3em,0) |- (q1.west);

  % Boucle de publication de c2
  \draw[->] (q2.south east) -| (q5.north);
  \draw[->] (q5.south) --++ (0,-1.5em) --++ (3em,0) |-   (q2.east);
  
  %end protocol
  \draw[->] (q1) -- (q4.north);
  \draw[->] (q4) -- (q7.north west);
  \draw[->] (q2) -- (q6.north);
  \draw[->] (q6) -- (q7.north east);
  \draw[->] (q7.west) -| (q8a.north);
  \draw[->] (q7.east) -| (q8b.north);
  \draw[->] (q8a) to node[right] {} (q9);
  \draw[->] (q8b) to node[right] {} (q9);
\end{tikzpicture}
    \caption{The MQTT global type}\label{fig:mqtt-globaltype}
\end{figure}
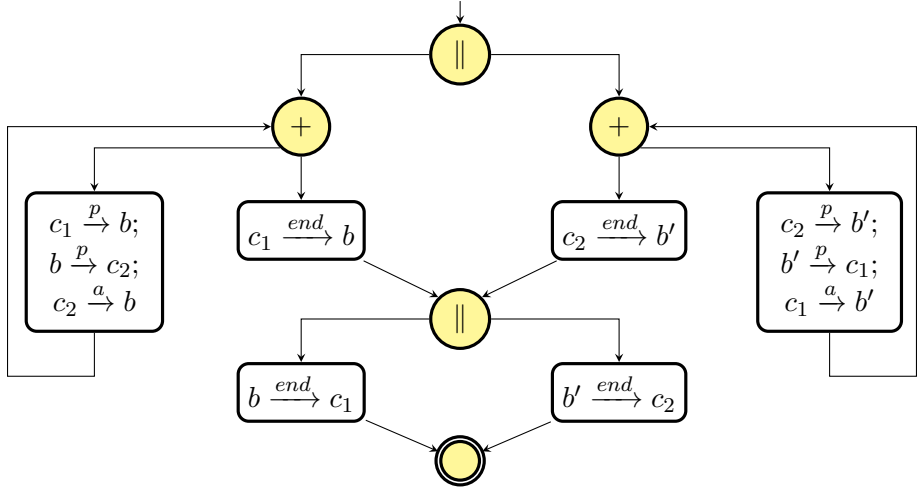

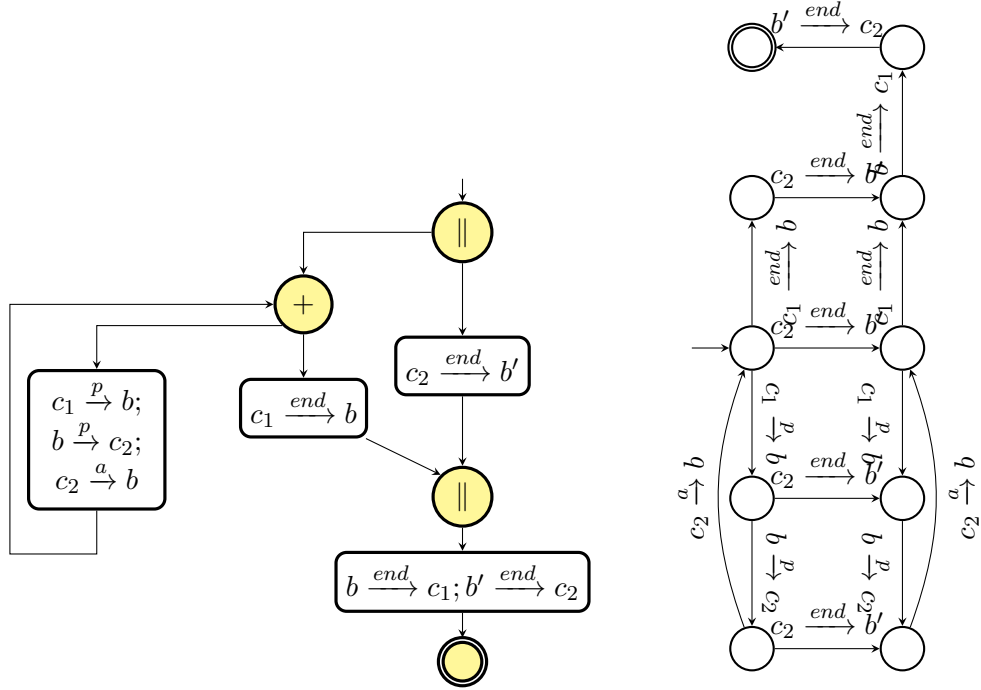
\begin{figure}
    \begin{tikzpicture}[>=stealth, node distance=2cm, auto]

  \node[noeud, initial above, initial text={}, initial distance=3mm, minimum size=1.5em] (q0) {$\|$};

  % États pour publication c1 
  \node[noeud, minimum size=1.5em] (q1) [below left=1em and 4em of q0] {$+$};
  \node[bloc, minimum size=1.5em] (q3) [below left =1.5em and 4em of q1] {$\begin{array}{c}\gtlabel{c_1}{b}{p};\\\gtlabel{b}{c_2}{p};\\\gtlabel{c_2}{b}{a}\end{array}$};
  \node[bloc, minimum size=1.5em] (q4) [below =1.5em of q1] {$\gtlabel{c_1}{b}{end}$};

  % États pour publication c2
%  \node[noeud, minimum size=1.5em] (q2) [below right=1em and 4em of q0] {$+$};
  % \node[bloc, minimum size=1.5em] (q5) [below right =1.5em and 4em of q2] {$\begin{array}{c}\gtlabel{c_2}{b}{p};\\\gtlabel{b}{c_1}{p};\\\gtlabel{c_1}{b}{a}\end{array}$};
  \node[bloc, minimum size=1.5em] (q6) [below =2.5em of q0] {$\gtlabel{c_2}{b'}{end}$};
            
  % États pour end
  \node[noeud, minimum size=1.5em] (q7) [below=7em  of q0] {$\|$};
  \node[bloc, minimum size=1.5em] (q8) [below=.8em  of q7] {$\gtlabel{b}{c_1}{end};\gtlabel{b'}{c_2}{end}$};
  \node[noeud, accepting, minimum size=1.5em] (q9) [below=3.8em of q7] {};

  \draw[->] (q0) -| (q1);
%  \draw[->] (q0) -| (q2);
  
  % Boucle de publication de c1
  \draw[->] (q1.south west) -| (q3.north);
  \draw[->] (q3.south) --++ (0,-1.5em) --++ (-3em,0) |- (q1.west);

  % Boucle de publication de c2
  % \draw[->] (q2.south east) -| (q5.north);
  % \draw[->] (q5.south) --++ (0,-1.5em) --++ (3em,0) |-   (q2.east);
  
  %end protocol
  \draw[->] (q1) -- (q4.north);
  \draw[->] (q4) -- (q7.north west);
  \draw[->] (q0) -- (q6.north);
  \draw[->] (q6) -- (q7.north);
  \draw[->] (q7) -- (q8.north);
  \draw[->] (q8) -- (q9);

    \begin{scope}[xshift=10em,yshift=-4em]
      \tikzstyle{every state}=[minimum size=1.5em,thick]
      \node[state,initial,initial text={}] (q00) at (0,0) {};
      \node[state] (q01) at (2,0) {};
      \node[state] (q10) at (0,-2) {};
      \node[state] (q11) at (2,-2) {};
      \node[state] (q20) at (0,-4) {};
      \node[state] (q21) at (2,-4) {};
      \node[state] (q30) at (0,2) {};
      \node[state] (q31) at (2,2) {};
      \node[state,accepting] (q40) at (0,4) {};
      \node[state] (q41) at (2,4) {};
      \draw[->] (q00) edge node[above,sloped] {$\gtlabel{c_1}{b}{p}$} (q10);
      \draw[->] (q00) edge node[above] {$\gtlabel{c_2}{b'}{end}$} (q01); 
      \draw[->] (q10) edge node[above,sloped] {$\gtlabel{c_2}{b'}{end}$} (q11);
      \draw[->] (q01) edge node[below,sloped] {$\gtlabel{c_1}{b}{p}$} (q11);
      \draw[->] (q10) edge node[above,sloped] {$\gtlabel{b}{c_2}{p}$} (q20);
      \draw[->] (q20) edge node[above,sloped] {$\gtlabel{c_2}{b'}{end}$} (q21);
      \draw[->] (q11) edge node[below,sloped] {$\gtlabel{b}{c_2}{p}$} (q21);
%      \draw[->] (q20) edge node[above,sloped] {$\gtlabel{c_1}{b}{end}$} (q30);
      \draw[->] (q30) edge node[above,sloped] {$\gtlabel{c_2}{b'}{end}$} (q31);
%      \draw[->] (q21) edge node[below,sloped] {$\gtlabel{c_1}{b}{end}$} (q31);
      \draw[->,bend left=20] (q20) edge node[above,sloped] {$\gtlabel{c_2}{b}{a}$} (q00);
      \draw[->,bend right=20] (q21) edge node[below,sloped] {$\gtlabel{c_2}{b}{a}$} (q01);
      \draw[->] (q31) edge node[above,sloped] {$\gtlabel{b}{c_1}{end}$} (q41);
      \draw[->] (q41) edge node[above,sloped] {$\gtlabel{b'}{c_2}{end}$} (q40);
      \draw[->] (q00) edge node[below,sloped] {$\gtlabel{c_1}{b}{end}$} (q30);
      \draw[->] (q01) edge node[above,sloped] {$\gtlabel{c_1}{b}{end}$} (q31);
    \end{scope}
\end{tikzpicture}
    \caption{A subtype of the MQTT global type and its representation as an automaton.}\label{fig:mqtt-subtype}
\end{figure}

Crucially, the choice of who sends first is not centralised—either client may initiate communication at any time, and the protocol loops accordingly. More precisely, the global type is not sender-driven, as for instance the initial state of the automaton in Figure~\ref{fig:mqtt-subtype} offers a choice between $\gtlabel{c_1}{b}{p}$ and $\gtlabel{c_2}{b'}{p}$. 
This is a common pattern in real-world protocols, where multiple participants can initiate communication independently\footnote{Another example, that we will not detail here, could come from the WebSocket protocol used by
Internet clients for low-latency, full-duplex communication with a server.
After the initial HTTP \emph{Upgrade} handshake, the channel becomes a persistent bidirectional byte stream.  
Either endpoint, the client or the server, may send a frame at any moment. 
The concurrency is even stronger than in MQTT: messages cross in flight, and neither side knows who
will speak next.}.
Despite not being \emph{sender-driven}, this global type is \emph{realisable} in our setting. 

\subsection{Beyond P2P Semantics: the Need for General Communication Models}
Historically,  research on asynchronous MPST has predominantly focused on peer-to-peer communications. However, in practice, many systems are built on top of different communication models, such as unordered (or bag-based), mailbox-based, causally ordered, or bus-based communications. Moreover, as mentioned above the question of implementability has been generally confined to adding some constraints on the syntax of global types (e.g., sender-driven choice that admits only non-mixed choice where sends are performed from the same participant). Such constraints respond to general properties that may seem arbitrary and loosely connected to the two guarantees mentioned above: session conformance and deadlock freedom. 

In this paper, we have a two-fold objective; we aim at:
\begin{enumerate}
\item introducing a theory that is parametric on the communication model,
\item weakening the constraints on global types that guarantee implementability.  
\end{enumerate}

The communication models mentioned before form a hierarchy (see Figure \ref{fig:communication-models-hierarchy}) from more relaxed to more constrained communication models, either in terms
of executions~\cite{DBLP:conf/fm/ChevrouH0Q19}, or in terms of message sequence 
charts~\cite{DBLP:journals/pacmpl/GiustoFLL23}.
Several interesting properties turn out to be monotonic in the communication model in the following sense: if a property is satisfied in a communication model, then it is preserved for more relaxed (resp. more constrained) ones. 
This  is true, for instance, when checking for reachability of a configuration,
for unspecified receptions (the first message in a queue cannot be received by its
destination process), orphan messages (i.e., messages sent by a peer but never received), etc.
Unfortunately, deadlock-freedom, and therefore deadlock-free realisability, are not monotonic.
To better understand this, take the following  two examples. First, consider an example of a global type (in standard MPST syntax) that is realisable in the peer-to-peer model but not in the mailbox one:  $$\gt_1 = \gtlabel{p}{q}{\msg_1} ;\  \gtlabel{r}{q}{\msg_2} ;\  \mathsf{end}.$$
 This global type is realisable in the peer-to-peer model, 
because the projection $p\parallel q \parallel r= !\msg_1 \parallel ?\msg_1.?\msg_2\parallel !\msg_2$ is deadlock-free and session conformant. However, this global type is not realisable in the mailbox model, where all incoming messages of process $q$ are stored in a unique FIFO queue,
because $\msg_2$ could be sent and hence stored in the queue before $\msg_1$  and cause a deadlock.
Second, consider this example where the converse happens: 
$$ 
\begin{array}{lll}
\gt_2 & = & \gtlabel{p}{q}{\msg_1} ;\  \textcolor{red}{\gtlabel{p}{r}{\msg_2}} ;\  \gtlabel{p}{q}{\msg_3} 
        ;\  \textcolor{blue}{\gtlabel{q}{r}{\msg_4}} ;\ \mathsf{end} 
\\
& + & \gtlabel{p}{q}{\msg_1'} ;\  \textcolor{blue}{\gtlabel{q}{r}{\msg_4}} ;\  \gtlabel{q}{p}{\msg_3'} 
        ;\  \textcolor{red}{\gtlabel{p}{r}{\msg_2}} ;\ \mathsf{end}  
\end{array}
$$
This global type $\gt_2$ is depicted in Figure~\ref{fig:gt2}, it features two possible
scenarios, starting either with $\msg_1$ or $\msg_1'$, (that  are depicted on the left). 
It is not realisable in the peer-to-peer model, because the projection 
$p\parallel q \parallel r$ also admits the interaction where  $p$ and $q$ follow the first scenario, but $r$ follows the second scenario, resulting in a non synchronisable interaction (also depicted in Figure~\ref{fig:gt2}). However, this global type is realisable in the mailbox model, because the messages sent to $r$ are stored in a unique FIFO queue, and their order in the queue, is sufficient for  $r$ to  not confuse the two scenarios.
Moreover, this global type, as well as the previous one, are realisable in the synchronous communication model, where ambiguity is resolved by the message ordering.   
\begin{figure}
    \begin{tikzpicture}[scale=0.7]
        \begin{scope}
            \newproc{0}{p}{-2.7};
            \newproc{2}{q}{-2.7};
            \newproc{4}{r}{-2.7};

            \newmsgm{0}{2}{-0.3}{-0.3}{1}{0.5}{black};
            \newmsgm{0}{4}{-1.0}{-1.0}{2}{0.25}{red};
            \newmsgm{0}{2}{-1.7}{-1.7}{3}{0.5}{black};
            \newmsgm{2}{4}{-2.4}{-2.4}{4}{0.5}{blue};
        \end{scope}

        \begin{scope}[xshift=7cm]
            \newproc{0}{p}{-2.7};
            \newproc{2}{q}{-2.7};
            \newproc{4}{r}{-2.7};

            \newmsgm{0}{2}{-0.3}{-0.3}{1'}{0.5}{black};
            \newmsgm{2}{4}{-1.0}{-1.0}{4}{0.5}{blue};
            \newmsgm{2}{0}{-1.7}{-1.7}{3'}{0.5}{black};
            \newmsgm{0}{4}{-2.4}{-2.4}{2}{0.25}{red};
        \end{scope}

        \begin{scope}[xshift=14cm]
            \newproc{0}{p}{-2.7};
            \newproc{2}{q}{-2.7};
            \newproc{4}{r}{-2.7};

            \newmsgm{0}{2}{-0.3}{-0.3}{1}{0.5}{black};
            \newmsgm{0}{4}{-1.0}{-2.4}{2}{0.25}{red};
            \draw[>=latex,->] (0,-1.45) -- node [below, pos = 0.25] {$\msg_{3}$} (2,-1.45);
            \newmsgm{2}{4}{-1.95}{-1.95}{4}{0.5}{blue};
        \end{scope}

    \end{tikzpicture}
    \caption{The two interactions specified by $\gt_2$, and another unspecified interaction that happens in the projected system in the p2p semantics}\label{fig:gt2}
\end{figure}
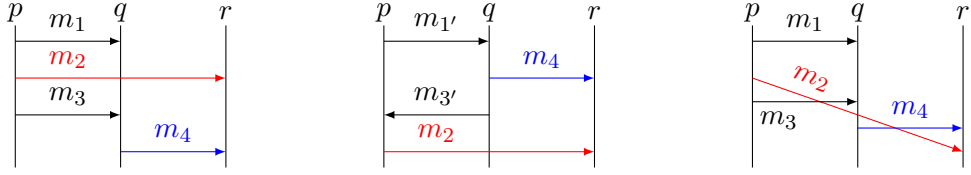

These two examples show that realisability is not monotonic,
and suggest that every new communication model should prompt the development of a new ad hoc theory of realisable global types, with its own syntactic conditions. 
% In particular the framework proposed by Stutz et al. \cite{DBLP:conf/ecoop/Stutz23}, which is the closest approach to ours (cfr. the Related Work section) is not amenable to a generalisation to a generic communication model.

%\subsection{Subtyping: Protocol Flexibility and Safe Refinement}
%While exact protocol descriptions are useful for specification, real systems often require adaptability—for example, to accommodate extensions, fallback behaviours, or optional features. This motivates the introduction of \emph{subtyping}: a mechanism for safely refining global types without compromising realisability.
%Returning to the MQTT example, suppose we define a \emph{subtype} of the protocol where only one client is allowed to publish, see Figure~\ref{fig:mqtt-subtype}.
%
%In our setting, we give a semantic definition of subtyping, based on the idea that a global type
%$\gt$ denotes a set of executions $\executionsofcfsms{\gt}{\acommunicationmodel}$ that is parameterised by a communication model $\acommunicationmodel$. We therefore say that $\gt$ is a subtype of $\gt'$ in a communication model $\acommunicationmodel$ if 
%$\executionsofcfsms{\gt}{\acommunicationmodel}\subseteq \executionsofcfsms{\gt'}{\acommunicationmodel}$. 
%However, it turns out that the communication model does not impacts our notion of subtyping essentially because global types also admit
%    a MSC-based semantics that is agnostic of the communication model and that defines the same semantic subtyping relation.

%\todo[inline]{subsection{Monotonic behavior}TODO}

\subsection{Subtyping: Protocol Flexibility and Safe Refinement}
While exact protocol descriptions are useful for specification, real systems often require adaptability—for example, to accommodate extensions, fallback behaviours, or optional features. This motivates the introduction of \emph{subtyping}: a mechanism for safely refining global types without compromising realisability.
Returning to the MQTT example, suppose we define a \emph{subtype} of the protocol where only one client is allowed to publish, see Figure~\ref{fig:mqtt-subtype}.

In our setting, we give a semantic definition of subtyping, based on the idea that a global type
$\gt$ denotes a set of executions $\executionsofcfsms{\gt}{\acommunicationmodel}$ that is parameterised by a communication model $\acommunicationmodel$. We therefore say that $\gt$ is a subtype of $\gt'$ in a communication model $\acommunicationmodel$ if 
$\executionsofcfsms{\gt}{\acommunicationmodel}\subseteq \executionsofcfsms{\gt'}{\acommunicationmodel}$. 
However, it turns out that the communication model does not impacts our notion of subtyping essentially because global types also admit
    a MSC-based semantics that is agnostic of the communication model and that defines the same semantic subtyping relation.

\subsection{Contributions}
Summing up, we make the following contributions.
\begin{itemize}
    \item Global type subtyping is a notion that is independent of the communication model (Theorem~\ref{thm:subtyping-agnostic-communication-model}).
    % \item For a fixed communication model, realisability is monotonic in subtyping (Theorem~\ref{thm:realisability-subtyping})
    \item As mentioned above, deadlock-free realisability is not monotonic in the communication model; however, if a global type is deadlock-free realisable in a communication model that contains all synchronous executions, then this global is also realisable in the synchronous model (Theorem~\ref{thm:pp-realizability-implies-synch-realizability}).
    \item It is decidable to verify whether $\gt$ is a subtype of $\gt'$ when $\gt$ and $\gt'$ are realisable (Corollary~\ref{cor:subtyping-decidability}), and more generally when $\gt'$ is complementable  together with few extra assumptions (Theorems~\ref{thm:subtyping-decidability} and \ref{thm:subtyping-decidability-2}).
    \item We show that realisability implies complementability (Theorem~\ref{thm:decidability-of-implementability-in-synch}), and that  existing frameworks in which MPST are realisable by construction hide a complementation procedure (Theorem~\ref{thm:sender-driven-choice-complementation}). 
    %We explicit a generic complementation procedure with a doubly exponential blowup in the size of the global type, and we give alternative complementation procedures with linear increase in the size of the global in a few cases, like when the global types involves at most three participants, or when all choices are sender-driven (Theorem~\ref{thm:complementation-procedures}).
    \item We show that it is decidable to veriify if  $\gt$ is realisable in a communication model $\acommunicationmodel$ assuming that a complement of $\gt$ is provided. The problem is decidable in PSPACE for the synchronous communication model, and in EXPSPACE for $\ppmodel$, $\bagmodel$, $\causalmodel$, and more generally for any causally closed, regular communication model (Theorem~\ref{thm:decidability-of-implementability-in-p2p})
\end{itemize}

%\todo[inline]{rewrite contribution}

\subsubsection*{Outline.} The paper is organised as follows: Section~\ref{sec:preliminaries} introduces the necessary background on executions, Message Sequence Charts, communication models and communicating finite state machines. 
MPST, realisability and complementability are introduced in Section~\ref{sec:global_types}. Finally Section \ref{sec:decidability}  shows  our results on the reduction of realisability to the synchronous communication model based on the notion of effectively complementable global types.  
Section \ref{sec:concl} concludes with some final remarks and discusses related works.

This paper extends and revise the work in \cite{DGLP25}
%\todo[inline]{add more here}

	\section{Preliminaries}\label{sec:preliminaries}
	% !TEX root =  ../main.tex
%!TEX spellcheck = en_GB

%In this section we introduce some notation to talk about communicating automata, executions and message sequence charts.

We assume a finite set of \emph{processes}, sometimes also called participants $\Procs=\{p,q,\ldots\}$ and a finite set of messages (labels) 
$\Msg=\{\msg,\ldots\}$.
We consider two kinds of actions:  \emph{send actions} that are  of the form 
$\send{p}{q}{\msg}$ and are executed by  process $p$  sending message  $\msg$ to  $q$;
 \emph{receive actions} that are of the form $\recv{p}{q}{\msg}$ and are executed by  $q$ receiving  $\msg$ from  $p$.
We write $\Act$ for the finite set $\Procs\times\Procs\times\{!,?\}\times\Msg$ of all actions, and $\Actp$ for the set of actions 
that can be executed by $p$ (i.e., $\send{p}{q}{\msg}$ or $\recv{q}{p}{\msg}$). 
We omit processes when they are clear from the context and simply write $!\msg$ or $?\msg$ for a send or receive action, respectively.

An \emph{event} $\event$ of a sequence of actions $w\in \Act^*$, is an index $i$ in $\{1,\ldots,\length{w}\}$;
$i$ is a send (resp. receive) event of $w$ if $w[i]$ is a send (resp. receive) action.
We write $\sendeventsof{w}$ (resp. $\receiveeventsof{w}$) for the set of send (resp. receive) events 
of $w$ and $\eventsof{w}=\sendeventsof{w}\cup\receiveeventsof{w}$ for the set of all events of $w$.
When all events are labeled with distinct actions, we identify an event with its action.

\subsection{Automata}
We assume the standard notions of words $w\in\Sigma^*$ over a finite alphabet $\Sigma$, of 
deterministic finite automata (DFA) and nondeterministic finite automata (NFA);
while dealing with NFA, we assume that an NFA may have $\epsilon$ transitions. 

A non-deterministic finite automaton (NFA) is a tuple $\A=(Q,\Sigma,\delta,q_0,F)$, where $Q$ is a finite set of states, $\Sigma$ is a finite alphabet, $\delta: Q\times(\Sigma\cup \varepsilon)\times Q$ is the transition relation, $q_0\in Q$ is the initial state, and $F\subseteq Q$ is the set of accepting states. We write $\delta^*(s,w)$ to denote the set of states $s'$ that are reachable from $s$ following a path labeled with $w$. The language accepted by $\A$, denoted $\languageofnfa{\A}$, is the set of words $w\in\Sigma^*$ such that $\delta^*(q_0,w)\cap F\neq\emptyset$..

A deterministic finite automaton (DFA) is a special case of an NFA where the transition relation $\delta$ is a partial function
$\delta: Q\times\Sigma\to Q$; it is complete if the function is
total. Every DFA $\A=(Q,\Sigma,\delta,q_0,F)$ 
accepts the same language as the complete DFA
$\A'=(Q\cup\{\bot\},\Sigma,\delta',q_0,F)$ where 
$\delta'(q,a)=\bot$ whenever $\delta(q,a)$ is undefined, and $\delta'(\bot,a)=\bot$.

To every NFA $\A=(Q,\Sigma,\delta,q_0,F)$, one can associate a DFA $\detof{\A}\eqdef(Q',\Sigma,\delta',q_0',F')$ where $Q'=2^Q$, $q_0'=\{q_0\}$, $F'$ is the set of subsets of $Q$ that contain at least one accepting state, and $\delta'$ is defined such that for all $S\in Q'$ and $a\in\Sigma$, $\delta'(S,a)=\bigcup\{\delta^*(s,a)\mid s\in S\}$.

Given two NFAs $\A_1=(Q_1,\Sigma,\delta_1,q_{0,1},F_1)$ and $\A_2=(Q_2,\Sigma,\delta_2,q_{0,2},F_2)$, we define the product NFA 
$$\A_1\otimes\A_2~\eqdef~(Q_1\times Q_2,\Sigma,\delta,(q_{0,1}, q_{0,2}),F_1\times F_2)$$ 

where the transition relation $\delta$ is defined as follows: 

\begin{itemize}
    \item for all $(s_1,s_2)\in Q_1\times Q_2$, $a\in\Sigma$, and $(s'_1,s'_2)\in Q_1\times Q_2$, we have $((s_1,s_2),a,(s'_1,s'_2))\in\delta$ if and only if $(s_1,a,s'_1)\in\delta_1$ and $(s_2,a,s'_2)\in\delta_2$, and
    \item for all $(s_1,s_2)\in Q_1\times Q_2$ and $(s'_1,s'_2)\in Q_1\times Q_2$, we have $((s_1,s_2),\varepsilon,(s'_1,s'_2))\in\delta$ if and only if $(s_1,\varepsilon,s'_1)\in\delta_1$ or $(s_2,\varepsilon,s'_2)\in\delta_2$.
\end{itemize}
It holds that $\languageofnfa{\A_1\otimes\A_2}=\languageofnfa{\A_1}\cap \languageofnfa{\A_2}$.

Finally, if $\A=(Q,\Sigma,\delta,q_0,F)$ is a DFA, we can define its dual DFA $\dualdfaof{\A}\eqdef(Q,\Sigma,\delta,q_0,Q\setminus F)$. In this case, if $\A$ is complete, we have $\languageofnfa{\dualdfaof{\A}}=\Sigma^*\setminus\languageofnfa{\A}$.

We write $\languageofnfa{\A}$ 
for the language accepted by the NFA $\A$, $\detof{\A}$ for the DFA obtained by determinising an NFA $\A$, 
and $\A_1\otimes\A_2$ for the product of two NFAs $\A_1$ and $\A_2$ such that
$\languageofnfa{\A_1\otimes\A_2}=\languageofnfa{\A_1}\cap \languageofnfa{\A_2}$. 
We write $\acceptcompletion{\mathcal A}{}$ for the NFA obtained from $\mathcal A$ by setting
all states as accepting states. 
Finally, we write $\dualdfaof{\A}$ for the dual DFA of $\A$, where
accepting states and non-accepting ones are swapped\footnote{possibly after introducing a sink state so as to make $\A$ a complete DFA}, such that 
$\languageofnfa{\dualdfaof{\A}}=\Sigma^*\setminus\languageofnfa{\A}$.

\subsection{Executions}

An execution is a sequence of actions where 
a receive action is always preceded by a unique corresponding send action, i.e., this corresponds to the intuition that no message can be received without having being sent.

\begin{defi}[Execution]
	An \emph{execution} over $\Procs$ and $\Msg$ is a pair $e=(w,\source)$ where $w\in\Act^*$  
	and $\source:\receiveeventsof{w}\to\sendeventsof{w}$ is an injective function from receive events to send events such that for all receive event $r$ labeled with $\recv{p}{q}{\msg}$, $\source(r)$ is labeled with $\send{p}{q}{\msg}$ and
	$\source(r)<r$.
\end{defi}

%For two events $i,j\in\eventsof{e}$, we write $i\hbstrictof{e}{} j$ if $i<j$ (order on natural numbers).
When all send actions of a sequence of actions $w$ are distinct, there is at most one execution 
$\execution$ such that $w$ is the sequence of actions of $\execution$, and we often identify $w$ with $\execution$
and let $\source$ be implicit. 

An execution $e_1=(w_1,\source_1)$ is a \emph{prefix} of an execution $e_2=(w_2,\source_2)$ if $w_1$ is a prefix of $w_2$ and 
$\source_1$ is the restriction of $\source_2$ to $\receiveeventsof{w_1}$.
For a set of executions $\mathcal{E}$, we write $\prefixclosureof{\mathcal{E}}$ 
for the set of all prefixes of the executions in $\mathcal{E}$. 
We say that an execution $e_2$ is a \emph{completion} of an execution $e_1$
if $e_1$ is a prefix of $e_2$. 
% A \emph{concatenation} $e_1\cdot w$ of an execution
% $e_1$ and a sequence of actions $w$ is the execution $e_2 = e_1\cdot w_2$ where 
% $e_2$ is a completion of $e_1$ (note that $w$ is not an execution, since it may contain receive events
% which sources are in $e_1$).  
The \emph{projection} $\projofon{e}{p}$ of an execution $e$ on a process $p$ 
is the subsequence of actions in $\Actp$.  
A send event $s$ is \emph{matched} if there is a receive event $r$ such that $s=\source(r)$.
An execution $e$ is \emph{orphan-free} if $\source$ is surjective over the send events of $e$, i.e.,
all sent messages have been received.

\subsection{Communication Models}

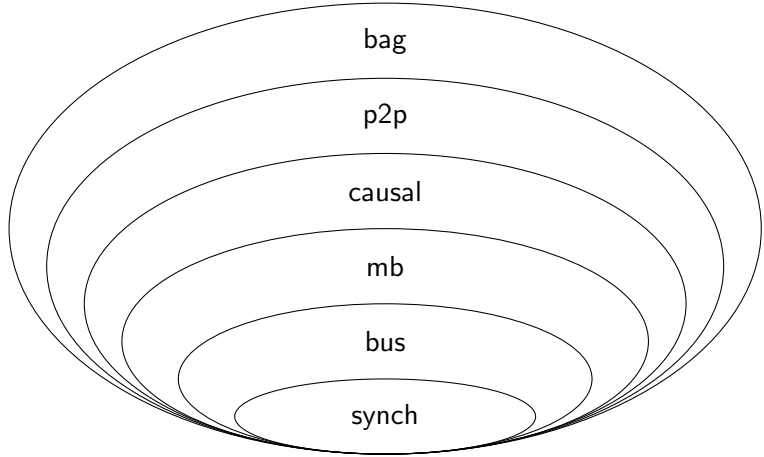
\begin{figure}
	\centering
	\begin{tikzpicture}

    \begin{scope}[xshift=0.5cm]
        \node[draw, ellipse, minimum width=10cm, minimum height=6cm, align=center] at (0,4.5) {};
        \node[draw, ellipse, minimum width=9cm, minimum height=5cm, align=center] at (0,4) {};
        \node[draw, ellipse, minimum width=8cm, minimum height=4cm, align=center] at (0,3.5) {};
        \node[draw, ellipse, minimum width=7cm, minimum height=3cm, align=center] at (0,3) {};
        \node[draw, ellipse, minimum width=5.5cm, minimum height=2cm, align=center] at (0,2.5) {};
        \node[draw, ellipse, minimum width=4cm, minimum height=1cm, align=center] at (0,2) {};
        \node at (0,7) {$\bagmodel$};
        \node at (0,6) {$\ppmodel$};
        \node at (0,5) {$\causalmodel$};
        \node at (0,4) {$\mbmodel$};
        \node at (0,3) {$\busmodel$};
        \node at (0,2) {$\synchmodel$};
        %     \node[draw, ellipse, minimum width=8cm, minimum height=4cm, align=center] at (0,3) {\begin{tabular}{c}$\causalmodel$\\[5.5cm]\end{tabular}};
    %     \node[draw, ellipse, minimum width=6cm, minimum height=3cm, align=center] at (0,2) {\begin{tabular}{c}$\mbmodel$\\[4.5cm]\end{tabular}};
    %     \node[draw, ellipse, minimum width=4cm, minimum height=2cm, align=center] at (0,1) {\begin{tabular}{c}$\ppmodel$\\[3.5cm]\end{tabular}};
    %     \node[draw, ellipse, minimum width=2cm, minimum height=1cm, align=center] at (0,0) {\begin{tabular}{c}$\synchmodel$\\[2.5cm]\end{tabular}};
    \end{scope}

    \begin{scope}[xshift=-0.5cm]
    \end{scope}

\end{tikzpicture}
	\caption{Hierarchy of communication models, both as sets of executions and as sets of MSCs.}\label{fig:communication-models-hierarchy}
\end{figure}

A communication model defines which executions conform to some rules. They are sometimes defined operationally, based on implementation of the transport layer (like queues). In this paper, we adopt a more denotational approach, and consider that a communication model is fully determined by the set of executions it allows. 

\begin{defi}[Communication model]
	\label{def:communication-model}
	A \emph{communication model} $\acommunicationmodel$ is a set 
	$\executionsofmodel{\acommunicationmodel}$ of executions.
\end{defi}

% A communication model is \emph{prefix-closed} if 
% $\executionsofmodel{\acommunicationmodel}$ is prefix-closed, i.e., for all
% $e,e'\in\executionsofmodel{\acommunicationmodel}$, if $e\prefixorder e'$, then $e\in\executionsofmodel{\acommunicationmodel}$.

The simplest communication model is the \emph{bag communication model} ($\bagmodel$),
where all executions are allowed, i.e., any sent message can overtake another one, and no assumption holds on the order of message delivery. It represents the loosest degree of asynchrony.

\begin{defi}[$\bagmodel$]
	\label{def:bag-communication-model}
	$\executionsofmodel{\bagmodel}$ is the set of all executions.
\end{defi}

On the other side of the spectrum, in the synchronous communication model ($\synchmodel$), message exchanges can be thought as rendez-vous synchronisations.
In other words, an execution $e$ belongs to $\executionsofmodel{\synchmodel}$ if
all send events are immediately followed by their corresponding receive events.

\begin{defi}[$\synchmodel$]
	An execution $\execution=(w,\source) \in \executionsofmodel{\synchmodel}$ if
	for all send event $s\in\sendeventsof{e}$, $s+1$ is a receive event of $e$ and $\source(s+1)=s$.
\end{defi}

Between these two extremes, we can find a number of communication models, with increased levels of synchrony. In the peer-to-peer communication model ($\ppmodel$), messages sent by a process $p$ to $q$ transit over a FIFO channel
that is specific to the pair $(p,q)$: if $p$ sends first $m_1$ then $m_2$ to $q$, 
then $m_2$ cannot overtake $m_1$ in the FIFO channel. In particular:
\begin{itemize}
	\item if $m_1$ is not received, then $m_2$ is not received either;
	\item if both are received, then $m_1$ is received before $m_2$.
\end{itemize}
      
\begin{defi}[$\ppmodel$]\label{def:queue-based-communication-model}
	$\executionsofmodel{\ppmodel}$  is
	the set of executions $e$ such that for any two send events $s_1=\send{p}{q}{\msg_1}$ and 
	$s_2=\send{p}{q}{\msg_2}$ in $\sendeventsof{e}$, 
	with $s_1 < s_2$,
	one of the two holds:
	\begin{itemize}
		\item $s_2$ is unmatched, or
		\item there exists $r_1,r_2$ such that $r_1<r_2$, $\source(r_1)=s_1$, and $\source(r_2)=s_2$.
	\end{itemize}
\end{defi}

The causal communication model ($\causalmodel$) was introduced by Lamport~\cite{lamport1978time}; unlike other communication models, it is easier to define it denotationally rather than operationally, as the possible implementations of this communication model are rather subtle. Intuitively, an execution is causal (or causally ordered) if a message sent to a process $p$ cannot overtake another message sent to the same process $p$ \emph{causally before}, even if the senders are different. We postpone the formal definition of this model as it requires some further notions (see Definition~\ref{def:causal-model} later).

% Note that, $\executionsofmodel{\synchmodel}\subset\executionsofmodel{\ppmodel}$.
% Moreover, if $e$ is an execution in $\ppmodel$, then  $\source_e$ is defined as follows:
% the source of the $i$-th receive event of $q$ from $p$ is the $i$-th send event of $p$ to $q$.
% If $e$ is an execution in $\synchmodel$, then $\source_e$ is defined as follows: 
% for all receive event $i$, $\source_{e}(i)=i-1$.

The variety of communication models is large, and we do not aim at being exhaustive here. We can mention the other FIFO-based communication models
like the \emph{bus communication model} (a single FIFO queue merging all messages), the \emph{mailbox communication model} (a FIFO queue per process merging all incoming messages), bounded variants of these models (the FIFO queues are often bounded),communication models with errors (message loss, message duplication, message corruption), etc.
Figure~\ref{fig:communication-models-hierarchy} shows a simplified hierarchy of some of these communication models (see also~\cite{ENGELS2002253,DBLP:conf/fm/ChevrouH0Q19,DBLP:journals/pacmpl/GiustoFLL23}
for more refined hierarchies). Figure~\ref{fig:exmscs} exemplifies this hierarchy.
Every example from~\ref{fig:bag_ex} to~\ref{fig:bus_ex} 
is in the corresponding communication model, but is also a counter-example for the next 
one in the hierarchy.

\subsection{Message Sequence Charts}
While executions correspond to a total order view of the events occurring in a 
system, message sequence charts (MSC) adopt a distributed, a partial order view on the events.
For a tuple $\msc=(w_p)_{p\in\Procs}$, each $w_p\in\Actsp$ 
is a sequence of actions representing  the ones executed by process $p$ according to some total, locally observable, order.
We write $\eventsof{\msc}$  for the set $\{(p,i) \mid p \in \Procs \text{ and } 0 \leq i < \length{w_p}\}$.
The label $\actionof{\event}$ of the event $\event=(p,i)$ is the action $w_p[i]$. The event $\event$
is a send (resp. receive) event if it is labeled with a send (resp. receive) action.
We write $\sendeventsof{\msc}$ (resp. $\receiveeventsof{\msc}$) for the set of send (resp. receive) events of $\msc$; we also write $\messageof{\event}$ for the message sent or received at event $\event$, and 
$\processof{\event}$ for the process executing $\event$. Finally, we write 
$\verticalorderstrict{\event_1}{\event_2}$ if there is a process $p$ and $i<j$ such that
$\event_1=(p,i)$ and $\event_2=(p,j)$.
%\pascal{Notations $\sendevents{\msc}$,  $\messageof{e}$, $\receiveevents{\msc}$ and $\processof{e}$ are used once in the paper. Should we keep them?}

\begin{defi}[MSC]\label{def:msc}
	An {MSC} over $\Procs$ and $\Msg$ is a tuple $\msc = \big((w_p)_{p\in\Procs},\source\big)$
    where
    \begin{enumerate}
        \item for each process $p$, $w_p\in\Actsp$ is a finite sequence of actions;
        \item $\source: \receiveeventsof{\msc} \to \sendeventsof{\msc}$ is 
            an injective function from receive events to send events such that
			for all receive event $\event$ labeled with $\recv{p}{q}{\msg}$,
            $\source(\event)$ is labeled with $\send{p}{q}{\msg}$.
    \end{enumerate}
\end{defi}

\begin{figure}[t]
			\captionsetup[subfigure]{justification=centering}
	% \centering
	\begin{subfigure}[t]{0.24\textwidth}\centering

		\begin{tikzpicture}[scale=0.7, every node/.style={transform shape}]
			\newproc{0}{p}{-2.2};
			\newproc{2}{q}{-2.2};

			\newmsgm{0}{2}{-1.7}{-0.5}{1}{0.25}{black};
			\newmsgm{2}{0}{-1.7}{-0.5}{2}{0.25}{black};

			\end{tikzpicture}
		\caption{not an MSC}	\label{fig:raw_msc_ex}

		\end{subfigure}
        \begin{subfigure}[t]{0.24\textwidth}\centering

            \begin{tikzpicture}[scale=0.7, every node/.style={transform shape}]
                \newproc{0}{p}{-2.2};
                \newproc{2}{q}{-2.2};
    
                \newmsgm{0}{2}{-0.5}{-1.5}{1}{0.1}{black};
                \newmsgm{0}{2}{-1}{-1}{2}{0.8}{black};
    
                \end{tikzpicture}
            \caption{$\bagmodel$, and $\erlangmodel$ if $m_1\neq m_2$}	\label{fig:bag_ex}
    
            \end{subfigure}
        % \centering
	\begin{subfigure}[t]{0.24\textwidth}\centering
		\begin{tikzpicture}[scale=0.7, every node/.style={transform shape}]
			\newproc{0}{p}{-2.2};
			\newproc{1}{q}{-2.2};
			\newproc{2}{r}{-2.2};

			\newmsgm{0}{1}{-0.3}{-2}{1}{0.15}{black};
			\newmsgm{0}{2}{-0.9}{-0.9}{2}{0.75}{black};
			\newmsgm{2}{1}{-1.5}{-1.5}{3}{0.5}{black};
			% \newmsgm{2}{1}{-2}{-2}{4}{0.5}{black};

			% \newflechevert{Purple}{0}{-0.3}{-0.9};
			% \newflechehor{Purple}{-0.9}{0}{2};
			% \newflechevert{Purple}{2}{-0.9}{-1.5};
		\end{tikzpicture}
		\caption{$\ppmodel$} \label{fig:pp_ex}
	\end{subfigure}
	\begin{subfigure}[t]{0.24\textwidth}\centering
		\begin{tikzpicture}[scale=0.7, every node/.style={transform shape}]
			\newproc{0}{p}{-2.2};
			\newproc{1}{q}{-2.2};
			\newproc{2}{r}{-2.2};
			\newproc{3}{s}{-2.2};
	
			\newmsgm{0}{3}{-0.6}{-0.6}{2}{0.8}{black};
			\newmsgm{2}{3}{-1.1}{-1.1}{3}{0.5}{black};
			\newmsgm{2}{1}{-1.5}{-1.5}{4}{0.5}{black};
			\newmsgm{0}{1}{-0.3}{-2}{1}{0.5}{black};
		\end{tikzpicture}
		\caption{$\causalmodel$} \label{fig:causal_ex}
	\end{subfigure}
	% \hfill
\begin{subfigure}[t]{0.24\textwidth}\centering
	\begin{tikzpicture}[scale=0.7, every node/.style={transform shape}]
		\newproc{0}{p}{-2.2};
		\newproc{1}{q}{-2.2};
		\newproc{2}{r}{-2.2};
		\newproc{3}{s}{-2.2};

		\newmsgm{3}{0}{-0.6}{-0.6}{2}{0.2}{black};
		\newmsgm{3}{2}{-1.1}{-1.1}{3}{0.5}{black};
		\newmsgm{2}{1}{-1.5}{-1.5}{4}{0.5}{black};
		\newmsgm{0}{1}{-0.3}{-2}{1}{0.5}{black};
	\end{tikzpicture}
			\caption{$\mbmodel$} \label{fig:mb_ex}
\end{subfigure}
\begin{subfigure}[t]{0.24\textwidth}\centering
	\begin{tikzpicture}[scale=0.7, every node/.style={transform shape}]
		\newproc{0}{p}{-2.2};
		\newproc{1}{q}{-2.2};
		\newproc{2}{r}{-2.2};

		\newmsgm{0}{1}{-0.3}{-1.5}{1}{0.15}{black};
		\newmsgm{1}{0}{-0.9}{-0.9}{2}{0.25}{black};
		\newmsgm{1}{2}{-2}{-2}{3}{0.5}{black};
		% \newmsgm{2}{1}{-2}{-2}{4}{0.5}{black};
	\end{tikzpicture}
		\caption{$\busmodel$} \label{fig:bus_ex}
\end{subfigure}
\begin{subfigure}[t]{0.24\textwidth}\centering
	\begin{center}
		\begin{tikzpicture}[scale=0.7, every node/.style={transform shape}]
			\newproc{0}{p}{-2.2};
			\newproc{1}{q}{-2.2};
			\newproc{2}{r}{-2.2};

			\newmsgm{0}{1}{-0.5}{-0.5}{1}{0.5}{black};
			\newmsgm{1}{0}{-1}{-1}{2}{0.5}{black};
			\newmsgm{1}{2}{-1.6}{-1.6}{3}{0.5}{black};

		\end{tikzpicture}
		\caption{$\synchmodel$}
		\label{fig:rsc_ex}
	\end{center}
\end{subfigure}
\begin{subfigure}[t]{0.24\textwidth}\centering

	\begin{tikzpicture}[scale=0.7, every node/.style={transform shape}]
		\newproc{0}{p}{-2.2};
		\newproc{2}{q}{-2.2};

		\newmsgum{0}{2}{-0.8}{1}{0.2}{black};
		\newmsgm{0}{2}{-1.6}{-1.6}{2}{0.2}{black};

		\end{tikzpicture}
	\caption{an MSC with an orphan message}	\label{fig:orphan_ex}
\end{subfigure}
		\caption{Examples of MSCs for various communication models. The sender of a message is at the origin of the arrow and the receiver at the destination. Unmatched send events are depicted with dashed arrows.}\label{fig:exmscs}
\end{figure}
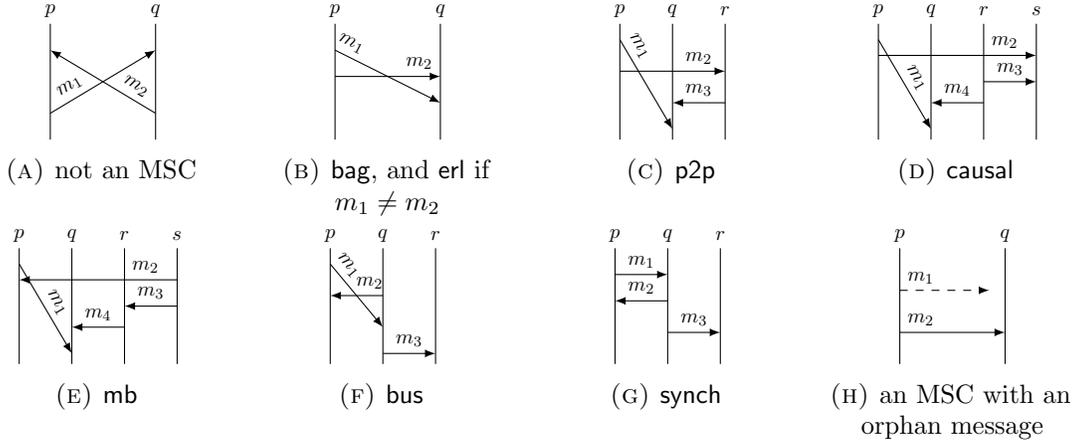

For an execution $\execution$,  $\mscof{\execution}$ is the MSC 
$\big((w_p)_{p\in\Procs},\source\big)$ where $w_p$ is the subsequence of $\execution$ restricted to the actions of $p$,
and $\source$ is the lifting of $\source_{\execution}$ to the events of $(w_p)_{p\in\Procs}$.

\begin{exa}
    \label{ex:msc}
     MSC $\msc$ in Figure~\ref{fig:raw_msc_ex} is an MSC over $\Procs=\{p,q\}$ and $\Msg=\{\msg_1,\msg_2\}$
    with $\msc=\large((w_p,w_q),\source\large)$, $w_p = ?\msg_2!\msg_1$, $w_q = ?\msg_1!\msg_2$, $\source((p,0)) = (q,1)$,
    and $\source((q,0)) = (p,1)$. Note that there is no execution $\execution$ such that 
	$\msc=\mscof{\execution}$ as all receptions precede the corresponding sends. On the other hand, the MSC
	in Figure~\ref{fig:pp_ex} is $\mscof{\execution}$ for the execution $\execution$
	$=\large(!\msg_1!\msg_2?\msg_2!\msg_3?\msg_3?\msg_1,\source\large)$ where 
	$\source(3) = 2$, $\source(5) = 2$ and $\source(6) = 4$. It is the only execution that induces this MSC, but 
	in general there might exist several executions inducing the same MSC.
\end{exa}

For a set of processes $\Procs$, 
an MSC $M=\big((w_p)_{p\in\Procs},\source\big)$ is the \emph{prefix} 
of another MSC $M'=\big((w'_p)_{p\in\Procs},\source'\big)$, in short $M \prefixorder M'$,
if for all $p\in\Procs$, $w_p$ is a prefix of $w'_p$ and $\source'(e)=\source(e)$ for all receive events $e$ of $M$.
The \emph{concatenation} of MSCs $M_1$ and $M_2$ is the MSC  $M_1\cdot M_2$ 
obtained by gluing vertically  $M_1$ before $M_2$:
formally, if $M_1=((w_p^1)_{p\in\Procs},\source_1)$ and $M_2=((w_p^2)_{p\in\Procs},\source_2)$, 
then $M_1\cdot M_2=((w_p)_{p\in\Procs},\source)$ where
%\begin{inparaenum}[(i)]
 %   \item 
    for all $p$, $w_p=w_p^1\cdot w_p^2$, and
 %   \item 
    $\source$ is defined by $\source(e) = \source_i(e)$ for all receive events $e$ of $M_i$ ($i=1,2$).
%\end{inparaenum}

\subsection{Happens-before relation and linearisations}
%\etienne{TODO cette section et ailleurs: relire en faisant attention à la convention de notation entre ordres partiels $\prec$ et totaux $<$}
In a given MSC $M$,
an event $\event$ happens before  
 $\event'$, if 
(i) $\event$ and $\event'$ are events
of a same process $p$ and happen in that order on 
the time line of $p$,
or (ii) $\event$ is send event matched by $\event'$,
or (iii) a sequence of such situations defines
a path from $\event$ to $\event'$.

\begin{defi}[Happens-before relation]
Let $M$ be an MSC. The happens-before relation over $M$
is the binary relation $\happensbeforestrict$ defined as 
the least transitive relation over $\eventsof{\msc}$ such that:
\begin{itemize}
   \item 
   for all 
   $p,i,j$, if $i<j$, then $(p,i)\happensbeforestrict (p,j)$, and
   \item 
   for all receive events $\event$, $\source(\event) \happensbeforestrict \event$.
\end{itemize}
\end{defi}

The link between executions and MSCs goes through the notion of linearisation: 
\begin{defi}[Linearisation]
	\label{def:linearisation}
	A \emph{linearisation} of an MSC $\msc$ is a
total order $\alinearisation$ on $\eventsof{\msc}$
	that refines $\happensbeforestrict$:  for all events $\event,\event'$, 
	if $\event\happensbeforestrict \event'$, then $\event\alinearisation \event'$. 
\end{defi}
We write $\linearisationsof{\msc}{}$ for the set of all linearisations of $\msc$.
We often identify a linearisation
with the execution it induces. 

\begin{exa}
	\label{ex:linearisation}
	Take the MSC $\msc$  in Figure~\ref{fig:rsc_ex}. 
	These messages are ordered following the happens before relation: $!m_1 \happensbeforestrict ?m_1\happensbeforestrict!m_3$, and both $!m_3$ and $!m_2$
	happen before $?m_3$.
	Moreover, $\happensbeforestrict$ is a partial order, 
	and 
	$!m_1!m_2?m_1!m_3?m_3?m_2 \in \linearisationsof{\msc_c}{}$.
	On the other hand, consider the MSC $M$ in 
	in Figure~\ref{fig:raw_msc_ex}; then
	$\happensbeforestrictof{\msc'}$ is not a partial order, because
	$?m_2\happensbeforestrictof{\msc'}!m_1\happensbeforestrictof{\msc'}?m_1\happensbeforestrictof{\msc'}!m_2\happensbeforestrictof{\msc'}?m_2$,
	therefore $\linearisationsof{\msc'}{}=\emptyset$.
\end{exa}

We can now formally define the causal communication model.

\begin{defi}[$\causalmodel$]\label{def:causal-model}
	$\executionsofmodel{\causalmodel}$ is the set
	of executions $e$, with $M=\mscof{e}$, such that for any two send events $s_1=\send{p}{q}{\msg_1}$ and $s_2=\send{p'}{q'}{\msg_2}$ in $\sendeventsof{e}$, if
	$s_1 \porderstrictof{M} s_2$ and $q=q'$, then one of the two holds:
	\begin{itemize}
		\item $s_2$ is unmatched, or
		\item there exist $r_1,r_2$ such that $r_1 \porderstrictof{M} r_2$, $\source(r_1)=s_1$, and $\source(r_2)=s_2$.	
	\end{itemize}
\end{defi}

Given an MSC $\msc$, we write 
$\linearisationsof{\msc}{\acommunicationmodel}$
to denote 
$\linearisationsof{\msc}{}\cap\executionsofmodel{\acommunicationmodel}$; 
the executions of $\linearisationsof{\msc}{\acommunicationmodel}$
are called the linearisations of $\msc$ 
in the communication model $\acommunicationmodel$. 	

\begin{defi}[$\acommunicationmodel$-linearisable MSC]
	\label{def:linearisable-msc}
	An MSC $\msc$ is \emph{linearisable} in a communication model $\acommunicationmodel$ if 
	$\linearisationsof{\msc}{\acommunicationmodel}\neq\emptyset$.
	We write $\mscsetofmodel{\acommunicationmodel}$ for the set of all MSCs linearisable in $\acommunicationmodel$.
\end{defi}

\begin{exa}
	\label{ex:msc-linearisable}
	The MSC $M_b$ in Figure~\ref{fig:pp_ex} is linearisable in $\ppmodel$.
	It admits only one linearisation, and $\linearisationsof{M_b}{\ppmodel}=
%	$ 
%	$$
		\{!m_1!m_2?m_2!m_3?m_3?m_1\}
		$.
%	$$
	However, $M_b$ is not linearisable in $\synchmodel$.
	Finally, $M_c$ in Figure~\ref{fig:rsc_ex} is linearisable in $\synchmodel$ with
	$$\linearisationsof{M_c}{\synchmodel}=
	\{~!m_1?m_1!m_2?m_2!m_3?m_3~,~
	   !m_2?m_2!m_1?m_1!m_3?m_3~\}.
	$$
\end{exa}

Finally, we introduce a property that will be helpful in the next paragraph 
for
giving an alternative characterisation of deadlock-freedom of a system of communicating finite state machines.
This property states that  for each execution in the communication model, if we consider the corresponding MSC then all the linearisations of such an MSC belong to the same communication model.

\begin{defi}[Causally-closed communication model]\label{def:causally-closed-communication-model}
	A communication model $\acommunicationmodel$ is \emph{causally-closed} if for all MSCs $M$,
	$\linearisationsof{\msc}{\acommunicationmodel}\neq\emptyset$ implies that
	$\linearisationsof{\msc}{\acommunicationmodel}=\linearisationsof{\msc}{}$.
\end{defi}

Observe that not all communication models are causally closed. It is the case for $\ppmodel$,  but it is immediate  to see that the property is not valid  for $\synchmodel$. Take for instance  MSC $M$ in Figure~\ref{fig:rsc_ex},  its linearisation 
$!m_1!m_2?m_1?m_2!m_3?m_3$ 
  does not belong to $\linearisationsof{M_c}{\synchmodel}$.  
  
% !TEX root =  ../main.tex
%!TEX spellcheck = en_GB

\begin{lem}\label{lem:reformulation-of-ppmodel-def}
    Let $e$ be an execution and  $M=\mscof{e}$.
    Then $e$ is an execution of $\ppmodel$ if and only if
    for any two send events $s_1=\send{p}{q}{\msg_1}$ and $s_2=\send{p}{q}{\msg_2}$ in $\sendeventsof{e}$, with $s_1 \porderof{M} s_2$,
        one of the two holds:
        \begin{itemize}
            \item $s_2$ is unmatched, or
            \item there exists $r_1,r_2$ such that $r_1 \porderof{M} r_2$, $\source(r_1)=s_1$, and $\source(r_2)=s_2$.
        \end{itemize}
\end{lem}
\begin{proof}
    Assume that $e$ is a $\ppmodel$ execution.
    Let $s_1=\send{p}{q}{\msg_1}$ and $s_2=\send{p}{q}{\msg_2}$ be two send events in $\sendeventsof{e}$ such that $s_1 \porderof{M} s_2$.
    Then $s_1<s_2$ in $e$, because $<$ is a linearisation of $\porderof{M}$.
    By  definition of $\ppmodel$, $s_2$ is unmatched or there exists $r_1,r_2$ such that $r_1 < r_2$, $\source(r_1)=s_1$, and $\source(r_2)=s_2$.
    In the second case, $r_1 < r_2$ implies that
    $r_1 \porderof{M} r_2$, because both $r_1$ and $r_2$ occur on the same process $q$.
    Conversely, assume that $e$ and $M$ satisfy the above reformulation of the definition of $\ppmodel$. 
    Let $s_1=\send{p}{q}{\msg_1}$ and $s_2=\send{p}{q}{\msg_2}$ be two send events in $\sendeventsof{e}$ such that $s_1 < s_2$.
    Then $s_1 \porderof{M} s_2$
    because $s_1$ and $s_2$ occur on the same process $p$.
    By the reformulation of the definition of $\ppmodel$, $s_2$ is unmatched or there exists $r_1,r_2$ such that $r_1 \porderof{M} r_2$, $\source(r_1)=s_1$, and $\source(r_2)=s_2$.
    In the second case, $r_1 \porderof{M} r_2$
    implies that $r_1 < r_2$ because $<$ is a  linearisation of $\porderof{M}$.
\end{proof}

  The following lemma gives examples of causally-closed communication models.
\begin{lem}\label{lem:pp-is-causally-closed}
	$\executionsofmodel{\ppmodel}$, $\executionsofmodel{\causalmodel}$, and $\executionsofmodel{\bagmodel}$ are causally-closed.
\end{lem}
\begin{proof}
    It is obvious for $\bagmodel$ because $\linearisationsof{M}{}=\linearisationsof{M}{\bagmodel}$. For $\causalmodel$, it is enough to observe that the definition of $\causalmodel$ (Definition~\ref{def:causal-model}) only depends on the partial order $\porderof{M}$ induced by the MSC $M=\mscof{e}$.
    Let $M\in\mscsetofmodel{\ppmodel}$ 
    and  $e$ be a linearisation of $M$.
    We show that $e\in\executionsofmodel{\ppmodel}$.
    By definition of $\mscsetofmodel{\ppmodel}$, 
    there is an execution 
    $e'\in\executionsofmodel{\ppmodel}$
    such that $\mscof{e'}=M$.
    By Lemma~\ref{lem:reformulation-of-ppmodel-def},
    $e$ is also a $\ppmodel$ execution.
\end{proof}

\subsection{Communicating finite state machines}
% We assume standard notations for automata, words and languages. As usual, a non-deterministic finite state automaton (NFA) is a tuple
% $\mathcal A = (Q,\Sigma, \delta, l_0, F)$ where $Q$ is a set of control states, 
% $\Sigma$ is an alphabet, $\delta:Q\times\Sigma\to 2^Q$ 
% is the transition relation,
% $l_0$ is the initial control state, and $F\subseteq Q$ is the set of accepting states.
% The language $\languageofnfa{\mathcal A}$ of an NFA $\mathcal A$ and the notion of deterministic
% finite state automaton (DFA) or $\varepsilon$ transitions are defined as usual.
% We write $\acceptcompletion{\mathcal A}{}$ for the automaton obtained from $\mathcal A$ by setting
% $F=Q$. 
Communicating finite state machines~\cite{BrandZafiropulo} are non-deterministic finite state automata (NFA), possibly with $\varepsilon$-transitions, 
where the alphabet is the set of actions that can be executed by the processes.

\begin{defi}[CFSM] 
    A communicating finite state machine (CFSM) is an NFA with $\varepsilon$-transitions $\acfsm$ over the alphabet $\Act$.
    A system of CFSMs is a tuple $\cfsms = (\acfsm_p)_{p\in\Procs}$.
\end{defi}

\begin{exa}
	\label{ex:cfsm}
	Figure~\ref{fig:cfsm_ex} depicts a system $(\acfsm_p,\acfsm_q,\acfsm_r)$ with three processes $p,q,r$. 
\end{exa}

\begin{figure}
	\centering
	\begin{tikzpicture}
    % \begin{scope}
    %     \node[state, initial, initial text={}, accepting] (q0) at (0,0) {$q_0$};
    %     \node[state] (q1) at (2,0) {$q_1$};
    %     \node[state] (q2) at (2,2) {$q_2$};
    %     \draw[->, bend left=20] (q0) edge node[above] {$\gtlabel{p}{q}{a}$} (q1);
    %     \draw[->] (q1) edge node[right] {$\gtlabel{q}{r}{b}$} (q2);
    %     \draw[->, bend left=20] (q1) edge node[below] {$\gtlabel{q}{p}{c}$} (q0);
    %     \draw[->, bend right=20] (q2) edge node[above,sloped] {$\gtlabel{r}{p}{d}$} (q0);
    % \end{scope}
    \begin{scope}[xshift=-4.5cm]
        \node[state, initial, initial text={}, accepting] (q0) at (0,0) {$q_0$};
        \node[state] (q1) at (2,0) {$q_1$};
        \node[state] (q2) at (2,2) {$q_2$};
        \draw[->, bend left=20] (q0) edge node[above] {$\send{p}{q}{a}$} (q1);
        \draw[->] (q1) edge node[right] {$\varepsilon$} (q2);
        \draw[->, bend left=20] (q1) edge node[below] {$\recv{q}{p}{c}$} (q0);
        \draw[->, bend right=20] (q2) edge node[above,sloped] {$\recv{r}{p}{d}$} (q0);
    \end{scope}
    \begin{scope}
        \node[state, initial, initial text={}, accepting] (q0) at (0,0) {$q_0$};
        \node[state] (q1) at (2,0) {$q_1$};
        \node[state] (q2) at (2,2) {$q_2$};
        \draw[->, bend left=20] (q0) edge node[above] {$\recv{p}{q}{a}$} (q1);
        \draw[->] (q1) edge node[right] {$\send{q}{r}{b}$} (q2);
        \draw[->, bend left=20] (q1) edge node[below] {$\send{q}{p}{c}$} (q0);
        \draw[->, bend right=20] (q2) edge node[above,sloped] {$\varepsilon$} (q0);
    \end{scope}
    \begin{scope}[xshift=4.5cm]
        \node[state, initial, initial text={}, accepting] (q0) at (0,0) {$q_0$};
        \node[state] (q1) at (2,0) {$q_1$};
        \node[state] (q2) at (2,2) {$q_2$};
        \draw[->, bend left=20] (q0) edge node[above] {$\varepsilon$} (q1);
        \draw[->] (q1) edge node[right] {$\recv{q}{r}{b}$} (q2);
        \draw[->, bend left=20] (q1) edge node[below] {$\varepsilon$} (q0);
        \draw[->, bend right=20] (q2) edge node[above,sloped] {$\send{r}{p}{d}$} (q0);
    \end{scope}
\end{tikzpicture}
	\caption{A system of CFSMs with three processes.}\label{fig:cfsm_ex}
\end{figure}
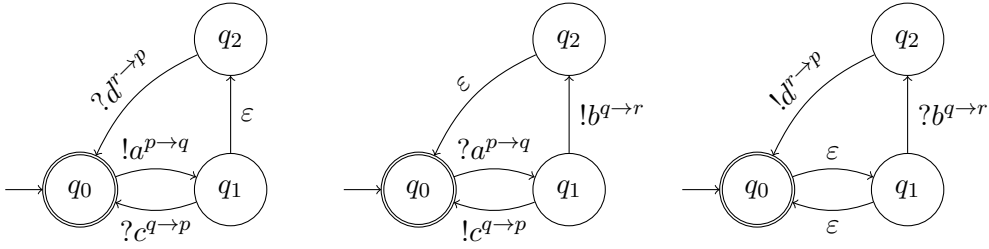

\begin{figure}
  \begin{center}
	\begin{tikzpicture}[>=stealth, node distance=2cm, auto]
    % État initial
    % \begin{scope}[xshift=-5cm]
    %     \node[state, accepting, initial, initial text={}] (s0) {$s_0$};

    %     % États pour mlication c1 -> c2
    %     \node[state] (s1) [above right=2cm and 1cm of s0] {$s_1$};
    %     \node[state] (s2) [right=2cm of s0] {$s_2$};

    %     % Boucle de mlication de c1
    %     \draw[->, bend left=40] (s0) to node[above,sloped]{$\gtlabel{c1}{b}{m}$}(s1);
    %     \draw[->, bend right=20] (s1) -- node[above,sloped] {$\gtlabel{b}{c2}{m}$} (s0);

    %     % Boucle de mlication de c2
    %     \draw[->, bend left=20] (s0) to node[above,sloped]{$\gtlabel{c2}{b}{m}$}(s2);
    %     \draw[->, bend right=20] (s2) -- node[below,sloped]{$\gtlabel{b}{c1}{m}$} (s0);
    % \end{scope}
    \begin{scope}[xshift=-4.5cm] %c1
        \node[state, accepting, initial, initial text={}] (s0) {$s_0$};

        \node[state] (s1) [above right=2cm and 1cm of s0] {$s_1$};
        \node[state] (s2) [right=2cm of s0] {$s_2$};

        \draw[->, bend left=40] (s0) to node[above,sloped]{$\send{c1}{b}{m}$}(s1);
        \draw[->, bend right=20] (s1) -- node[above,sloped] {$\varepsilon$} (s0);

        \draw[->, bend left=20] (s0) to node[above,sloped]{$\varepsilon$}(s2);
        \draw[->, bend right=20] (s2) -- node[below,sloped]{$\recv{b}{c1}{m}$} (s0);
    \end{scope}
    \begin{scope}
        \node[state, accepting, initial, initial text={}] (s0) {$s_0$};

        \node[state] (s1) [above right=2cm and 1cm of s0] {$s_1$};
        \node[state] (s2) [right=2cm of s0] {$s_2$};

        \draw[->, bend left=40] (s0) to node[above,sloped]{$\recv{c1}{b}{m}$}(s1);
        \draw[->, bend right=20] (s1) -- node[above,sloped] {$\send{b}{c2}{m}$} (s0);

        \draw[->, bend left=20] (s0) to node[above,sloped]{$\recv{c2}{b}{m}$}(s2);
        \draw[->, bend right=20] (s2) -- node[below,sloped]{$\send{b}{c1}{m}$} (s0);
    \end{scope}
    \begin{scope}[xshift=4.5cm]
        \node[state, accepting, initial, initial text={}] (s0) {$s_0$};

        \node[state] (s1) [above right=2cm and 1cm of s0] {$s_1$};
        \node[state] (s2) [right=2cm of s0] {$s_2$};

        \draw[->, bend left=40] (s0) to node[above,sloped]{$\varepsilon$}(s1);
        \draw[->, bend right=20] (s1) -- node[above,sloped] {$\recv{b}{c2}{m}$} (s0);

        \draw[->, bend left=20] (s0) to node[above,sloped]{$\send{c2}{b}{m}$}(s2);
        \draw[->, bend right=20] (s2) -- node[below,sloped]{$\varepsilon$} (s0);
    \end{scope}
\end{tikzpicture}
    \caption{A system of CFSMs with orphan messages.}\label{fig:cfsm_orphans}
     \end{center}
\end{figure}
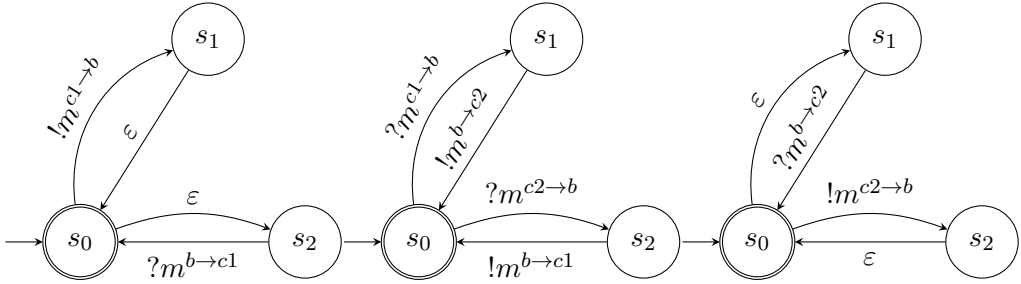

\begin{exa}
	\label{ex:cfsm_orphan}
	Figure~\ref{fig:cfsm_orphans} depicts a system $(\acfsm_{c1},\acfsm_q,\acfsm_{c2})$ with three processes $c1,q,c2$ that models a naive message passaging system.
	In this systems the broker $b$ passes messages that it receives from $c1$ to $c2$ and conversely. 
	In $\synchmodel$, the system is orphan-free and deadlock-free since the broker is always able to treat all the incoming messages, whatever their order of emission. 
	However, in  $\ppmodel$, the system is still deadlock-free, but their can be orphan messages, since all processes can be in a accepting states with a client unaware of an incoming message, e.g., after the execution
	$\send{c1}{b}{m}\recv{c1}{b}{m}\send{b}{c2}{m}$.
\end{exa}

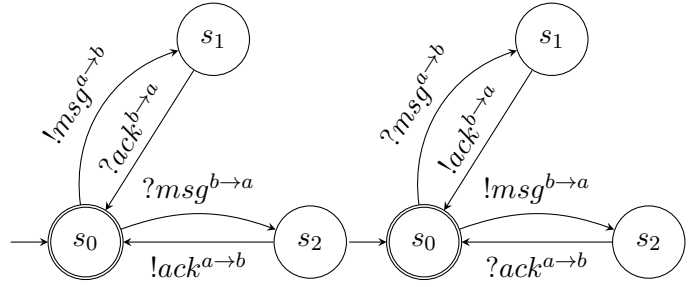
\begin{figure}
  \begin{center}
	\begin{tikzpicture}[>=stealth, node distance=2cm, auto]
    % État initial
    % \begin{scope}[xshift=-5cm]
    %     \node[state, accepting, initial, initial text={}] (s0) {$s_0$};

    %     % États pour mlication c1 -> c2
    %     \node[state] (s1) [above right=2cm and 1cm of s0] {$s_1$};
    %     \node[state] (s2) [right=2cm of s0] {$s_2$};

    %     % Boucle de mlication de c1
    %     \draw[->, bend left=40] (s0) to node[above,sloped]{$\gtlabel{c1}{b}{m}$}(s1);
    %     \draw[->, bend right=20] (s1) -- node[above,sloped] {$\gtlabel{b}{c2}{m}$} (s0);

    %     % Boucle de mlication de c2
    %     \draw[->, bend left=20] (s0) to node[above,sloped]{$\gtlabel{c2}{b}{m}$}(s2);
    %     \draw[->, bend right=20] (s2) -- node[below,sloped]{$\gtlabel{b}{c1}{m}$} (s0);
    % \end{scope}
    \begin{scope}[xshift=-4.5cm] %c1
        \node[state, accepting, initial, initial text={}] (s0) {$s_0$};

        \node[state] (s1) [above right=2cm and 1cm of s0] {$s_1$};
        \node[state] (s2) [right=2cm of s0] {$s_2$};

        \draw[->, bend left=40] (s0) to node[above,sloped]{$\send{a}{b}{msg}$}(s1);
        \draw[->, bend right=20] (s1) -- node[above,sloped] {$\recv{b}{a}{ack}$} (s0);

        \draw[->, bend left=20] (s0) to node[above,sloped]{$\recv{b}{a}{msg}$}(s2);
        \draw[->, bend right=20] (s2) -- node[below,sloped]{$\send{a}{b}{ack}$} (s0);
    \end{scope}
    \begin{scope}
        \node[state, accepting, initial, initial text={}] (s0) {$s_0$};

        \node[state] (s1) [above right=2cm and 1cm of s0] {$s_1$};
        \node[state] (s2) [right=2cm of s0] {$s_2$};

        \draw[->, bend left=40] (s0) to node[above,sloped]{$\recv{a}{b}{msg}$}(s1);
        \draw[->, bend right=20] (s1) -- node[above,sloped] {$\send{b}{a}{ack}$} (s0);

        \draw[->, bend left=20] (s0) to node[above,sloped]{$\send{b}{a}{msg}$}(s2);
        \draw[->, bend right=20] (s2) -- node[below,sloped]{$\recv{a}{b}{ack}$} (s0);
    \end{scope}
\end{tikzpicture}
    \caption{A system of CFSMs with a deadlock.}\label{fig:cfsm_deadlock}
     \end{center}
\end{figure}

\begin{exa}
	\label{ex:cfsm_deadlock}
	Figure~\ref{fig:cfsm_deadlock} depicts a system $(\acfsm_a,\acfsm_b)$ with two processes $a,b$ exchanging messages $msg$ and acknowledgements $ack$.
	In $\synchmodel$, the system is orphan-free and deadlock-free since only one process at a time can emit a message. 
	Contrary to  previous example~\ref{ex:cfsm_orphan}, the acknowledgements prevents orphans, 
	but, in $\ppmodel$, the system will face a deadlock, since both processes could send $msg$ concurrently,
	blocking them in acknowledgement waiting. 
\end{exa}

Given a system of CFSMs $\cfsms=(\acfsm_p)_{p\in\Procs}$,
we write $\acceptcompletion{\cfsms}{}$ for the system of CFSMs $\acceptcompletion{\cfsms}=(\acceptcompletion{\acfsm_p})_{p\in\Procs}$
where all states are accepting.

\begin{defi}[Executions of a CFSMs in $\acommunicationmodel$]\label{def:executions-of-cfsms}
Given a system 
$\cfsms = (\acfsm_p)_{p\in\Procs}$ of CFSMs, and a communication model $\acommunicationmodel$,
 $\executionsofcfsms{\cfsms}{\acommunicationmodel}$ is the set  of \emph{executions} $e\in\executionsofmodel{\acommunicationmodel}$
such that for all $p$, $\projofon{e}{p}$ is in $\languageofnfa{\acfsm_p}{}$.
We write $\mscsofcfsms{\cfsms}{\acommunicationmodel}$ for the set $\{\mscof{e} \mid e\in\executionsofcfsms{\cfsms}{\acommunicationmodel}\}$.
\end{defi}

\begin{exa}
	\label{ex:executions-of-cfsms}
	Let $\cfsms=(\acfsm_p,\acfsm_q,\acfsm_r)$ be the system of CFSMs in Figure~\ref{fig:cfsm_ex}.
	Then 
	$$
		\executionsofcfsms{\cfsms}{\acommunicationmodel} =
		(!a?a!c?c+!a?a!b?b!d?d)^*
	$$ 
	provided $\executionsofmodel{\acommunicationmodel}\supseteq\executionsofmodel{\synchmodel}$.
\end{exa}

\begin{rem}
	Assume $\acommunicationmodel$ is a communication model, 
	$\cfsms$  a system of CFSMs, and $e,e'\in\executionsofmodel{\acommunicationmodel}$ such that $\mscof{e}=\mscof{e'}$.
	Then $e\in\executionsof{\cfsms}{\acommunicationmodel}$ 
	if and only if $e'\in\executionsof{\cfsms}{\acommunicationmodel}$. This follows from the fact that $\projofon{e}{p}= \projofon{e'}{p}$ for all $p$.
\end{rem}

A system is orphan-free if, whenever all machines have reached an accepting state, no message
remains in transit, i.e., no message is sent but not received. All  synchronous executions are orphan-free by definition.

\begin{defi}[Orphan-free]\label{def:orphan-free}
	A system of CFSMs $\cfsms$ is \emph{orphan-free} for a communication model 
	$\acommunicationmodel$ if
	for all $e\in\executionsofcfsms{\cfsms}{\acommunicationmodel}$,
	$e$ is orphan-free.
\end{defi}

A system is deadlock-free if  
any \emph{partial} execution can be completed to an accepting execution.

\begin{defi}[Deadlock-free]\label{def:deadlock-free}
	A system of CFSMS $\cfsms$ is \emph{deadlock-free} for a communication model
	$\acommunicationmodel$ if for all 
	$e\in\executionsofcfsms{\acceptcompletion{\cfsms}}{\acommunicationmodel}$,
	there is an execution $e'\in\executionsofcfsms{\cfsms}{\acommunicationmodel}$ 
	such that $e\prefixorder e'$.
\end{defi}

\begin{rem}\label{rem:equivalent-formulation-of-deadlock-free-cfsms}
	A system of CFSMS $\cfsms$ is \emph{deadlock-free} for 
	$\acommunicationmodel$ if and only if 
	$$
	\executionsofcfsms{\acceptcompletion{\cfsms}}{\acommunicationmodel}\subseteq
	\prefixclosureof{\executionsofcfsms{\cfsms}{\acommunicationmodel}}
	$$
\end{rem}

The following result shows that, for  causally-closed communication models and $\synchmodel$,
the deadlock-freedom property of a system of CFSMs can be expressed as a property on the MSCs
of the system. %The proof can be found in Appendix \ref{app:deadlock-free-as-a-property-on-mscs-for-p2p-and-synch}.

\begin{prop}[Deadlock-freedom as an MSC property]
	\label{prop:deadlock-free-as-a-property-on-mscs-for-p2p-and-synch}
	Assume $\acommunicationmodel$ is causally-closed (respectively, $\acommunicationmodel$ is $\synchmodel$).
	Then a system of CFSMs $\cfsms$ is deadlock-free for $\acommunicationmodel$ if and only if
	$\mscsofcfsms{\acceptcompletion\cfsms}{\acommunicationmodel}\subseteq\prefixclosureof{\mscsofcfsms{\cfsms}{\acommunicationmodel}}$.
\end{prop}
% !TEX root =  ../main.tex
%!TEX spellcheck = en_GB

\begin{proof}
    We show each implication separately.
    \begin{description}
        \item[$\Rightarrow$:]
        if $\cfsms$ is deadlock-free: $\executionsof{\acceptcompletion{\cfsms}}{\acommunicationmodel}\subseteq\prefixclosureof{\executionsof{\cfsms}{\acommunicationmodel}}$, 
        then 
        $\msclanguageof{\acceptcompletion{\cfsms}}{\acommunicationmodel}\subseteq
        \prefixclosureof{\msclanguageof{\cfsms}{\acommunicationmodel}}$.
        For this implication, $\acommunicationmodel$ can be any communication model.
        Let $\msc\in\msclanguageof{\acceptcompletion{\cfsms}}{\acommunicationmodel}$, 
       we show that $\msc\in\prefixclosureof{\msclanguageof{\cfsms}{\acommunicationmodel}}$.
        By definition, there is an execution $e\in\executionsof{\acceptcompletion{\cfsms}}{\acommunicationmodel}$ such that $\msc=\mscof{e}$.
        By hypothesis, there is a completion $e' \in\executionsof{\cfsms}{\acommunicationmodel}$ of $e$.
        By definition, $\mscof{e'}\in\msclanguageof{\acceptcompletion{\cfsms}}{\acommunicationmodel}$,
        so $\msc\in\prefixclosureof{\msclanguageof{\cfsms}{\acommunicationmodel}}$.

        \item[$\Leftarrow$:] for  $\acommunicationmodel$ causally-closed:
        if  
        $\msclanguageof{\acceptcompletion{\cfsms}}{\acommunicationmodel}\subseteq
        \prefixclosureof{\msclanguageof{\cfsms}{\acommunicationmodel}}$,
        then $\cfsms$ is deadlock-free.
        Let $e\in\executionsof{\acceptcompletion{\cfsms}}{\acommunicationmodel}$ be an execution,
        we show that $e$ has a completion in $\executionsof{\cfsms}{\acommunicationmodel}$.
        By definition, $\mscof{e}\in\msclanguageof{\acceptcompletion{\cfsms}}{\acommunicationmodel}$.
        By hypothesis, there is a MSC $\msc'\in\msclanguageof{\cfsms}{\acommunicationmodel}$
        such that $\mscof{e}\prefixordermsc\msc'$. By definition of $\prefixordermsc$ on MSCs,
        there are two executions $e_1,e'$ such that 
        $e_1\prefixorder e'$, $\mscof{e'}=\msc'$, 
        and $\mscof{e_1}=\mscof{e}$.
        Let $<$ be the binary relation on $\eventsof{M'}$ 
        defined as $< \eqdef \torderstrictof{e}\cup <_1\cup <_2$, where:
        \begin{itemize}
            \item $\event <_1<\event'$ if $\event\not\in\eventsof{M}$,
            $\event'\not\in\eventsof{M}$ and $\event <_{e'} \event'$;
            \item $\event <_2\event'$ if $\event\in\eventsof{M}$ and $\event'\not\in\eventsof{M}$.
        \end{itemize}
        Note that $<_1$ and $<_2$ are transitive. We claim that $<$
        is also transitive. Indeed, 
        $$
        \begin{array}{rcl}
           \torderof{e}\cdot <_1 & = & <_1\\
            <_1\cdot <_2 & = & <_1 \\
            \torderof{e}\cdot<_2 & = & \emptyset \\
            (<_1\cup <_2)\cdot \torderof{e} & = & \emptyset \\
            (\torderof{e}\cup <_2)\cdot <_1 & = & \emptyset \\
        \end{array}
        $$
        Moreover, $<$ is irreflexive, because $\torderof{e}$ is irreflexive and $<_1$ and $<_2$ are irreflexive.
        Let $\leq$ denote the reflexive closure of $<$.
        Then this is a partial order on $\eventsof{M'}$.
        By construction, $\leq$ is actually a total order on $\eventsof{M'}$.       
        Finally, we claim that $\leq$ is a linearisation of $M'$, i.e., $\porderof{M'}\subseteq \leq$.
        Indeed, assume that $\event\porderstrictof{M'}\event'$, we show that $\event< \event'$:
        \begin{itemize}
            \item if $\event\in\eventsof{M}$ and $\event'\in\eventsof{M}$,
            then $\event\porderstrictof{M} \event'$ (as $M$ is a prefix of $M'$), hence $\event\torderof{e} \event'$ 
            (because $e$ is a linearisation of $M$),
            and finally $\event < \event'$ by definition of $<$.
            
            \item if $\event\not\in\eventsof{M}$ and $\event'\not\in\eventsof{M}$,
            then $\event\torderof{e'} \event'$ (because $e'$ is a linearisation of $M'$), 
            therefore $\event <_1 \event'$ by definition of $<_1$, and finally $\event < \event'$ by definition of $<$.

            \item if $\event\in\eventsof{M}$ and $\event'\not\in\eventsof{M'}$,
            then $\event<_2 \event'$ by definition of $<_2$, therefore $\event < \event'$ by definition of $<$.
        \end{itemize}
        Therefore, $\leq$ is a linearisation of $M'$. Let $e''$ be the execution associated to $\leq$.
        Then
        \begin{itemize}
            \item $e\prefixorder e''$ by definition of $\leq$;
            \item $\mscof{e''}=M'$ as $\leq$ is a linearisation of $M'$;
            \item $e''\in\executionsof{\cfsms}{\acommunicationmodel}$ because $e''$ is a linearisation of $M'$, $M'\in\msclanguageof{\cfsms}{\acommunicationmodel}$,
            and $\acommunicationmodel$ is causally closed by hypothesis.
        \end{itemize}
        Altogether, we have shown that $e$ has a completion in $\executionsof{\cfsms}{\acommunicationmodel}$, which ends the proof of 
        the reverse implication.

        \item[$\Leftarrow$:]   for $\acommunicationmodel=\synchmodel$:
            the proof is similar to previous case. We define the linearisation $\leq$ of $M'$ 
            in the exactly  same way, but we cannot use the fact that $\acommunicationmodel$ is causally closed.
            Instead, we use the fact that both $<_e$ and  $<_1$ are synchronous linearisations,
            and therefore $<$ is a synchronous linearisation as the concatenation of two synchronous linearisations.
    \qedhere
    \end{description}
  \end{proof}

%\iflong
%\include{proofs/prop-deadlock-free-as-a-property-on-mscs-for-p2p-and-synch}
%\else
%The proof of Proposition~\ref{prop:deadlock-free-as-a-property-on-mscs-for-p2p-and-synch} can be found in Appendix~\ref{app:deadlock-free-as-a-property-on-mscs-for-p2p-and-synch}.
%\fi

	\section{Global Types}\label{sec:global_types}
	% !TEX root =  ../main.tex
%!TEX spellcheck = en_GB

%\todo[inline]{Définir les global types syntaxiquement, et leur associer un choreography automaton.}
%
%\todo[inline]{Définir les processus (threads? sans parallelisme interne) et les automates communicants comme deux choses différentes, et introduire une notion de typage d'un processus par un local type}
%
%\todo[inline]{Expliquer le lien entre realisability et type safety.}
%
%\todo[inline]{Ajouter un commentaire sur le lien entre realisability et implementabilité, et justifier le choix de ne pas regarder les traces infinies}
%
%\todo[inline]{A propos de la complementabile, faire ref a PPDP et rappeler que realisable implique complementable et que ce n'est pas une hypothese restrictive}
%
%\todo[inline]{Faire un petit exemple en intro avec 3 machines et receiver driven}
%

We now introduce global types. Our definition deviates from the standard one as it allow for more liberal behaviours. For instance, we can have mixed choice, and when considering non-mixed choice we allow global types to send messages from different participants (meaning that our global types can be non sender-driven). This relaxation from the standard definitions, gives the possibility to implement some parallelism, see example in the Introduction. 

In our setting, global types are automata over an alphabet that represent message exchanges that we call arrows. Such automata describe a language of MSCs. 
%and explain how they can further describe a language of executions when the communication model
%is fixed. We conclude with the definition of implementability of a global type.
An \emph{arrow} is a triple $(p,q,m)\in\Procs\times\Procs\times\Messages$ with $p\neq q$;
we often write $\marrow{p}{q}{m}$ instead of $(p,q,m)$, and write $\labelalphabet$
to denote the finite set of arrows.

\begin{defi}[Global Type]
    A global type $\gt$ is a DFA over the alphabet $\labelalphabet$.
\end{defi}

\begin{rem}
    In order to make a connection with common definitions of global types, we can think of a global type as 
    a term defined by the following grammar:
    $$
    \gt ::= \mathsf{end} \mid \gtlabel{p}{q}{m};\ \gt \mid \gt+\gt \mid X \mid \mathsf{rec}~X.\gt
    $$
    with a few obvious syntactic restrictions to immediately get a \emph{deterministic} finite state automaton.
\end{rem}

\begin{figure}[t]
    \begin{tikzpicture}
    \node[state, initial, initial text={}, accepting] (q0) at (0,0) {$q_0$};
    \node[state] (q1) at (2,0) {$q_1$};
    \node[state] (q2) at (2,2) {$q_2$};
    \draw[->, bend left=20] (q0) edge node[above] {$\gtlabel{p}{q}{a}$} (q1);
    \draw[->] (q1) edge node[right] {$\gtlabel{q}{r}{b}$} (q2);
    \draw[->, bend left=20] (q1) edge node[below] {$\gtlabel{q}{p}{c}$} (q0);
    \draw[->, bend right=20] (q2) edge node[above,sloped] {$\gtlabel{r}{p}{d}$} (q0);
\end{tikzpicture}
    \caption{An example of a global type with three processes $p$, $q$, and $r$, and four messages $a$, $b$, $c$, and $d$.}\label{fig:ex-global-type}
\end{figure}
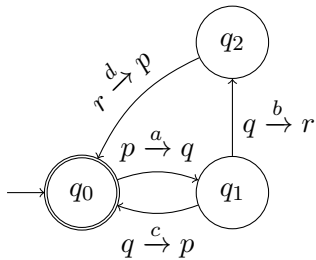

\begin{exa}\label{ex:global-type}
    The global type in Figure~\ref{fig:ex-global-type} could be defined syntactically as
    $$
    \mathsf{rec}~X.\ \mathsf{end}+\gtlabel{p}{q}{a};\ (\gtlabel{q}{p}{c};\ X+\gtlabel{q}{r}{b};\ \gtlabel{r}{p}{d};\ X).
    $$
\end{exa}

We quickly recall the definition of sender-driven choice.

\begin{defi}[Sender-driven choice]
    A global type $\gt$ has sender-driven choices
    if for all state $s$ of $\gt$, if $\gtlabel{p}{q}{m}$ and
    $\gtlabel{p'}{q'}{m'}$ label two transitions outgoing from $s$,
    then $p=p'$.
\end{defi}

\begin{exa}
    The global type in Figure~\ref{fig:ex-global-type} has sender-driven choices, as all choices are made by the sender of the message.
    However, the global type on Figure~\ref{fig:mqtt-subtype} does not have sender-driven choices, as the choice between $\gtlabel{c_1}{b}{p}$ and $\gtlabel{c_2}{b}{p}$ is made by the receiver $b$ in the initial state.
\end{exa}

%% DIRE EN QUOI C'EST PLUS GENERAL
% Notice that our definition is slightly more general then the standard one (such as that in \cite{HondaYC08}). 

%\begin{rem} We do not require global types 
%to have sender-driven choice, as in \cite{DBLP:conf/cav/LiSWZ23}.
%\end{rem}

The projection of a global type $\gt$ on a process $p$ is the CFSM $\gt_p$ obtained by 
replacing each arrow $\marrow{q}{r}{\msg}$ of a transitions
of $\gt$ by the corresponding action of $p$ (either $\send{p}{r}{\msg}$ if $p=q$, or $\recv{q}{p}{\msg}$ if $p=r$, 
or $\varepsilon$ otherwise).

\begin{defi}[Projected system of CFSMs]\label{def:projected-system}
    The projected system of CFSMs, $\projectionof{\gt}$, 
    associated to a global type $\gt$ is the tuple $(\gt_p)_{p\in\Procs}$.
\end{defi}

\begin{exa}\label{ex:projection}
    Consider the  system of CFSMs  in Figure~\ref{fig:cfsm_ex}. This system can be obtained by projecting the  global type  of Figure~\ref{fig:ex-global-type}.
\end{exa}

%\etienne{Notion de produit synchrone: faire une definition plus marquante? dire que c'est le global type "naturel" d'un systeme?}
For every system, we can associate a global type by computing its synchronous product. 
Let $\cfsms=(\mathcal A_p)_{p\in\Procs}$ be a system of CFSMs, where 
$\mathcal A_p=(L_p,\Actp,\delta_p,l_{0,p},F_p)$ is the CFSM associated to process $p$.
The \emph{synchronous product} of $\cfsms$ is the NFA $\preproductof{\cfsms}=(L,\labelalphabet,\delta,l_{0},F)$, where
$(L,\labelalphabet,\delta,l_{0},F)$ where $L=\Pi_{p\in\Procs}L_p$, 
$l_0=(l_{0,p})_{p\in\Procs}$, $F=\Pi_{p\in\Procs}F_p$, and
$\delta$ is the transition relation defined by
\begin{itemize}
    \item
$(\vec l,\marrow{p}{q}{m},\vec l')\in\delta$ if
$(l_p,\send{p}{q}{m},l'_p)\in\delta_p$,
$(l_q,\recv{p}{q}{m},l'_q)\in\delta_q$,
and $l'_r=l_r$ for all $r\not\in\{p,q\}$, and
    \item $(\vec l,\varepsilon,\vec l')\in\delta$ if
    there is $p\in\Procs$ such that $(l_p,\varepsilon,l'_p)\in\delta_p$, and $l_q=l'_q$ for all $q\neq p$.
\end{itemize}

Note that, in general,  $\preproductof{\cfsms}$ is not a global type, as it may be non-deterministic. 
However, one can compute its determinisation $\productof{\gt}\eqdef\detof{\preproductof{\cfsms}}$ by standard powerset construction obtaining a  global type.

We write $\labeltoexec{w}$ to denote the synchronous execution coded by the sequence of
arrows $w$, i.e., the execution obtained by replacing each arrow $\marrow{p}{q}{m}$ of $w$ with the 
execution $\send{p}{q}{m}\cdot\recv{p}{q}{m}$.
We write $\labeltomsc{w}$ to denote the MSC $\mscof{\labeltoexec{w}}$.

A global type defines a language of MSCs in two different ways, one existential and one universal.
Let $\labellanguageof{\gt}$ be the set of 
sequences of arrows $w$ accepted by $\gt$.
Note that for 
$w\in\Arrows^*$, the function $w \mapsto\labeltomsc{w}$ with $ \labeltomsc{w}\in\mscsetofmodel{\synchmodel}$ is not injective, as two arrows with disjoint pairs of processes
commute. We write $w_1\sim w_2$ if 
$\labeltomsc{w_1}=\labeltomsc{w_2}$, and
$[w]$ for the equivalence class of $w$ wrt $\sim$.
The existential MSC language $\existentialmsclanguageof{\gt}$ of a global type $\gt$ is the set 
of MSCs that admit at least one representation as a sequence of arrows in $\labellanguageof{\gt}$,
and the universal MSC language $\universalmsclanguageof{\gt}$ of a global type $\gt$ is the set
of MSCs whose representations as a sequence of arrows are 
all in $\labellanguageof{\gt}$:
\iftwocolumns
$$
\begin{array}{rl}
\existentialmsclanguageof{\gt}\eqdef & \{\labeltomsc{w}\mid w\in\labellanguageof{\gt}\}\\
\universalmsclanguageof{\gt}\eqdef & \{\labeltomsc{w}\mid [w]\subseteq\labellanguageof{\gt}\}.
\end{array}
$$
\else
$$
\existentialmsclanguageof{\gt}\eqdef \{\labeltomsc{w}\mid w\in\labellanguageof{\gt}\}
\qquad
\universalmsclanguageof{\gt}\eqdef\{\labeltomsc{w}\mid [w]\subseteq\labellanguageof{\gt}\}.
$$
\fi

\begin{defi}[Commutation-closed]
    A global type $\gt$ is \emph{commutation-closed} if 
    $$
    \existentialmsclanguageof{\gt}=\universalmsclanguageof{\gt}.
    $$
\end{defi} 
In that case, we write $\msclanguageofcc{\gt}{}$ for the common language.

\begin{exa}
The global type in Figure~\ref{fig:mqtt-globaltype} is commutation-closed.
On the other hand, the global type 
$$
\gtlabel{p_1}{q}{m_1};\ \gtlabel{p_2}{q}{m_2};\ \mathsf{end}
$$
is not commutation-closed, as the two arrows $\gtlabel{p_1}{q}{m_1}$ and $\gtlabel{p_2}{q}{m_2}$ commute, 
and $\existentialmsclanguageof{\gt}=\{\mscof{m_1m_2}\}$
while $\universalmsclanguageof{\gt}=\emptyset$.
\end{exa}

\begin{prop} If $\#\Procs\leq 3$, then $\gt$ is commutation-closed.
\end{prop}
\begin{proof}
Any two arrows share at least one process and therefore do not commute.
\end{proof}

%\begin{defi}\label{def:commutation-closed-global-type}
%  A global type $\gt$ is \emph{commutation-closed} if 
%$\existentialmsclanguageof{\gt}=\universalmsclanguageof{\gt}$. 
%\end{defi}

%\cinzia{remove proposition and make it as plain text and move it to section 6 when it is needed for the decidability}
% \begin{prop}\label{prop:three-machines-implies-commutation-closed}
%     Let $\gt$ be global type with $\#\Procs\leq 3$ and no machine sending messages 
%     to itself.
%     Then $\gt$ is  commutation-closed.
% \end{prop}
% \begin{proof}
%     It follows easily by observing that there are no commutative arrows.
% \end{proof}

% \cinzia{never used, to remove}
% A MSC language $L\subseteq\mscsetofmodel{\synchmodel}$ is \emph{recognisable} if there is
% a commutation-closed global type $\gt$ such that $L=\msclanguageofcc{\gt}$.

The following properties/lemma will be useful in the rest of the paper. Recall that $\otimes$ is the product of NFAs, or here DFAs, mimicking language intersection over words.

\begin{prop}\label{prop:universal-existential-synch-inclusion}
    For all global type $\gt$, 
    $
    \universalmsclanguageof{\gt} \subseteq \existentialmsclanguageof{\gt} \subseteq \msclanguageofsystem{\projectionof{\gt}}{\synchmodel}.
    $
\end{prop}
\begin{proof}
    The first inclusion is immediate from the definitions.
    For the second inclusion, let $M\in \existentialmsclanguageof{\gt}$, we show that $M\in\msclanguageofsystem{\projectionof{\gt}}{\synchmodel}$.
    By definition of $\existentialmsclanguageof{\gt}$, there is a word $w\in\languageofnfa{\gt}$ such that
    $\mscof{w}=M$. Let $\rho$ be an accepting run of $\gt$ for $w$. For every $p\in\Procs$, let $\rho_p$ denote the run of $\gt_p$ (the machine of $p$ in 
    $\projectionof{\gt}$) obtained by projecting $\rho$: note that we kept $\varepsilon$ transitions in $\gt_p$, see Definition~\ref{def:projected-system},
    so $\rho_p$ is obviously defined. Then $\rho_p$ is an accepting run of $\gt_p$, therefore 
    $M\in\msclanguageof{\gt}{\synchmodel}$ (by Definition~\ref{def:executions-of-cfsms}).
\end{proof}

\begin{lem}\label{lem:product-of-gt}
    Let $\gt_1$ and $\gt_2$ be two global types, with $\gt_2$ commutation-closed. Then
    $$
    \existentialmsclanguageof{\gt_1\otimes\gt_2}=\existentialmsclanguageof{\gt_1}\cap\msclanguageofcc{\gt_2}.
    $$
\end{lem}
\begin{proof}
    We reason by double inclusion.
    \begin{itemize}
        \item[$\Rightarrow$] Assume that $M\in\existentialmsclanguageof{\gt_1\otimes\gt_2}$.
        By definition, there is a word $w$ such that $M=\mscof{w}$, and $w\in\languageofnfa{\gt_1\otimes\gt_2}$.
        By the definition of the product of global types, $\languageofnfa{\gt_1\otimes\gt_2}=\languageofnfa{\gt_1}\cap\languageofnfa{\gt_2}$.
        In particular, $w\in\languageofnfa{\gt_1}$, thus $\mscof{w}\in\existentialmsclanguageof{\gt_1}$.
        Similarly, $w\in\languageofnfa{\gt_2}$, thus $\mscof{w}\in\existentialmsclanguageof{\gt_2}$.
        Since $\gt_2$ is commutation-closed, $\mscof{w}\in\msclanguageofcc{\gt_2}$.
        Therefore, $\mscof{w}\in\existentialmsclanguageof{\gt_1}\cap\msclanguageofcc{\gt_2}$.
        \item[$\Leftarrow$] Conversely, assume that $M\in\existentialmsclanguageof{\gt_1}\cap\msclanguageofcc{\gt_2}$.
        In particular, $M\in\existentialmsclanguageof{\gt_1}$, which means that there is a word $w_1$ such that 
        $M=\mscof{w_1}$ and $w_1\in\languageofnfa{\gt_1}$.
        Since $M\in\msclanguageofcc{\gt_2}$, there is a word $w_2$ such that $M=\mscof{w_2}$ and $w_2\in\languageofnfa{\gt_2}$.
        Since $\gt_2$ is commutation-closed and $\mscof{w_1}=\mscof{w_2}$, $w_1\in \languageofnfa{\gt_2}$.
        Therefore, $w_1\in\languageofnfa{\gt_1}\cap\languageofnfa{\gt_2}=\languageofnfa{\gt_1\otimes\gt_2}$, 
        which means that $M\in\existentialmsclanguageof{\gt_1\otimes\gt_2}$.
    \end{itemize}
\end{proof}

\begin{prop}\label{prop:product-is-commutation-closed}
    For all system $\cfsms$ of communicating finite state machines, 
    $\productof{\cfsms}$ is commutation-closed and
    \iftwocolumns
    $$
    \msclanguageof{\cfsms}{\synchmodel} = \msclanguageofcc{\productof{\cfsms}}.
    $$
    \else
    $
    \msclanguageof{\cfsms}{\synchmodel} = \msclanguageofcc{\productof{\cfsms}}.
    $
    \fi
\end{prop}
\begin{proof}
It is sufficient to observe that the runs of the synchronous executions $$!m_1?m_1\ldots!m_k?m_k$$ of $\cfsms$ are in one-to-one correspondence with the 
    runs of the words $m_1\ldots m_k$ of $\preproductof{\cfsms}$.
\end{proof}

When a global type is implemented in a concrete system, its behaviour obviously depend on the chosen communication model. 

\begin{defi}[Global Type Semantics]\label{def:global-type-language}
    Let $\gt$ be a global type and $\acommunicationmodel$ a communication model. 
    The language of $\gt$ in $\acommunicationmodel$ is the set
    $
    \executionsof{\gt}{\acommunicationmodel}\eqdef
    \bigcup\{\linearizationsof{M}{\acommunicationmodel}\mid M\in\existentialmsclanguageof{\gt}\}
    $.
\end{defi}

\begin{exa}\label{ex:global-type-sem1}
    For any communication model $\acommunicationmodel$ that contains $\synchmodel$,
    The global type in Figure~\ref{fig:ex-global-type} has the following semantics:
    $$
    \executionsof{\gt}{\acommunicationmodel} = (!a?a!c?c+!a?a!b?b!d?d)^*.
    $$    
\end{exa}

\begin{exa}
    Consider the global type $$\gt=\mathsf{rec}~X.\mathsf{end}+\gtlabel{p}{q}{m_1};X+\gtlabel{p}{q}{m_2};X.$$
    The semantics of $\gt$ in the communication model $\bagmodel$ is
    the set of executions $e$ such that in any prefix of $e$ the number of $!m_1$ is greater or equal to the number of $?m_1$, 
    and the number of $!m_2$ is greater or equal to the number of $?m_2$, and moreover the numbers of 
    $!m_1$ and $?m_1$ are the same in the whole execution, as well as the numbers of
    $!m_2$ and $?m_2$.
    On the other hand, the semantics of $\gt$ in the communication model $\ppmodel$ (resp. $\causalmodel$) 
    is the set of executions $e$ that moreover contain the same word over $m_1$ and $m_2$
    when retaining only sends on the one hand and only receives on the other hand.
    For instance, $!m_1!m_2?m_2?m_1$ is in $\executionsofcfsms{\gt}{\bagmodel}$ but not in 
    $\executionsofcfsms{\gt}{\ppmodel}$.
\end{exa}

	\section{Realisability}\label{sec:realisability}
	% !TEX root =  ../main.tex
%!TEX spellcheck = en_GB
%\subsection{Definition}\label{sec:realisability_definition}

We have finally  collected all the definitions needed to introduce our notion of realisability
of global types that is parametric in a given communication model.
\begin{defi}[Deadlock-free realisability]\label{def:realisability}
    A global type $\gt$ is \emph{deadlock-free realisable}\footnote{We sometimes say simply \emph{realisable} instead of \emph{deadlock-free realisable}.}
    in the communication model
    $\acommunicationmodel$
    if the following two conditions hold:
    \begin{description}
    \item[(CC)] $\executionsof{\projectionof{\gt}}{\acommunicationmodel} = 
                \executionsof{\gt}{\acommunicationmodel}$.
    \item[(DF)] $\projectionof{\gt}$ is deadlock-free in $\acommunicationmodel$.
    \end{description}
\end{defi}

Condition~(CC) of Definition~\ref{def:realisability} corresponds to  \emph{session conformance}: the executions of the projected system do not 
deviate from the ones specified by the global type.

Deadlock-free realisability is equivalent to the notion of 
\emph{safe realisability} of ~\cite{DBLP:journals/tcs/AlurEY05} 
when $\acommunicationmodel$ is $\ppmodel$ or $\synchmodel$.
This is not the case for other communication models, our definition better captures the fact that a key factor for 
realisability is deadlock-freedom and deadlock freedom is strongly dependent on the communication model being causally-closed. 

\begin{exa}
    The global type in Figure~\ref{fig:ex-global-type} is deadlock-free realisable in any communication model that contains $\synchmodel$.
    Indeed, the projected system of CFSMs is the one in Figure~\ref{fig:cfsm_ex}, which is deadlock-free. 
    Moreover, the executions of the projected system are exactly those of the global type, as shown in Example~\ref{ex:global-type-sem1}.
\end{exa}

\begin{exa}
    The global type $\gtlabel{p_1}{q}{m_1};\ \gtlabel{p_2}{q}{m_2};\ \mathsf{end}$ is realisable in $\synchmodel$ and $\ppmodel$, but its projection is deadlock in $\mbmodel$.
\end{exa}

\begin{exa}
    The global type $\gtlabel{p_1}{q_1}{m_1};\ \gtlabel{p_2}{q_2}{m_2}; \ \mathsf{end}$ is realisable in $\synchmodel$;
    indeed, it is deadlock-free 
    and $\executionsof{\gt}{\synchmodel}=\{!m_1?m_1!m_2?m_2,!m_2?m_2!m_1?m_1\}=\executionsof{\projectionof{\gt}}{\synchmodel}$.
\end{exa}

\begin{exa}
    The global type $\gt$ in Figure~\ref{fig:mqtt-subtype} is deadlock-free realisable in $\synchmodel$.
    Indeed, it can be checked that 
    $\productof{\projectionof{\gt}}$ accepts the same sequences of arrows as 
    $\gt$, up to the commutation of the two last arrows $\gtlabel{b}{c_1}{end}$ and $\gtlabel{b'}{c_2}{end}$. 
    It therefore holds that $\executionsof{\projectionof{\gt}}{\synchmodel} = \executionsof{\gt}{\synchmodel}$.
    Moreover, $\gt$ contains no sink state (all states are coreachable from the accepting state), therefore the projected system is deadlock-free
    in $\synchmodel$.
\end{exa}

The following proposition reformulates conformance as an MSC property.
\begin{prop}%[Global type conformance as an MSC property]
    \label{prop:msc-version-of-cond1-of-realizability}
    Assume $\executionsofmodel{\acommunicationmodel}\supseteq\executionsofmodel{\synchmodel}$.
    Condition~(CC) of Definition~\ref{def:realisability} is equivalent to
    $$
        \mbox{(CC')} \quad \mscsofcfsms{\projectionof{\gt}}{\acommunicationmodel} \subseteq \existentialmsclanguageof{\gt}.
    $$
\end{prop}

\begin{proof}
    We establish the equivalence between Condition~(CC) of Definition~\ref{def:realisability} and Condition~(CC').
    We prove the two implications separately.
    \begin{itemize}
        \item[$\Rightarrow$]
            Assume that Condition~(CC) of Definition~\ref{def:realisability} holds.
            Let $M\in\mscsofcfsms{\projectionof{\gt}}{\acommunicationmodel}$; we prove that
            $M\in\existentialmsclanguageof{\gt}$.
            By definition of $\mscsofcfsms{\projectionof{\gt}}{\acommunicationmodel}$, there exists an execution $e\in\executionsof{\projectionof{\gt}}{\acommunicationmodel}$ such that $M=\mscof{e}$. By Condition~(CC) of Definition~\ref{def:realisability}, $e\in\executionsof{\gt}{\acommunicationmodel}$.
            By Definition~\ref{def:global-type-language}, $\mscof{e}\in\existentialmsclanguageof{\gt}$, which shows (CC').
        \item[$\Leftarrow$] 
            Assume that Condition~(CC') holds.
            We prove that Condition~(CC) of Definition~\ref{def:realisability} holds
            by showing the two inclusions.
            \begin{itemize}
                \item Let $e\in\executionsof{\projectionof{\gt}}{\acommunicationmodel}$.
                    By Definition~\ref{def:executions-of-cfsms}, $\mscof{e}\in\mscsofcfsms{\projectionof{\gt}}{\acommunicationmodel}$.
                    By Condition~(CC'), $\mscof{e}\in\existentialmsclanguageof{\gt}$, which means that $e\in\executionsof{\gt}{\acommunicationmodel}$.
                    Thus, $\executionsof{\projectionof{\gt}}{\acommunicationmodel} \subseteq \executionsof{\gt}{\acommunicationmodel}$.
                \item Let $e\in\executionsof{\gt}{\acommunicationmodel}$.
                    By Definition~\ref{def:global-type-language}, $\mscof{e}\in\existentialmsclanguageof{\gt}$.
                    By Condition~(CC'), $\mscof{e}\in\mscsofcfsms{\projectionof{\gt}}{\acommunicationmodel}$, which means that $e\in\executionsof{\projectionof{\gt}}{\acommunicationmodel}$.
                    Thus, $\executionsof{\gt}{\acommunicationmodel} \subseteq \executionsof{\projectionof{\gt}}{\acommunicationmodel}$.
            \end{itemize}
    \end{itemize}
\end{proof}

% \begin{wrapfigure}{r}{0.3\textwidth}
%     \centering
%     \input{figures/mqtt-bad.tex}
%     \caption{An MSC of the projection of the global type of Figure~\ref{fig:mqtt-subtype}.}
%     \label{fig:mqtt-bad}
% \end{wrapfigure}

\begin{exa}
    The global type $\gt$ in Figure~\ref{fig:mqtt-subtype} is
    deadlock-free realisable in $\ppmodel$: the communication model does not change the sets of MSCs of the projected system, i.e. $\msclanguageof{\projectionof{\gt}}{\ppmodel} = \msclanguageof{\projectionof{\gt}}{\synchmodel}$, and the projected system is deadlock-free in $\ppmodel$.
\end{exa}

\subsection{Subtyping}\label{sec:subtyping}
We adopt the following
semantic definition of subtyping.

\begin{defi}[Subtyping]\label{def:subtyping}
    A global type $\gt_1$ is a \emph{subtype} of a global type $\gt_2$ in the communication model $\acommunicationmodel$ if
    $
    \executionsof{\gt_1}{\acommunicationmodel}\subseteq\executionsof{\gt_2}{\acommunicationmodel}.
    $
\end{defi}

\begin{exa}
    The global type 
    $\mathsf{rec}~X.\mathsf{end}+\gtlabel{p}{q}{a}+\gtlabel{q}{p}{c};X$
    is a subtype of the global type
    in Figure~\ref{fig:ex-global-type}
    Indeed, the semantics of $\gt$ is the set of executions $!a?a!c?c^*$, which is contained in 
    the semantics of the other type already given in Example~\ref{ex:global-type-sem1}.
\end{exa}

Interestingly, we can prove that the communication model does not play a role in the definition of subtyping. This follows from the characterisation of global types as set of MSCs.

\begin{thm}
    \label{thm:subtyping-agnostic-communication-model}
    Let $\gt_1$ and $\gt_2$ be two global types, and $\acommunicationmodel$ a communication model
    such that $\executionsofmodel{\acommunicationmodel}\supseteq\executionsofmodel{\synchmodel}$.
    Then $\gt_1$ is a subtype of $\gt_2$ in $\acommunicationmodel$ if and only if
    $\existentialmsclanguageof{\gt_1}\subseteq\existentialmsclanguageof{\gt_2}$.
\end{thm}

\begin{proof}
    Assume $\gt_1$ is a subtype of $\gt_2$ in $\acommunicationmodel$.
    By Definition~\ref{def:subtyping}, we have
    $$\executionsof{\gt_1}{\acommunicationmodel}\cap \executionsofmodel{\synchmodel}\subseteq
    \executionsof{\gt_2}{\acommunicationmodel}\cap \executionsofmodel{\synchmodel}.$$
    Since $\executionsofmodel{\synchmodel}\subseteq\executionsofmodel{\acommunicationmodel}$,
    $$
    \executionsof{\gt_i}{\acommunicationmodel}\cap \executionsofmodel{\synchmodel} = \executionsof{\gt_i}{\synchmodel} = \existentialmsclanguageof{\gt_i}
    $$ 
    by Definition~\ref{def:global-type-language}.
    We conclude that
    $\existentialmsclanguageof{\gt_1}\subseteq\existentialmsclanguageof{\gt_2}$.
    The converse also follows from 
    Definitions~\ref{def:subtyping} and \ref{def:global-type-language}.
\end{proof}

\subsection{Complementability}\label{sec:complementability}
	% !TEX root =  ../main.tex
%!TEX spellcheck = en_GB

Next, we introduce the notion of complementation for global types, which is a key concept for our decidability results. 

\begin{defi}[Complement of a  Global Type]
    A global type $\comp{\gt}$ is a complement of $\gt$ 
    if $\existentialmsclanguageof{\gt}=
     \mscsetofmodel{\synchmodel}
    \setminus
    \existentialmsclanguageof{\comp{\gt}}$.
    We say that $\gt$ is {\em complementable} if it admits at least one complement.
\end{defi}

Interestingly,  not all global types are complementable. Moreover, we can compute the complement of some already known classes of global types such as 
commutation-closed, $\synchmodel$-realisable  and global types
with sender-driven choice. 

% !TEX root =  ../main.tex
%!TEX spellcheck = en_GB

We first give an example of non-complementable global type.

\begin{figure}
    \begin{tikzpicture}
    \draw[->] (0,0) node [above] {p} -- (0,-6);
    \draw[->] (1,0) node [above] {q} -- (1,-6);
    \draw[->] (2,0) node [above] {r} -- (2,-6);
    \draw[->] (3,0) node [above] {s} -- (3,-6);

    \draw[->] (0,-.5) -- node [above] {$m_1$} (1,-.5);
    \draw[->] (0,-1) -- node [above] {$m_1$} (1,-1);
    \node at (.5,-1.2) {$\vdots$};
    \draw[->] (0,-2.1) -- node [above] {$m_1$} (1,-2.1);

    \draw[->] (2,-1.5) -- node [above] {$m_2$} (3,-1.5);
    \draw[->] (2,-2) -- node [above] {$m_2$} (3,-2);
    \node at (2.5,-2.2) {$\vdots$};
    \draw[->] (2,-3.1) -- node [above] {$m_2$} (3,-3.1);

    \draw[->] (0,-3.5) -- node [above] {$m_3$} (1,-3.5);
    \draw[->] (0,-4) -- node [above] {$m_3$} (1,-4);
    \node at (.5,-4.2) {$\vdots$};
    \draw[->] (0,-5.1) -- node [above] {$m_3$} (1,-5.1);

    % on ajoute des accolades
    \draw [decorate,decoration={brace,amplitude=5pt,mirror}] (-.3,-.5) -- node [midway, left=1em] {$k_1$} (-.3,-2.1) ;
    \draw [decorate,decoration={brace,amplitude=5pt,mirror}] (-.3,-3.5) -- node [midway, left=1em] {$k_3$} (-.3,-5.1) ;
    \draw [decorate,decoration={brace,amplitude=5pt}] (3.3,-1.5) -- node [midway, right=1em] {$k_2$} (3.3,-3.1) ;

\end{tikzpicture}
    \caption{The shape of the MSCs in $\msclanguageofcc{\gt_0}{}$ for the global type $\gt_0$ in the proof of Theorem~\ref{thm:not-all-gt-are-complementable}
    }\label{fig:msc-regular}
\end{figure}
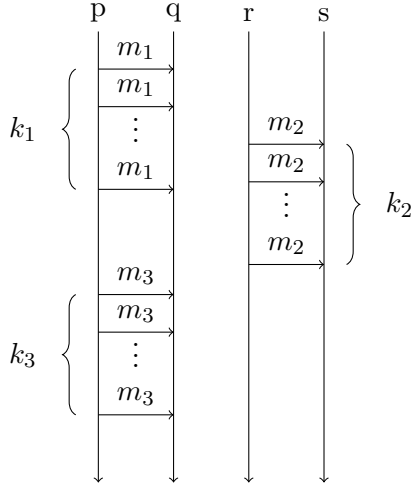  

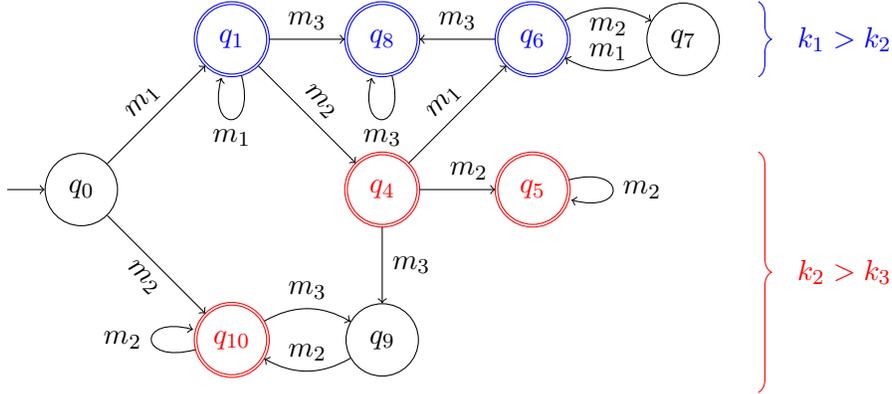
\begin{figure}[t]
    \begin{tikzpicture}

    \node[state, initial, initial text={}] (0) at (0,2) {$q_0$};
    \node[state, accepting, color=blue] (1) at (2,4) {$q_1$};
    % \node[state, accepting, color=teal] (2) at (2,-2) {$q_2$};
    % \node[state, accepting, color=brown] (3) at (2,0) {$q_3$};
    \node[state, accepting, color=red] (4) at (4,2) {$q_4$};
    \node[state, accepting, color=red] (5) at (6,2) {$q_5$};
    \node[state, accepting, color=blue] (6) at (6,4) {$q_6$};
    \node[state] (7) at (8,4) {$q_7$};
    \node[state, accepting, color=blue] (8) at (4,4) {$q_8$};
    \node[state] (9) at (4,0) {$q_9$};
    \node[state, accepting, color=red] (10) at (2,0) {$q_{10}$};
    % \node[state] (11) at (8,0) {$q_{11}$};
    % \node[state, accepting, color=brown] (12) at (10,0) {$q_{12}$};
    % \node[state] (13) at (6,-2) {$q_{13}$};
    % \node[state, accepting, color=teal] (14) at (8,-2) {$q_{14}$};
    % \node[state, accepting, color=teal] (15) at (10,-2) {$q_{15}$};

    \draw[->] (0) edge node [sloped,above] {$m_1$} (1);
    \draw[->] (0) edge node [sloped,below]{$m_2$} (10);
    % \draw[->] (0) edge node [above] {$m_3$} (3);
    \draw[->, loop below] (1) edge node {$m_1$} (1);
    % \draw[->, loop below] (2) edge node {$m_2$} (2);
    % \draw[->, loop below] (3) edge node {$m_3$} (3);
    % \draw[->] (1) edge node [left] {$m_3$} (3);
    \draw[->] (1) edge node [sloped, above] {$m_2$} (4);
    \draw[->] (1) edge node [above] {$m_3$} (8);
    \draw[->] (4) edge node [sloped, pos=.65, above] {$m_2$} (5);
    \draw[->, loop right] (5) edge node [right] {$m_2$} (5);
    \draw[->] (4) edge node [right] {$m_3$} (9);
    \draw[->] (4) edge node [sloped, above] {$m_1$} (6);
    \draw[->, bend left] (6) edge node [sloped, below] {$m_2$} (7);
    \draw[->, bend left] (7) edge node [sloped, above] {$m_1$} (6);
    \draw[->] (6) edge node [above] {$m_3$} (8);
    \draw[->, loop below] (8) edge node [below] {$m_3$} (8);
    \draw[->, bend left] (9) edge node [above] {$m_2$} (10);
    \draw[->, bend left] (10) edge node [above] {$m_3$} (9);
    \draw[->, loop left] (10) edge node [left] {$m_2$} (10);
    % \draw[->] (3) edge node [above] {$m_2$} (11);
    % \draw[->, bend left] (11) edge node [above] {$m_3$} (12);
    % \draw[->, bend left] (12) edge node [below] {$m_2$} (11);
    % \draw[->] (2) edge node [above] {$m_1$} (13);
    % \draw[->, bend left] (13) edge node [above] {$m_2$} (14);
    % \draw[->, bend left] (14) edge node [below] {$m_1$} (13);
    % \draw[->] (14) edge node [above] {$m_3$} (15);
    % \draw[->, loop right] (15) edge node [right] {$m_3$} (15);

    % on ajoute des accolades
    \draw [decorate,decoration={brace,amplitude=5pt}, color=blue] (9,4.5) -- node [midway, right=1em] {$k_1>k_2$} (9,3.5) ;
    \draw [decorate,decoration={brace,amplitude=5pt}, color=red] (9,2.5) -- node [midway, right=1em] {$k_2>k_3$} (9,-0.7) ;
    % \draw [decorate,decoration={brace,amplitude=5pt}, color=brown] (12,0.6) -- node [midway, right=1em] {$k_3>k_2$} (12,-.6) ;
    % \draw [decorate,decoration={brace,amplitude=5pt}, color=teal] (12,-1.5) -- node [midway, right=1em] {$k_2>k_1$} (12,-2.5) ;

\end{tikzpicture}
    \caption{A non-complementable global type}
    \label{fig:non-complementable-gt-1}
\end{figure}

\begin{lem}\label{thm:not-all-gt-are-complementable}
    Not all global types are complementable.
\end{lem}
\begin{proof}
    Let $\Procs=\{p,q,r,s\}$, $\Messages=\{m_1,m_2,m_3\}$, and 
    $\Arrows=\{\gtlabel{p}{q}{m_1},\gtlabel{r}{s}{m_2},\gtlabel{p}{q}{m_3}\}$.
    Consider a global type $\gt_0$ such that $\languageofnfa{\gt_0}=(m_1+m_2)^*(m_2+m_3)^*$ and take  the minimal DFA recognising this language.     Note that $\gt_0$ is commutation-closed and the MSCs of $\existentialmsclanguageof{\gt_0}$ are of the form depicted in Figure~\ref{fig:msc-regular}.

    Now consider the global type $\gt$ depicted in 
    Figure~\ref{fig:non-complementable-gt-1}. 
    We claim that for all natural numbers $k_1,k_2,k_3$,
    $$
    \mscof{m_1^{k_1}m_2^{k_2}m_3^{k_3}} \in \existentialmsclanguageof{\gt}
    \quad \mbox{if and only if} \quad
    k_1 > k_2 ~~\mbox{or}~~ k_2 > k_3.
    $$
    The claim follows by analysis on the paths leading to accepting states in $\gt$, as hinted by
    Figure~\ref{fig:non-complementable-gt-1}.
    For instance, the language of the words accepted by $\gt$ with state $q_8$ as the last state is of the form 
    $m_1^+(m_2m_1)^*m_3^+$, therefore lead to MSCs with $k_1>k_2$.

    We show that $\gt$ is not complementable. By contradiction, suppose that $\gt$ is complementable and $\gt'$ is a complement.
    Let $\gt''=\gt'\otimes \gt_0$; by Lemma~\ref{lem:product-of-gt}, and the fact that $\gt_0$ is commutation-closed,
    $$
        \existentialmsclanguageof{\gt''}~=~\msclanguageofcc{\gt_0}{}\cap\existentialmsclanguageof{\gt'}~=~
        \{\mscof{m_1^{k_1}m_2^{k_2}m_3^{k_3}}~\mid k_1\leq k_2\leq k_3\}.
    $$
    Now, let $\mathcal A$ denote the NFA obtained from $\gt''$ after replacing each $m_2$ transition with an $\varepsilon$ transition.
    Then $\languageofnfa{\mathcal A}=\{m_1^{k_1}m_3^{k_3}\mid k_1\leq k_3\}$, which is not a regular language, and hence  the contradiction.
\end{proof}

\paragraph{Complementation By Duality}
Recall that we write $\dualdfaof{\gt}$ for the dual DFA of a 
global type $\gt$, where accepting states and non-accepting ones are swapped (possibly completing first $\gt$ with a sink state).
It follows from the definition of $\existentialmsclanguageof{\gt}$ 
and $\universalmsclanguageof{\gt}$ that
$
\existentialmsclanguageof{\dualdfaof{\gt}}=\mscsetofmodel{\synchmodel}\setminus\universalmsclanguageof{\gt}.
$

In general $\dualdfaof{\gt}$ is not a complement of $\gt$, still  
duality can be used to obtain a complement of $\gt$ 
in a few cases. The first and most obvious case is the following.
\begin{lem}\label{prop:commutation-closed-implies-complementable}
    If $\gt$ is commutation-closed, then $\dualdfaof{\gt}$ is a complement of $\gt$.  
\end{lem}

As observed before, all global types with less than three participants are commutation closed, 
hence from Lemma~\ref{prop:commutation-closed-implies-complementable}
we have the following immediate result.

\begin{cor}\label{coro:3-participants-complementable}
    If $\cardinalof{\Procs}\leq 3$, then $\dualdfaof{\gt}$ is a complement of $\gt$.  
\end{cor}

Now, even if a global type $\gt$ is not commutation-closed, 
it may still be completed by commutation-closure and 
later complemented by duality; this construction yields
a complement of $\gt$ if $\gt$ is deadlock-free realisable in $\synchmodel$.

\begin{lem}\label{thm:realisable-complementable}
    If a global type $\gt$ is deadlock-free realisable in $\synchmodel$, 
    then
    $\dualdfaof{\productof{\projectionof{\gt}}}$
    is a complement of $\gt$.
\end{lem}
\begin{proof}
    Let $\gt$ be a global type that is  deadlock-free realisable in $\synchmodel$.
    By Condition~(CC) of Definition~\ref{def:realisability}, 
    $\executionsof{\projectionof{\gt}}{\synchmodel} = \executionsof{\gt}{\synchmodel}$.
    If two sets of executions are equal, the corresponding sets of their MSCs are also equal, 
    and thus $\msclanguageof{\projectionof{\gt}}{\synchmodel}=\msclanguageof{\gt}{\synchmodel}$.
    By Proposition~\ref{prop:product-is-commutation-closed}, the synchronous product $\productof{\projectionof{\gt}}$
    is a commutation-closed global type whose runs are exactly the synchronous executions of the projected system, so 
    $$\msclanguageofcc{\productof{\projectionof{\gt}}}=\msclanguageof{\projectionof{\gt}}{\synchmodel}=\existentialmsclanguageof{\gt}.$$

    Finally, by Proposition~\ref{prop:commutation-closed-implies-complementable},
    $\dualdfaof{\productof{\projectionof{\gt}}}$ is a complement of $\productof{\projectionof{\gt}}$, thus of $\gt$.
\end{proof}

\begin{rem}
    In terms of complexity, complementation by duality is
    a linear in the number of states of the global type $\gt$ (possibly after adding a sink state). However, 
    $\productof{\projectionof{\gt}}$ involves in the worst case a doubly exponential blowup. This may be  mitigated  if we can avoid to explicitly compute
    $\productof{\projectionof{\gt}}$.
\end{rem}

\paragraph{Complementation by Renunciation for Sender-Driven Choices}
Next, we introduce another complementation procedure  for
global types with sender-driven choices, and more generally presenting a form of determinism in commutations; this construction is also linear, although it does not preserve determinism. We first recall the definition of sender-driven choice and commutation-determinism, then define our complementation procedure, and finally establish its correctness.

For a state $s$ of $\gt$, let $\choicesof{\gt}{s}$ be the set of arrows labelling outgoing transitions of $s$ in $\gt$.

\begin{defi}[Sender-Driven~\cite{stutz-phd}]\label{def:sender-driven}
    A global type $\gt$ is \emph{sender-driven} if it is deterministic and for every state $s$ of $\gt$,
    $$
    \choicesof{\gt}{s} = \{\marrow{p}{q_i}{m_i}\mid i=1,\ldots,n\}
    $$
    for some process $p$ and some processes $q_i$ and messages $m_i$.
\end{defi}

Stutz et al. introduced this assumption on global types in order to prove that $\ppmodel$-implementability (a notion close to realisability) is decidable~\cite{DBLP:conf/cav/LiSWZ23}. We slightly generalise this condition as the complementation procedure we are about to present also works in this larger setting.

\begin{defi}[Commutation-determinismic]
    A global type $\gt$ is commutation-deterministic if it is deterministic and for every state $s$ of $\gt$,
    for every two arrows $a,b\in\choicesof{\gt}{s}$, $a$ and $b$ do not commute.
\end{defi}

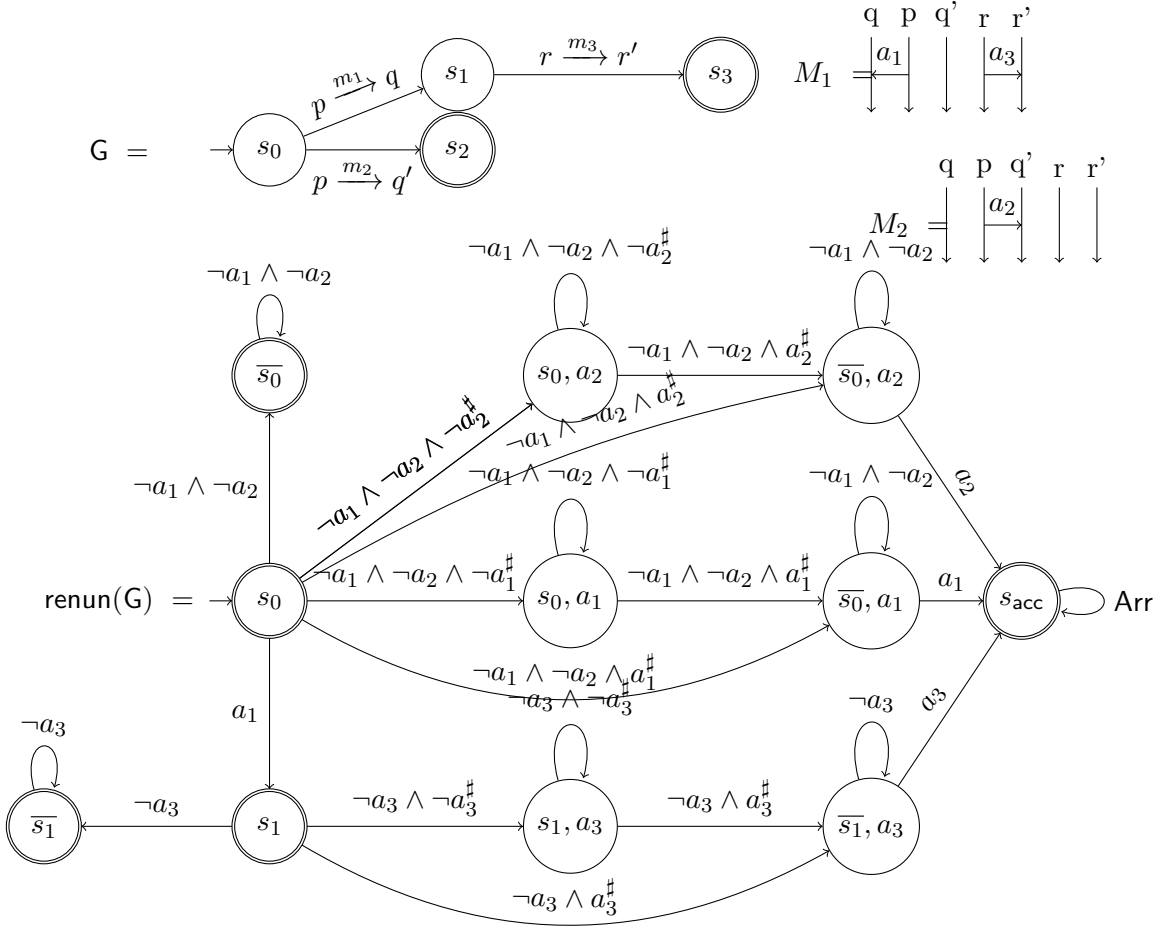
\begin{figure*}[ht]
        \begin{tikzpicture}
        \begin{scope}
            \node at (-1,0) {$\gt~=$};
            \node[state,initial,initial text={},initial distance=3mm] (s0) at (1,0) {$s_0$};
            \node[state] (s1) at (3.5,1) {$s_1$};
            \node[state,accepting] (s2) at (3.5,0)  {$s_2$};
            \node[state,accepting] (s3) at (7,1) {$s_3$};
            \draw[->] (s0) -- node[above,sloped] {$\gtlabel{p}{q}{m_1}$} (s1);
            \draw[->] (s0) -- node[below] {$\gtlabel{p}{q'}{m_2}$} (s2);
            \draw[->] (s1) -- node[above] {$\gtlabel{r}{r'}{m_3}$} (s3);
        \end{scope}
        \begin{scope}[xshift=8.5cm, yshift=1cm, scale=.5]
            \node at (0,0) {$M_1~=$};
            \draw[->] (1,1) node [above] {q}  -- (1,-1);
            \draw[->] (2,1) node [above] {p}  -- (2,-1);
            \draw[->] (3,1) node [above] {q'} -- (3,-1);
            \draw[->] (4,1) node [above] {r}  -- (4,-1);
            \draw[->] (5,1) node [above] {r'} -- (5,-1);
            \draw[->] (2,0) -- node [above] {$a_1$} (1,0);
            \draw[->] (4,0) -- node [above] {$a_3$} (5,0);
        \end{scope}
        \begin{scope}[xshift=9.5cm, yshift=-1cm, scale=.5]
            \node at (0,0) {$M_2~=$};
            \draw[->] (1,1) node [above] {q}  -- (1,-1);
            \draw[->] (2,1) node [above] {p}  -- (2,-1);
            \draw[->] (3,1) node [above] {q'} -- (3,-1);
            \draw[->] (4,1) node [above] {r}  -- (4,-1);
            \draw[->] (5,1) node [above] {r'} -- (5,-1);
            \draw[->] (2,0) -- node [above] {$a_2$} (3,0);
        \end{scope}
        \begin{scope}[yshift=-6cm]
            \node at (-1,0) {$\renun{\gt}~=$};
            \node[state,initial,initial text={},initial distance=3mm, accepting] (s0) at (1,0) {$s_0$};
            \node[state,accepting] (s1) at (1,-3) {$s_1$};
            \node[state,accepting] (s0bar) at (1,3)  {$\comp{s_0}$};
            \node[state,accepting] (s1bar) at (-2,-3) {$\comp{s_1}$};
            \node[state] (s0a1) at (5,0) {$s_0,a_1$};
            \node[state] (s0a2) at (5,3) {$s_0,a_2$};
            \node[state] (s0bara1) at (9,0) {$\comp{s_0},a_1$};
            \node[state] (s0bara2) at (9,3) {$\comp{s_0},a_2$};
            \node[state,accepting] (sacc) at (11,0) {$s_{\mathsf{acc}}$};
            \node[state] (s1a3) at (5,-3) {$s_1,a_3$};
            \node[state] (s1bara3) at (9,-3) {$\comp{s_1},a_3$};
            \draw[->] (s0) -- node[left] {$a_1$} (s1);
            \draw[->] (s0) -- node[left] {$\neg a_1\wedge \neg a_2$} (s0bar);
            \draw[->] (s0bar) edge [loop above] node[above] {$\neg a_1\wedge \neg a_2$} (s0bar);
            \draw[->] (s0) edge [bend right] node[above,sloped] {$\neg a_1\wedge \neg a_2\wedge\conflictingarrows{a_1}$} (s0bara1);
            \draw[->] (s0) -- node[above,sloped] {$\neg a_1\wedge \neg a_2\wedge\neg\conflictingarrows{a_1}$} (s0a1);
            \draw[->] (s0a1) edge [loop above] node[above] {$\neg a_1\wedge \neg a_2\wedge\neg\conflictingarrows{a_1}$} (s0a1);
            \draw[->] (s0a1) -- node[above,sloped] {$\neg a_1\wedge \neg a_2\wedge\conflictingarrows{a_1}$} (s0bara1);
            \draw[->] (s0bara1) edge [loop above] node[above] {$\neg a_1\wedge\neg a_2$} (s0bara1);
            \draw[->] (s0bara1) -- node[above,sloped] {$a_1$} (sacc);
            \draw[->] (s0) edge [bend left=9] node[above,sloped,pos=0.6] {$\neg a_1\wedge \neg a_2\wedge\conflictingarrows{a_2}$} (s0bara2);
            \draw[->] (s0) -- node[above,sloped] {$\neg a_1\wedge \neg a_2\wedge\neg\conflictingarrows{a_2}$} (s0a2);
            \draw[->] (s0) -- node[above,sloped] {$\neg a_1\wedge \neg a_2\wedge\neg\conflictingarrows{a_2}$} (s0a2);
            \draw[->] (s0a2) edge [loop above] node[above] {$\neg a_1\wedge \neg a_2\wedge\neg\conflictingarrows{a_2}$} (s0a2);
            \draw[->] (s0a2) -- node[above,sloped] {$\neg a_1\wedge \neg a_2\wedge\conflictingarrows{a_2}$} (s0bara2);
            \draw[->] (s0bara2) edge [loop above] node[above] {$\neg a_1\wedge\neg a_2$} (s0bara2);
            \draw[->] (s0bara2) -- node[above,sloped] {$a_2$} (sacc);
            \draw[->] (s1) -- node [above] {$\neg a_3$}(s1bar);
            \draw[->] (s1bar) edge[loop above] node[above] {$\neg a_3$} (s1bar);
            \draw[->] (s1) edge [bend right] node[above,sloped] {$\neg a_3\wedge \conflictingarrows{a_3}$} (s1bara3);
            \draw[->] (s1) -- node[above,sloped] {$\neg a_3\wedge\neg\conflictingarrows{a_3}$} (s1a3);
            \draw[->] (s1a3) edge [loop above] node[above] {$\neg a_3\wedge\neg\conflictingarrows{a_3}$} (s1a3);
            \draw[->] (s1a3) -- node[above,sloped] {$\neg a_3\wedge \conflictingarrows{a_3}$} (s1bara3);
            \draw[->] (s1bara3) edge [loop above] node[above] {$\neg a_3$} (s1bara3);
            \draw[->] (s1bara3) -- node[above,sloped] {$a_3$} (sacc);

            \draw[->] (sacc) edge [loop right] node[right] {$\Arrows$} (sacc);

        \end{scope}

    \end{tikzpicture}
    \caption{A sender-driven global type $\gt$, with $\existentialmsclanguageof{\gt}=\{M_1,M_2\}$, and its
            complement $\renun{\gt}$.
            For better readability, the states $s_2$ and $s_3$ (that are sink states in $\renun{\gt}$) have been pruned.
            $\conflictingarrows{a}$ denotes the set of arrows that do not commute with $a$.
            }\label{fig:example-of-sender-driven-gt-complementation}
\end{figure*}

We now formally define the complementation procedure for commutation-deterministic global types.

\begin{defi}[Complementation by renunciation]\label{def:complement-of-a-sender-driven-gt}
    Let $\gt$ be a commutation-deterministic global type. Let $S$ be the set of states of $\gt$, $\comp{S}\eqdef\{\comp{s}\mid s\in S\}$,
    and $S'=S\uplus\comp{S}$.
   $\renun{\gt}$ is the global type with set of states $S'\cup S'\times\Arrows \cup\{s_{\mathsf{acc}}\}$ defined as follows:
    \begin{itemize}
        \item the initial state of $\renun{\gt}$ is the initial state of $\gt$;
        \item let $F\subseteq S$ be the set of accepting states of $\gt$; the set of accepting states of 
        $\renun{\gt}$ is $(S\setminus F)\cup \comp{S}\cup\{s_{\mathsf{acc}}\}$;
        \item for any states $s,s'\in S$ of $\gt$, for any arrow $a\in \Arrows$,
        if $(s,a,s')$ is a transition in $\gt$, then $(s,a,s')$ 
        is a transition in $\renun{\gt}$;
        \item for any state $s$ of $\gt$, 
        for any arrow $a\not\in\choicesof{\gt}{s}$,
        $(s,a,\comp{s})$ (resp. $(\comp{s},a,\comp{s})$) 
        is a transition in $\renun{\gt}$;
        \item for any state $s$ of $\gt$,
        for any arrow $a\in\choicesof{\gt}{s}$,
        for any arrow $b\not\in\choicesof{\gt}{s}$ that commutes with $a$,
        $\big(s,b,(s,a)\big)$ 
        (resp. $\big((s,a),b,(s,a)\big)$) 
        is a transition in $\renun{\gt}$
        \item for any state $s$ of $\gt$,
        for any arrow $a\in\choicesof{\gt}{s}$,
        for any arrow $b\not\in\choicesof{\gt}{s}$ that does not commute with $a$,
        $\big(s,b,(\comp{s},a)\big)$ 
        (resp. $\big((s,a),b,(\comp{s},a)\big)$) 
        is a transition in $\renun{\gt}$
        \item for any state $s$ of $\gt$,
        for any arrow $a\in\choicesof{\gt}{s}$,
        for any arrow $b\not\in\choicesof{\gt}{s}$,
        $\big((\comp{s},a),b,(\comp{s},a)\big)$
        (resp. $\big((\comp{s},a),a,s_{\mathsf{acc}}\big)$)
        is a transition in $\renun{\gt}$;
        \item for any arrow $a$, $(s_{\mathsf{acc}},a,s_{\mathsf{acc}})$ is a transition in $\renun{\gt}$.
    \end{itemize}
\end{defi}

\begin{exa}\label{ex:example-of-sender-driven-gt-complementation}
    Figure~\ref{fig:example-of-sender-driven-gt-complementation} depicts a sender-driven global type
    and the complement computed according to Definition~\ref{def:complement-of-a-sender-driven-gt}.
\end{exa}

Note that the number of states of $\renun{\gt}$ is linear in the number of states of $\gt$ (assuming the alphabet $\Arrows$ is fixed); it can also be observed that, using Boolean
expressions to label transitions, the size of $\renun{\gt}$ can also be kept linear in the size of $\gt$.

\begin{exa}
    Let $\gt,\renun{\gt}, M_1, M_2$ be the global types and MSCs depicted in Figure~\ref{fig:example-of-sender-driven-gt-complementation}.
    It is easy to verify that $M_2\not\in\existentialmsclanguageof{\renun{\gt}}$, because the only sequence of arrows
    $w_2=a_2$ such that $M_2=\mscof{w_2}$ is not in $\languageofnfa{\renun{\gt}}$.
    It can also be checked that $M_1\not\in\existentialmsclanguageof{\renun{\gt}}$, because the only two  sequences of arrows $w_{1,1}=a_1a_3$
    and $w_{1,2}=a_3a_1$ such that $\mscof{w_{1,1}}=\mscof{w_{1,2}}=M_1$ are not in $\languageofnfa{\renun{\gt}}$.
\end{exa}

Intuitively, $\renun{\gt}$ describes a MSC $M$ that starts with a prefix that is in $\existentialmsclanguageof{\gt}$ up to a certain state $s$, at which point a renunciation of the choices of $\gt$ at $s$ occurs. 
There are two kinds of renunciation: definitive renunciation means that none of the arrows available in the choices will ever occur later.
Provisory renunciation means that at least one arrow in the choices will occur later, but not immediately: if $a$ is the first
arrow of the choice that occurs later, another arrow $b$ occurs before $a$ such that $b$ does not commute with $a$ and is not in the choices that
have been renunciated. The first kind of renunciation corresponds to moving to a state of the form $\comp{s}$, while the second kind corresponds to moving to a state of the form $(s,a)$ or $(\comp{s},a)$.

\begin{exa}
    Consider the MSC $M_3=\mscof{a_1a_2'a_3}$, with $a_2'=\gtlabel{p}{q'}{m_2'}$ for some $m_2'\neq m_2$. Then $M_3$ is accepted by $\renun{\gt}$ because $w_{3,1}=a_1a_2'a_3$ is accepted in $\comp{s_1}$. However, the other sequence of arrows $w_{3,2}=a_3a_1a_2'$
    such that $\mscof{w_{3,2}}=M_3$ is not in $\languageofnfa{\renun{\gt}}$. More generally, $\renun{\gt}$ is \emph{not} commutation-closed, and for any MSC $M\not\in\existentialmsclanguageof{\gt}$, the
    sequence of arrows $w$ such that $M=\mscof{w}$ and $w\in\languageofnfa{\renun{\gt}}$ should be carefully constructed. The strategy
    consists in picking a sequence of arrows $w=w_1\cdot w_2$
    with the longest possible $w_1$ 
    in $\prefixclosureof{\languageofnfa{\gt}}$, i.e., renunciating to a choice as late as possible.
\end{exa}

In order to formalise this intuition, we introduce a few notions.
Given an MSC $M$ and a state $s$ of a commutation-deterministic global type $\gt$,
we write $\nextarrow{M}{s}$ for the first arrow of $\choicesof{\gt}{s}$ that occurs in a sequence of arrows $w$
such that $\mscof{w}=M$. Note that $\nextarrow{M}{s}$ is not defined if $M$ contains no arrow of 
$\choicesof{\gt}{s}$; note also that $\nextarrow{M}{s}$, when defined, does not depend on the choice of the sequence of arrows 
$w$ such that $\mscof{w}=M$ (otherwise two arrows of $\choicesof{\gt}{s}$ would commute, contradicting the commutation-determinism of $\gt$).

\begin{exa}
    Consider the MSC $M_4=\mscof{a_1a_2a_3}=\mscof{a_3a_1a_2}$, and the initial state $s_0$ of $\gt$ in Figure~\ref{fig:example-of-sender-driven-gt-complementation}.
    Then $\nextarrow{M_4}{s_0}=a_1$.
\end{exa}

We also define $\nextmsc{M}{s}$ as $\mscof{w}$ for some $w$ such
that $M=\mscof{\nextarrow{M}{s}w}$; intuitively, we remove the arrow $\nextarrow{M}{s}$ from $M$, provided it is not blocked by 
a previous non-commuting arrow. Note that $\nextmsc{M}{s}$ is not defined if 
in all sequences of arrows $w$ such that $\mscof{w}=M$, $\nextarrow{M}{s}$ never occurs first.

\begin{exa}
    Consider the MSC $M_4=\mscof{a_1a_2a_3}$, and the state $s_0$ of $\gt$ in Figure~\ref{fig:example-of-sender-driven-gt-complementation}.
    Then $\nextmsc{M_4}{s_0}=\mscof{a_2a_3}$.
    On the other hand, for the MSC $M_5=\mscof{a_4a_1a_2}$, with $a_4=\gtlabel{q}{q'}{m_4}$, $\nextmsc{M_5}{s_0}$
    is undefined, because $a_1$ is blocked by $a_4$.
\end{exa}

\begin{lem}\label{lem:nextarrow-and-nextmsc}
    Let $\gt$ be a commutation-deterministic global type with initial state $s_0$, and $M$ a non-empty MSC.
    Then $M\not\in\existentialmsclanguageof{\gt}$ if and only if one of the following holds:
    \begin{itemize}
        \item $\nextarrow{M}{s_0}$ is undefined, or 
        \item $\nextmsc{M}{s_0}$ is undefined, or
        \item $a=\nextarrow{M}{s_0}$, $s$ is the $a$-successor of $s_0$ in $\gt$,
        $M'=\nextmsc{M}{s_0}$ and 
        $M'\not\in\existentialmsclanguageof{\gt_s}$,
        where $\gt_s$ is the global type obtained from $\gt$ by setting the initial state to $s$.
    \end{itemize}
\end{lem}
\begin{proof}
    If $\nextarrow{M}{s_0}$ is undefined, then $M$ contains no arrow of $\choicesof{\gt}{s_0}$, 
    so for every sequence of arrows $w$ such that $\mscof{w}=M$, $w$ does not start with
    an arrow of $\choicesof{\gt}{s_0}$, therefore $w\not\in\languageofnfa{\gt}$.
    Similarly, when $\nextmsc{M}{s_0}$ is undefined, every sequence of arrows $w$ such that $\mscof{w}=M$ does not start with
    an arrow of $\choicesof{\gt}{s_0}$, and again $w\not\in\languageofnfa{\gt}$.
    In the third case, assume by contradiction that $M\in\existentialmsclanguageof{\gt}$,
    i.e. there is a sequence of arrows $w$ such that $\mscof{w}=M$ and $w\in\languageofnfa{\gt}$.
    Then $w=aw'$ for some $a=\nextarrow{M}{s_0}$ and $w'\in\languageofnfa{\gt,s}$, with $\nextmsc{M}{s_0}=\mscof{w'}$. It follows that     
    $\nextmsc{M}{s_0}\in\existentialmsclanguageof{\gt_s}$, and the contradiction.
\end{proof}

\begin{thm}\label{thm:sender-driven-choice-complementation}
    Assume $\gt$ is a commutation-deterministic global type,
    and let $\renun{\gt}$ denote the global type defined as in Def~\ref{def:complement-of-a-sender-driven-gt}.
    Then $\renun{\gt}$ is a complement of $\gt$.
\end{thm}
\begin{proof}
    It is routine to check that an MSC $M$ is accepted by $\renun{\gt}$ 
    starting from a state $\comp{s}$ if and only if $\nextarrow{M}{s}$ is undefined. It is also routine to check that an MSC $M$ is accepted by $\renun{\gt}$
    starting from a state $(s,a)$ if $a=\nextarrow{M}{s}$ and $\nextmsc{M}{s}$ is undefined.
    It is then also routine, thanks to Lemma~\ref{lem:nextarrow-and-nextmsc}, to show by induction on the size of an MSC $M$ that
    $M\not\in\existentialmsclanguageof{\gt_s}$ if and only if $M\in\existentialmsclanguageof{\renun{\gt_s}}$.
\end{proof}

Hence summing up:
\begin{thm}\label{thm:complement}
Not all global types are complementable. A global type $\gt$ is complementable 
in the following cases:
\begin{enumerate}
\item if $\gt$ is commutation-closed (in particular with at most 3 processes), then $\dualdfaof{\gt}$ is a complement of $\gt$
 (with a linear increase in the number of states)
\item if $\gt$ is deadlock-free realisable in $\synchmodel$, 
then $\dualdfaof{\productof{\projectionof{\gt}}}$ is a complement of $\gt$ 
(with a doubly exponential blowup in the number of states)
\item if $\gt$ has sender-driven choice, it admits an effectively computable complement  with a linear increase in the number of states
\end{enumerate}
\end{thm}
\begin{proof}
\begin{enumerate}
\item  Follows from Lemma \ref{thm:not-all-gt-are-complementable}
\item 
\begin{enumerate}
 \item Follows from Lemma \ref{prop:commutation-closed-implies-complementable}
 \item Follows from Lemma \ref{thm:realisable-complementable}
 \item Follows from Lemma \ref{thm:sender-driven-choice-complementation}\qedhere
 \end{enumerate}
 \end{enumerate}
\end{proof}

Complexities of construction are  summarised in the following table:

\begin{center}
\begin{tabular}{|c|c|c|c|}
\hline
\textbf{Class of global types} & \textbf{Complementation procedure} & \textbf{Size} \\
\hline
$|\Procs|\leq 3$ & $\gt\mapsto \dualdfaof{\gt}$ & linear \\
\hline
commutation-closed & $\gt\mapsto \dualdfaof{\gt}$ & linear \\
\hline
realisable & $\gt\mapsto \dualdfaof{\productof{\projectionof{\gt}}}$ & doubly exponential \\
\hline
sender-driven choices & $\gt\mapsto \renun{\gt}$ & linear \\
\hline
\end{tabular}
\end{center}

Notice that apart from the classes mentioned above, deciding whether a complement exists and compute it, is an open problem.

	\section{Decidability}\label{sec:decidability}
	
% !TEX root =  ../main.tex
%!TEX spellcheck = en_GB

Deadlock-free realisability in $\ppmodel$ is known to be undecidable for global types~\cite{DBLP:journals/tcs/Lohrey03}, and decidable for sender-driven choices~\cite{Stutz24-phd}. Here we  show that the realisability problem, for global types for which an explicit complement is given, is decidable in PSPACE for the synchronous communication model and in EXPSPACE for RegSC (see Definition~\ref{def:regsc-comm-model} below), causally-closed, communication models.
We show the result for $\synchmodel$, and then extend it to any causally-closed communication model.
But first, we consider the decidability of the subtyping problem.

\subsection{Decidability of Subtyping}
\label{sec:decidability-subtyping}

\begin{thm}\label{thm:subtyping-decidability}
    Given a global type $\gt$ 
    and a complementable global type $\gt'$ with an explicit commutation-closed complement $\comp{\gt'}$, it is decidable in NLOGSPACE whether $\gt$ is a subtype of $\gt'$.
\end{thm}

\begin{proof}
    Thanks to Theorem~\ref{thm:subtyping-agnostic-communication-model},
    $\gt$ is a subtype of $\gt'$ if and only if
    $\existentialmsclanguageof{\gt} \subseteq \existentialmsclanguageof{\gt'}$,
    i.e., $\existentialmsclanguageof{\gt} \cap \msclanguageofcc{\comp{\gt'}} = \emptyset$.
    By Lemma~\ref{lem:product-of-gt}, this reduces to check whether
    $\languageofnfa{\gt\otimes\comp{\gt'}} = \emptyset$, which is decidable in NLOGSPACE.
\end{proof}

It is possible to relax the assumption on the complement of $\gt'$ to be commutation-closed, as follows.

\begin{thm}\label{thm:subtyping-decidability-2}
    Given a global type $\gt$ that is realisable in $\synchmodel$, and a complementable global type $\gt'$ with an explicit complement $\comp{\gt'}$, it is decidable in EXPSPACE whether $\gt$ is a subtype of $\gt'$.
\end{thm}

\begin{proof}
    By Theorem~\ref{thm:complement}, there is a commutation-closed global type $\comp{\comp{\gt}}$
    of double exponential size such that 
    $\existentialmsclanguageof{\gt} = \msclanguageofcc{\comp{\comp{\gt}}}$.
    Thanks to Theorem~\ref{thm:subtyping-agnostic-communication-model},
    $\gt$ is a subtype of $\gt'$ if and only if
    $\existentialmsclanguageof{\gt} \cap  \existentialmsclanguageof{\comp{\gt'}}=
    \msclanguageofcc{\comp{\comp{\gt}}} \cap  \existentialmsclanguageof{\comp{\gt'}}
    \emptyset$.
    By Lemma~\ref{lem:product-of-gt}, this reduces to check whether
    $\languageofnfa{\comp{\comp{\gt}}\otimes\comp{\gt'}} = \emptyset$.
\end{proof}

\begin{cor}
    \label{cor:subtyping-decidability}
    Subtyping between two realisable global types 
    is decidable in EXPSPACE.
\end{cor}

\subsection{Realisability in the Synchronous Model}

\begin{thm}
    \label{thm:decidability-of-implementability-in-synch}
    Given a global type $\gt$ and its complement $\gt'$, it is decidable in PSPACE whether  $\gt$ is deadlock-free realisable in $\synchmodel$.    
%  The  problem is in PSPACE:
  %  \begin{itemize}
    %    \item Input: a global type $\gt$ and a complement $\gt'$ of $\gt$.
      %  \item Question: is ?
    %\end{itemize}
\end{thm}

\begin{proof}
    Condition~(CC) of Definition~\ref{def:realisability}
    is equivalent to $$\msclanguageofcc{\productof{\projectionof{\gt}}} \cap \existentialmsclanguageof{\gt'} = \emptyset.$$
    By Lemma~\ref{lem:product-of-gt}, this reduces to checking the emptiness of $\languageofnfa{\productof{\projectionof{\gt}}\otimes\gt'}$.
    Recall that $\productof{\cfsms}=\detof{\preproductof{\cfsms}}$, where 
    $\preproductof{\cfsms}$ denotes the synchronous product of the communicating machines keeping all $\varepsilon$-transitions
    and non-determinism.
    It follows that Condition~(CC) of Definition~\ref{def:realisability} is equivalent to checking the emptiness of the language of the NFA $\preproductof{\projectionof{\gt}}\otimes\gt'$.
    The non-emptiness of the language of an NFA
    can be checked in non-deterministic logarithmic space in the size of the NFA,
    which yields a PSPACE algorithm in that case, as the exponential size NFA  $\preproductof{\projectionof{\gt}}$
    can be lazily constructed while guessing a path to an accepting state.
    
    Assuming that Condition~(CC) of Definition~\ref{def:realisability} holds,
    the second condition of Definition~\ref{def:realisability}  is 
    equivalent to whether all control states 
    of $\preproductof{\projectionof{\gt}}$ that are reachable from the initial state
    can reach a final state, which again, for the same reason, can be checked in PSPACE.
\end{proof}

\subsection{Realisability in Causally-Closed Communication Models}

We now show that the realisability problem for complementable global types, given with an explicit complement, is decidable in EXSPACE 
for $\bagmodel$, $\ppmodel$, and $\causalmodel$.
The proof goes throughout the following steps.
\begin{enumerate}
    \item We show that realisability 
    in any communication model that contains $
    \synchmodel$ implies realisability in $\synchmodel$ (Theorem~\ref{thm:pp-realizability-implies-synch-realizability}).
    \item We address the contraposite, and identify three  necessary and sufficient conditions on the projection of $\gt$ in $\acommunicationmodel$ (Theorem~\ref{thm:main-theorem-realisability}), namely 
    \begin{enumerate}
        \item $\projectionof{\gt}$ is RSC~\cite{germerie-phd} 
        \item $\projectionof{\gt}$ is orphan-free, and
        \item $\projectionof{\gt}$ is deadlock-free.
    \end{enumerate}
    \item We show that each of these three  conditions is decidable in EXPSPACE.
    \item We conclude with the decidability of the realisability problem in $\acommunicationmodel$ (Theorem~\ref{thm:decidability-of-implementability-in-p2p} thanks to the  decidability of the realisability problem in $\synchmodel$ (Theorem~\ref{thm:decidability-of-implementability-in-synch}).
\end{enumerate}

\begin{thm}\label{thm:pp-realizability-implies-synch-realizability}
%    If $\gt$ is deadlock-free realisable in $\acommunicationmodel$, then $\gt$ is deadlock-free realisable in $\synchmodel$.
    Assume $\executionsofmodel{\acommunicationmodel}\supseteq\executionsofmodel{\synchmodel}$.
%    and $\acommunicationmodel$ is causally closed.
    If the global type $\gt$ is deadlock-free realisable in $\acommunicationmodel$, then $\gt$ is deadlock-free realisable in $\synchmodel$.
\end{thm}

\begin{proof}
Assume that $\gt$ is 
deadlock-free realisable in $\acommunicationmodel$. We show that $\gt$ is deadlock-free realisable in $\synchmodel$ by verifying the 
two conditions of Definition~\ref{def:realisability}.   
\begin{enumerate}
\item 
    By the hypothesis that 
    $\gt$ is $\acommunicationmodel$-realisable, we have that
    $
    \executionsofcfsms{\projectionof{\gt}}{\acommunicationmodel} = \executionsofcfsms{\gt}{\acommunicationmodel}
    $,
    and therefore
    $$
    \executionsofcfsms{\projectionof{\gt}}{\acommunicationmodel}\cap\executionsofmodel{\synchmodel} = \executionsofcfsms{\gt}{\acommunicationmodel}\cap\executionsofmodel{\synchmodel}.
    $$
    Since $\executionsofmodel{\synchmodel}\subseteq\executionsofmodel{\acommunicationmodel}$, we have
    $\executionsofcfsms{\projectionof{\gt}}{\synchmodel}\subseteq\executionsofcfsms{\projectionof{\gt}}{\acommunicationmodel}\cap\executionsofmodel{\synchmodel}$.
    On the other hand, by Definition~\ref{def:global-type-language}, $
    \executionsofcfsms{\gt}{\acommunicationmodel}\cap\executionsofmodel{\synchmodel} =
    \executionsofcfsms{\gt}{\synchmodel}$.
    Therefore, we have
    $$    \executionsofcfsms{\projectionof{\gt}}{\synchmodel}\subseteq\executionsofcfsms{\gt}{\synchmodel}.$$
    Conversely, 
    by Property~\ref{prop:universal-existential-synch-inclusion},
    $\existentialmsclanguageof{\gt}\subseteq
    \msclanguageof{\projectionof{\gt}}{\synchmodel}$,
    and by Definition~\ref{def:global-type-language}, we get 
    $\executionsof{\gt}{\synchmodel}\subseteq
    \executionsof{\projectionof{\gt}}{\synchmodel}$, which ends the proof of Condition~(CC) of Definition~\ref{def:realisability}.

\item By the hypothesis that 
    $\gt$ is $\acommunicationmodel$-realisable and 
    Definition~\ref{def:deadlock-free},
    we have that
    $$
    \executionsofcfsms{\projectionof{\acceptcompletion{\gt}}}{\acommunicationmodel} \subseteq\prefixclosureof{\executionsof{\projectionof{\gt}}{\acommunicationmodel}}=\prefixclosureof{\executionsof{\gt}{\acommunicationmodel}}.
    $$
    and therefore, by $\executionsofmodel{\synchmodel}\subseteq\executionsofmodel{\acommunicationmodel}$,
    $$
    \executionsofcfsms{\projectionof{\acceptcompletion{\gt}}}{\acommunicationmodel}\cap \executionsofmodel{\synchmodel} \quad \subseteq \quad
    \prefixclosureof{\executionsof{\projectionof{\gt}}{\acommunicationmodel}}\cap\executionsofmodel{\synchmodel}.
    $$
    By Definition~\ref{def:executions-of-cfsms},
    $\executionsofcfsms{\gt}{\acommunicationmodel}\cap\executionsofmodel{\synchmodel} = \executionsofcfsms{\gt}{\synchmodel}$, and on the other hand,
    $$
    \prefixclosureof{\executionsof{\projectionof{\gt}}{\acommunicationmodel}}\cap\executionsofmodel{\synchmodel} 
    ~ \subseteq ~ 
    \prefixclosureof{\executionsof{\projectionof{\gt}}{\acommunicationmodel}\cap\executionsofmodel{\synchmodel}}
    ~ = ~  
    \prefixclosureof{\executionsof{\projectionof{\gt}}{\synchmodel}}.
    $$
    Therefore, we have
    $
    \executionsofcfsms{\projectionof{\acceptcompletion{\gt}}}{\synchmodel}\subseteq\prefixclosureof{\executionsof{\projectionof{\gt}}{\synchmodel}}
    $, which ends the proof of Condition~(DF) of Definition~\ref{def:realisability}.
\end{enumerate}
\end{proof}

We now establish the converse of Theorem~\ref{thm:pp-realizability-implies-synch-realizability}:
we show that the other implication holds under some extra conditions on the projection of $\gt$ in $\acommunicationmodel$, assuming further that $\acommunicationmodel$ is causally closed.

\begin{thm}[Reduction to $\synchmodel$-implementability]
    \label{thm:main-theorem-realisability}
    Assume $\acommunicationmodel$ is a causally closed communication model such that 
    $\executionsofmodel{\acommunicationmodel}\supseteq\executionsofmodel{\synchmodel}$.
    The global type $\gt$ is deadlock-free realisable in $\acommunicationmodel$
        if and only if the following four conditions hold:
        \begin{enumerate}
            \item $\msclanguageof{\projectionof{\gt}}{\acommunicationmodel}\subseteq\prefixclosureof{\mscsetofmodel{\synchmodel}}$
            \item $\projectionof{\gt}$ is orphan-free in $\acommunicationmodel$, 
            \item $\mscsofcfsms{\projectionof{\acceptcompletion{\gt}}}{\acommunicationmodel}\subseteq\prefixclosureof{\mscsofcfsms{\projectionof{\gt}}{\acommunicationmodel}}$, and 
            \item $\gt$ is deadlock-free realisable in $\synchmodel$.
        \end{enumerate}  
\end{thm}
% !TEX root =  ../mainPPDP.tex
%!TEX spellcheck = en_GB

% \begin{lemma}\label{lem:synch-completion}
%     Assume $e$ is a synchronous execution, $\msc$ is a synchronous
%     MSC, and $\mscof{e}\prefixorder \msc$.
%     Then there is a synchronous execution $e'$ such that
%     $\msc=\mscof{e\cdot e'}$.
% \end{lemma}

% \begin{proof}
%     Let $e=s_0r_0s_1r_1..s_nr_n$ be a synchronous execution and 
%     $\msc$ is a synchronous MSC such that $\mscof{e}\prefixorder \msc$.
%     Let $e''=s''_0r''_0s''_1r''_1..s''_{n''}r''_{n''}$ be a synchronous linearization of $\msc$.
%     Since $\mscof{e}\prefixorder \msc$, we have that all events $s_i$ (and their respective $r_i$) 
%     of $e$ are present in $e''$. Let $e'$ be the restiction of $e''$ to the events that are not in $e$.
%     We get that $e'$ is a synchronous execution, and because $\mscof{e}\prefixorder \msc$, 
%     no event of $e$ is causally dependent on an event of $e'$ in $\msc$. Thus, 
%     $e'' \sim e\cdot e'$ and $\mscof{e\cdot e'}=\msc$.
% \end{proof}

%So we can state our main result.

\begin{proof}
We show both sides in turn:
\begin{description}
    \item[$\Rightarrow$]
    We first show that the four conditions are necessary.
    Let $\gt$ be a $\acommunicationmodel$-realisable global type.
    \begin{enumerate}
        \item        
            By Condition~(DF) of Definition~\ref{def:realisability},
            $\projectionof{\gt}$ is deadlock-free in $\acommunicationmodel$.
            By Proposition~\ref{prop:deadlock-free-as-a-property-on-mscs-for-p2p-and-synch},
            $$
                \mscsofcfsms{\projectionof{\acceptcompletion{\gt}}}{\acommunicationmodel}\subseteq\prefixclosureof{\mscsofcfsms{\projectionof{\gt}}{\acommunicationmodel}}.
            $$
            By Condition~(CC) of Definition~\ref{def:realisability},
            $$
                \mscsofcfsms{\projectionof{\gt}}{\acommunicationmodel} = \existentialmsclanguageof{\gt} \subseteq \mscsetofmodel{\synchmodel}.
            $$
            Therefore $\mscsofcfsms{\projectionof{\acceptcompletion{\gt}}}{\acommunicationmodel}\subseteq\prefixclosureof{\mscsetofmodel{\synchmodel}}$.
   
        \item Let $e\in\executionsof{\projectionof{\gt}}{\acommunicationmodel}$;
        we show that $e$ is orphan-free.
        By Condition~(CC) of Definition~\ref{def:realisability},
        $\mscof{e}\in\existentialmsclanguageof{\gt}$, therefore
        $\mscof{e}$ is synchronous, and in particular orphan-free.
       
        \item $\mscsofcfsms{\projectionof{\acceptcompletion{\gt}}}{\acommunicationmodel}\subseteq\prefixclosureof{\mscsofcfsms{\projectionof{\gt}}{\acommunicationmodel}}$ 
            by Condition~(DF) of Definition~\ref{def:realisability} and Proposition~\ref{prop:deadlock-free-as-a-property-on-mscs-for-p2p-and-synch}.

        \item  Follows from Theorem~\ref{thm:pp-realizability-implies-synch-realizability} 
    \end{enumerate}
\item[$\Leftarrow$]    Next we  show that the four conditions are sufficient.
    Let $\gt$ be a global type that verifies the four conditions.
%    \etienne{Ca marche pas! contre-exemple: $\gt = p->q:a + r->q:b$!}
    We prove that $\gt$ is $\acommunicationmodel$-realisable as it satisfies the two conditions of Definition~\ref{def:realisability}.
    \begin{description}
        \item[Condition~(CC)]
          By Proposition~\ref{prop:msc-version-of-cond1-of-realizability},
          we show that 
          $$
          \mscsofcfsms{\projectionof{\gt}}{\acommunicationmodel} \subseteq \existentialmsclanguageof{\gt}.
          $$
          By Condition~(CC) of Theorem~\ref{thm:main-theorem-realisability},
          $$
            \mscsofcfsms{\projectionof{\gt}}{\acommunicationmodel} = \mscsofcfsms{\projectionof{\gt}}{\synchmodel}
          $$
          and by Condition~4 of Theorem~\ref{thm:main-theorem-realisability} and Proposition~\ref{prop:msc-version-of-cond1-of-realizability},
          we have that  
          $$
            \mscsofcfsms{\projectionof{\gt}}{\synchmodel} \subseteq \existentialmsclanguageof{\gt}
          $$
          which concludes the proof of Condition~(CC).
          \item[Condition~(DF)]
        follows from Proposition~\ref{prop:deadlock-free-as-a-property-on-mscs-for-p2p-and-synch}
        and Condition~3 of Theorem~\ref{thm:main-theorem-realisability}. \qedhere
    \end{description}
\end{description}
\end{proof}

A system $\cfsms$ such that 
$\msclanguageof{\cfsms}{\acommunicationmodel}\subseteq
\prefixclosureof{\mscsetofmodel{\synchmodel}}$
is called \emph{realisable with synchronous communications} (RSC) 
in~\cite{germerie-phd}. Germerie~\cite{germerie-phd}
showed that the RSC property is decidable in PSPACE and that
whether a RSC system may reach a regular set of configurations is
also in PSPACE. As a consequence, conditions~1 and 2 of
Theorem~\ref{thm:main-theorem-realisability} can be checked in PSPACE.
In the remainder, we focus on the decidability of condition~3 of Theorem~\ref{thm:main-theorem-realisability}.

An execution $e=(w,\source)$ is RSC~\cite{germerie-phd} if
for all $r\in\receiveeventsof{e}$, $\source(r)=r-1$, i.e., 
all messages that are sent in $e$ are either orphan or immediately received. An MSC $M$ is RSC if it admits
an RSC linearisation. We write $\executionsofmodel{\rscmodel}$ and $\mscsetofmodel{\rscmodel}$ for the sets of RSC executions and MSCs respectively.

\begin{defi}[RegSC]\label{def:regsc-comm-model}
    A set $\aMSCset\subseteq \mscsetofmodel{\rscmodel}$ of 
    RSC MSCs is \emph{RegSC} 
    if there is a NFA $\A_{\aMSCset}$ such that
    $\languageofnfa{\A_{\aMSCset}}=
    \{e \in \executionsofmodel{\rscmodel} \mid \mscof{e} \in \aMSCset\}.
    $
    A communication model $\acommunicationmodel$ is \emph{RegSC} if $\mscsetofmodel{\acommunicationmodel}\cap\mscsetofmodel{\rscmodel}$ is RegSC, i.e.,
    there is an NFA $\A_{\acommunicationmodel}$ such that
    $
    \languageofnfa{\A_{\acommunicationmodel}} = 
    \{e \in \executionsofmodel{\rscmodel} \mid \mscof{e} \in\mscsetofmodel{\acommunicationmodel}\}.
    $
\end{defi}

All communication models mentioned in this paper are RegSC.

\begin{prop}\label{prop:regularity-of-com-models}
    The following communication models are regular: 
    $\bagmodel$, $\ppmodel$, $\causalmodel$, and $\synchmodel$.
\end{prop}

\begin{proof}(Sketch)
    We only detail the case of $\acommunicationmodel=\ppmodel$
    and $\acommunicationmodel=\causalmodel$; the other cases are  simpler and therefore omitted.
    \begin{itemize}
        \item Case of $\acommunicationmodel=\ppmodel$.
            Intuitively, the automaton $\A_{\ppmodel}$ 
            guesses for each send action whether it will be followed by its matching receive action or will remain unmatched. It also remembers
            the set of queues in which unmatched messages are stored, and 
            makes sure that whenever  another message is sent in such a queue, it remains unmatched.
            Formally, the automaton $\A_{\ppmodel}=(Q,\delta,q_0,F)$ is defined as follows:
            $Q=2^{\Procs\times\Procs}\times (\varepsilon\cup \Act)$,
            $q_0=(\emptyset,\varepsilon)$, $F=2^{\Procs\times\Procs}\times \{\varepsilon\}$, and
            $$
            \begin{array}{rcl}
            \delta & \eqdef &   \Big\{
                \big((S,\varepsilon),\send{p}{q}{m},(S,\send{p}{q}{m})\big) \mid (p,q)\not\in S, p,q\in\Procs, m\in\Msg
            \Big\} \\[.7em]
            & \cup &
            \Big\{
                \big((S,\varepsilon),\send{p}{q}{m},(S\cup\{(p,q)\},\varepsilon)\big) \mid p,q\in\Procs, m\in\Msg 
            \Big\} \\[.7em]
            & \cup &
            \Big\{
                \big((S,\send{p}{q}{m}),\recv{p}{q}{m},(S,\varepsilon)\big)     \mid p,q\in\Procs, m\in\Msg  \Big\}
            \end{array}
            $$
        \item Case of $\acommunicationmodel=\causalmodel$.
            We define an NFA $\A$ such that 
            $\languageofnfa{\A}=\executionsofmodel{\rscmodel}\setminus\executionsofmodel{\causalmodel}$;
            The NFA $\A_{\causalmodel}$, hence, can  be 
            defined by complementation.
            Intuitively, the automaton $\A$ accepts
            an execution $e$ if and only there is causal path from an unmatched send $s_1=\send{p}{q}{m}$ to some reception $r_2=\recv{p'}{q}{m'}$ for the same process $q$, i.e., if the first condition of  Definition~\ref{def:causal-model} is not satisfied (note that the other condition is always satisfied for RSC MSCs).
            Formally, the automaton $\A$ is defined as follows:
            $Q=S\times(\Act\cup\{\varepsilon\})$,
            with $S=\{q_\bot,q_\top\}\cup \Procs\times\Procs$,
            $q_0=(q_\bot,\varepsilon)$, $F=\{q_\top\}\times(\Act\cup\{\varepsilon\})$, and
            $$
            \begin{array}{rcl}
            \delta & \eqdef &   \Big\{
                \big((s,\varepsilon),\send{p}{q}{m},(s,\varepsilon)\big) \mid p,q\in\Procs, m\in\Msg, s\in S
            \Big\} \\[.7em]
            & \cup &
            \Big\{
                \big((s,\varepsilon),\send{p}{q}{m},(s,\send{p}{q}{m})\big) \mid p,q\in\Procs, m\in\Msg, s\in S
            \Big\} \\[.7em]
            & \cup &
            \Big\{
                \big((s,\send{p}{q}{m}),\recv{p}{q}{m},(s,\varepsilon)\big)     \mid p,q\in\Procs, m\in\Msg, s\in S  
            \Big\} \\[.7em]
            & \cup &
            \Big\{
                \big((q_\bot,\varepsilon),\send{p}{q}{m},((p,q),\varepsilon)\big)     \mid p,q\in\Procs, m\in\Msg  
            \Big\} \\[.7em]
            & \cup &
            \Big\{
                \big(((p,q),\send{p}{r}{m}),\recv{p}{r}{m},((r,q),\varepsilon)\big)     \mid p,q,r\in\Procs, m\in\Msg  
            \Big\} \\[.7em]
            & \cup &
            \Big\{
                \big(((p,q),\send{p}{q}{m}),\recv{p}{q}{m},(q_\top,\varepsilon)\big)     \mid p,q\in\Procs, m\in\Msg  
            \Big\}. \qedhere
        \end{array}
        $$

    \end{itemize}
\end{proof}

Next we show that for regular communication models, if a system is RSC then the set of its  executions  is also regular. 
\begin{lem}
    \label{lem:regularity-of-mscsofcfsms}
    Assume $\acommunicationmodel$ is RegSC
    and take a system of CFSM $\cfsms$ such that  $\msclanguageof{\cfsms}{\acommunicationmodel}\subseteq \mscsetofmodel{\rscmodel}$.
    Then $\mscsofcfsms{\cfsms}{\acommunicationmodel}$ is RegSC.
\end{lem}

\begin{proof}
    We first show the result for the case where $\acommunicationmodel=\bagmodel$, and then generalise to all RegSC communication models.
    \begin{itemize}
        \item case of $\acommunicationmodel=\bagmodel$.
            Assume $\cfsms=(\A_p)_{p\in\Procs}$, with
            $\A_p=(Q_p,q_{0,p},F_p,\delta_p)$. Let $\shuffleof{\cfsms}$
            denote the shuffle product of all $\A_p$'s, i.e.,
            $\shuffleof{\cfsms}=(Q,q_0,F,\delta)$, where
            $Q=\Pi_{p\in\Procs} Q_p$, $q_0=(q_{0,p})_{p\in\Procs}$, 
            $F=\Pi_{p\in\Procs} F_p$, and
            $\delta$ is defined as follows:
            $$            
                \delta \quad = \quad \bigcup_{a\in\Act\cup\{\varepsilon\}} \Big\{
                \big((q_p)_{p\in\Procs},a,(q'_p)_{p\in\Procs}\big) \mid \exists p\in\Procs, (q_p,a,q'_p)\in\delta_p \mbox{ and } \forall r\neq p, q_r=q_r'
                \Big\}
            $$
            The automaton $\shuffleof{\cfsms}$ accepts a sequence of actions $e$ if and only if
            for all $p\in\Procs$, $\projofon{e}{p}\in\languageofnfa{\A_p}$.
            In particular, if $e\in\executionsofmodel{\rscmodel}$, then
            $e$ is an execution, therefore $e\in\executionsof{\cfsms}{\bagmodel}$ by Definition~\ref{def:executions-of-cfsms}.
            Let $\aMSCset=\mscsofcfsms{\cfsms}{\bagmodel}\cap\mscsetofmodel{\rscmodel}$. We have 
            $\languageofnfa{\shuffleof{\cfsms}\otimes \A_{\rscmodel}}=\{e\in \executionsofmodel{\rscmodel}\mid \mscof{e}\in \aMSCset\}$, therefore $\aMSCset$ is RegSC.
            By hypothesis, $\msclanguageof{\cfsms}{\bagmodel}\subseteq \mscsetofmodel{\rscmodel}$, so $\aMSCset=\mscsofcfsms{\cfsms}{\bagmodel}$. We conclude that $\mscsofcfsms{\cfsms}{\bagmodel}$ is RegSC.
        \item general case of $\acommunicationmodel$ being RegSC.
            Let $\aMSCset=\mscsofcfsms{\cfsms}{\acommunicationmodel}\cap\mscsetofmodel{\rscmodel}$.
            By Definition~\ref{def:executions-of-cfsms},
            $\mscsofcfsms{\cfsms}{\acommunicationmodel}=\mscsofcfsms{\cfsms}{\bagmodel}\cap\mscsetofmodel{\acommunicationmodel}$,
            so 
            $\languageofnfa{\shuffleof{\cfsms}\otimes \A_{\rscmodel}\otimes \A_{\acommunicationmodel}}=\{e\in\executionsofmodel{\rscmodel}\mid \mscof{e}\in\aMSCset\}$, which shows that
             $\aMSCset$ is RegSC. By hypothesis,
            $\msclanguageof{\cfsms}{\acommunicationmodel}\subseteq \mscsetofmodel{\rscmodel}$, so $\aMSCset=\mscsofcfsms{\cfsms}{\acommunicationmodel}$.
            We conclude that $\mscsofcfsms{\cfsms}{\acommunicationmodel}$ is RegSC.
    \end{itemize}
\end{proof}

Similarly given a set of regular RSC MSCs,  its prefix closure is also regular. 
\begin{lem}
    \label{lem:regularity-of-prefixclosure}
    Assume $\aMSCset\subseteq \mscsetofmodel{\rscmodel}$ is RegSC.
    Then $\prefixclosureof{\aMSCset}$ is RegSC.
\end{lem}

\begin{proof}
    Let $\A=(Q,q_0,F,\delta)$ be such that $\languageofnfa{\A}=\{e\in\executionsofmodel{\rscmodel}\mid \mscof{e}\in\aMSCset\}$.
    We construct an NFA $\A'$ such that $\languageofnfa{\A'}=\{e\in\executionsofmodel{\rscmodel}\mid\prefixclosureof{\aMSCset}\}$.
    Intuitively, $\A'$ simulates a run of $\A$ in which for each process $p$ it can stop before reaching the final state. Once it stopped
    on $p$, all further transitions containing an action of $p$ in $\A$ are replaced with $\epsilon$-transitions. Formally, we define $\A'=(Q',q_0',F',\delta')$ as follows: $Q'=F'=Q\times 2^\Procs\times(\Act\cup\{\varepsilon\})$, $q_0'=(q_0,\emptyset,\varepsilon)$, and
    $$
    \begin{array}{rcl}
        \delta' & \eqdef & \Big\{
            \big((q,S,\varepsilon),\send{p}{r}{m},(q',S,\send{p}{r}{m})\big) \mid (q,\send{p}{r}{m},q')\in \delta, p\not\in S
        \Big\} \\[.7em]
        & \cup &
        \Big\{
            \big((q,S,\varepsilon),\send{p}{r}{m},(q',S,\varepsilon)\big) \mid (q,\send{p}{r}{m},q')\in \delta, p\not\in S 
        \Big\} \\[.7em]
        & \cup &
        \Big\{
            \big((q,S,\send{p}{r}{m}),\recv{p}{r}{m},(q',S,\varepsilon)\big) \mid (q,\recv{p}{r}{m},q')\in \delta, r\not\in S 
        \Big\} \\[.7em]
        & \cup &
        \Big\{
            \big((q,S,\varepsilon),\varepsilon,(q',S\cup\{p\},\varepsilon)\big) \mid (q,a,q')\in \delta, a\in \Act_p 
        \Big\} \qedhere
    \end{array}
    $$            
\end{proof}

We can now prove that condition~3 of Theorem~\ref{thm:main-theorem-realisability} is decidable provided that the considered global type satisfies condition~1 of Theorem~\ref{thm:main-theorem-realisability}.

\begin{lem}
    \label{lem:decidability-of-condition-3}
    Assume $\acommunicationmodel$ is RegSC.
    Given a global type G such that $\msclanguageof{\projectionof{\gt}}{\acommunicationmodel}\subseteq\prefixclosureof{\mscsetofmodel{\synchmodel}}$, it is decidable in EXPSPACE complexity whether $$\mscsofcfsms{\projectionof{\acceptcompletion{\gt}}}{\acommunicationmodel}\subseteq\prefixclosureof{\mscsofcfsms{\projectionof{\gt}}{\acommunicationmodel}}.$$
\end{lem}

\begin{proof}
    Since $\msclanguageof{\projectionof{\gt}}{\acommunicationmodel}\subseteq\prefixclosureof{\mscsetofmodel{\synchmodel}}$ (condition~1 of Theorem~\ref{thm:main-theorem-realisability}), we have that
    $\mscsofcfsms{\projectionof{\acceptcompletion{\gt}}}{\acommunicationmodel}\subseteq\prefixclosureof{\mscsetofmodel{\synchmodel}}$.
    Thanks to Lemma~\ref{lem:regularity-of-mscsofcfsms}, 
    $\mscsofcfsms{\projectionof{\acceptcompletion{\gt}}}{\acommunicationmodel}$ and 
    $\mscsofcfsms{\projectionof{\gt}}{\acommunicationmodel}$ are RegSC.
    Thanks to Lemma~\ref{lem:regularity-of-prefixclosure}, $\prefixclosureof{\mscsofcfsms{\projectionof{\gt}}{\acommunicationmodel}}$ is RegSC.
    Let $\A_1,\A_2$ be the NFAs such that
    $$
    \begin{array}{rcl}
    \languageofnfa{\A_1} & = & \{e\in\executionsofmodel{\rscmodel}\mid \mscof{e} \in \mscsofcfsms{\projectionof{\acceptcompletion{\gt}}}{\acommunicationmodel}\} \\
    \languageofnfa{\A_2} & = & \{e\in\executionsofmodel{\rscmodel}\mid \mscof{e} \in \prefixclosureof{\mscsofcfsms{\projectionof{\gt}}{\acommunicationmodel}}\}.
    \end{array}
    $$
    Then condition~3 of Theorem~\ref{thm:main-theorem-realisability} holds if and only if
    $\languageofnfa{\A_1}\subseteq\languageofnfa{\A_2}$.
    Since $\A_1$ and $\A_2$ are of  exponential size in the number of processes, hence in the size $\gt$, and language inclusion for NFAs is decidable in PSPACE, condition~3 of Theorem~\ref{thm:main-theorem-realisability} is decidable in EXPSPACE.
\end{proof}

% It remains to show that the third condition of
% Theorem~\ref{thm:main-theorem-implementability}
% is decidable in PSPACE. 
% This relies on the following lemma that states that 
% condition~3 of Theorem~\ref{thm:main-theorem-realisability} is equivalent to
% being dead-lock free in the sense of Definition~\ref{def:orphan-deadlock-free}.

% \begin{restatable}{lemma}{eventualterminationonmsc}\label{lem:eventual-termination-can-be-checked-on-mscs-for-p2p}
%     $\executionsof{\projectionof{\acceptcompletion{\gt}}}{\acommunicationmodel}=\prefixclosureof{\executionsof{\projectionof{\gt}}{\acommunicationmodel}}$ is equivalent to
%     is equivalent to $\gt$ being deadlock-free in $\acommunicationmodel$.
% \end{restatable}
% \input{proofs/lem-eventual-termination-can-be-checked-on-mscs-for-p2p.tex}

From these results we get an EXPSPACE decision procedure for deadlock-free realisability in asynchronous communication models.

\begin{thm}\label{thm:decidability-of-implementability-in-p2p}
   
   Let $\acommunicationmodel=\ppmodel$ (resp. $\bagmodel$, $\causalmodel$). Given a complementable global type $\gt$ and a complement $\comp{\gt}$,     it is decidable in EXPSPACE complexity whether $\gt$ is deadlock-free realisable in $\acommunicationmodel$.

%       The following problem is in PSPACE:
 %\begin{itemize}
  %      \item Input: a complementable global type $\gt$ and a complement $\gt'$ of $\gt$.
    %    \item Question: is $\gt$ deadlock-free realisable in $\acommunicationmodel$?
    %\end{itemize}
\end{thm}

% \subsection{Case Study}

% In this subsection, we will revisit the two examples prsented in the introduction about global types inspired from real-world protocols :
% MQTT and WebSocket (see Figure~\ref{example-protocols}). 

% \paragraph{\textsc{MQTT} (Simplified Publish/Subscribe).}  
% The global type involves three participants -- two clients $c_1,c_2$ and a broker $b$ -- and consists of an infinite loop where each 
% client can send a message that will be relayed to the other one by the broker.
% Because there is three processes, the global type is \emph{commutation‑closed}, and then, \emph{complementable}.
% Since every message have impies the broker at origin or destination, the synchronous product of the projection of the global type 
% is equal to the global type who is thus is deadlock-free realisable in the $\synchmodel$. 
% In the $\acommunicationmodel$, even if the two clients can send concurrently publish messages, 
% the broker will comsumme all them in a non-determinist order. 
% There is no deadlock since the broker is always able to relay messages and acknowledges to come back in the treatment state $q_0$.
% There is no orphan either, since existance a non received publish message implies that one the client is a non-terminal state.

	\section{Concluding Remarks} \label{sec:concl}	
	  % !TEX root =  ../main.tex
%!TEX spellcheck = en_GB

%\todo[inline]{change figure for monotonicity}
%\todo[inline]{remove figure for arxiv}
%\begin{figure}[t]
%    \centering
%    \input{figures/summary.tex}
%    \caption{Comparing the realisability of global types in synchronous and asynchronous communication models, and the role of complementation.}
%    \label{fig:summary}
%\end{figure}

In this paper, we investigated the realisability  of MPSTs across several synchronous and asynchronous communication models, including causally ordered ($\causalmodel$), peer-to-peer ordered ($\ppmodel$), or unordered ($\bagmodel$). %Figure~\ref{fig:summary} summarises our findings, highlighting the relationship between realisability in these models and the role of complementation.
We reduced the realisability in a complex asynchronous communication model to the simpler problem in the synchronous communication model. For a regular, causally-closed communication model $\acommunicationmodel$, as well as for $\acommunicationmodel=\synchmodel$, we showed that realisability in $\acommunicationmodel$ is decidable in PSPACE for a complementable global type $\gt$ provided with an explicit complement $\comp{\gt}$.
As a consequence, we established that the realisability problem for sender-drive global types is decidable in PSPACE, because they can be complemented with a linear increase in the size. 

% Our results show that existing work -— such as sender-driven global types~\cite{DBLP:conf/cav/LiSWZ23} and protocols with at most three participants -— belongs to effectively complementable classes of global types and thus confirms the decidability thresholds that Alur et al.~\cite{DBLP:journals/tcs/AlurEY05} and Lohrey~\cite{DBLP:journals/tcs/Lohrey03} established.

Traditionally, the realisability problem is considered only for the peer-to-peer communication model. This paper represents a first step towards a parametric framework for realisability checking, with potential applicability across diverse communication models. Throughout the paper, we have made a deliberate effort to explicitly state the hypotheses (related to the communication model) under which our theorems hold. 
A particularly critical assumption is that the communication models in question must be causally closed (see Definition \ref{def:causally-closed-communication-model}). This is still a bit restrictive, as many communication models, including mailbox, bus, or those based on bounded FIFO channels lack this property.
It would be interesting to study realisability for these communication models. 

On the complexity side, we showed a PSPACE upper bound on the realisability problem in $\synchmodel$, and an EXPSPACE upper bound in $\ppmodel$, $\causalmodel$, and $\bagmodel$. We believe that these complexity results can be improved, with more careful constructions (as it has been recently done in the case of sender-driven global types for $\ppmodel$).

Apart from exploring tighter complexity bounds, as for future work we plan to study the class of  complementable global types, possibly setting someun/decidability results. Moreover we plan to extend existing tool ReSCU \cite{DBLP:conf/coordination/DesgeorgesG23} to apply the theoretical insights of this paper.

% We also plan to revisit results on the realisability of MSCs, possibly linking criteria such as send-validity and receive-validity \cite{DBLP:conf/cav/LiSWZ23,DBLP:conf/ecoop/Stutz23} to our synchronous realisability conditions. This opens the door to formally investigating whether global types that are realisable in the synchronous model satisfy send-validity, or if this is a weaker or stronger condition—this remains an open and promising question.

% Our work brings that line of research into the realm of programming language design with MPST: we build upon the insights of MSC realisability (e.g., the importance of safety conditions for realisability) and apply them in a type-theoretic framework to obtain results directly applicable to real programming abstractions for concurrency. 
% In summary, by focusing on $\ppmodel$ and synchronous communication models, we shed new light on the implementability of multiparty protocols in practical distributed programming, overcoming the limitations of sender-driven choice and advancing the state of the art in MPST-based program verification.

\subsection{Related works}
We conclude with a brief excursus on related work. Naturally, the most closely related lines of research are those concerning MPST. Comprehensive surveys on the topic are available in \cite{Coppo2015, DBLP:conf/icdcit/YoshidaG20}.
It is noteworthy that in standard MPST literature, realisability—often referred to as implementability—is typically addressed through syntactic restrictions on global types. In most cases, a global type is  implementable if a projection function exists. These syntactic constraints, along with the existence of the projection, imply that the global type is sender-driven, thus placing it within one of the complementable classes discussed in this paper.
In contrast, a recent paper by Scalas and Yoshida~\cite{DBLP:journals/pacmpl/ScalasY19} offers a novel perspective. Their approach emphasizes session fidelity and deliberately abandons syntactic restrictions on session types. Instead, they adopt semantic criteria—validated through model checking—to establish implementability, marking a significant shift from the conventional syntactic paradigm or our approach.
Although, the comparison is difficult as the underlining models are substantially different, it is worth mentioning~\cite{10.1145/3678232.3678245} that discusses projectability in the context of synchronous MPST. 

The connection between communicating automata and behavioural types were 
first explored by Villard~\cite{villard-phd} for the binary case and Deniélou and Yoshida~\cite{DBLP:conf/esop/DenielouY12} for MPSTs. The connection between MPST and HMSCs was initiated by Stutz~\emph{et al.}~\cite{DBLP:conf/ecoop/Stutz23,Stutz24-phd}.
In this work, we push these connections much further, in the sense that (1) we view MPST as the subclass of HMSCs restricted to synchronous MSCs, and (2) we provide for the first time a simple, natural semantics of global types based on languages of MSCs. This MSC-based semantics of MPST has the appealing property that it is agnostic of the communication model. This semantics plays a pivotal role in our more traditional execution-based semantics of MPST
(Def~\ref{def:global-type-language}), as it puts the communication model parameter at the right place, namely the linearisations of MSCs that should be considered.

Our approach also suggest a rather simple and natural way to define a form of \emph{semantic subtyping}~\cite{DBLP:conf/icalp/CastagnaF05}
for session types (see Definition~\ref{def:subtyping}).
The boundaries of decidability for subtyping of MPST have been recently discussed for instance in \cite{DBLP:journals/lmcs/BravettiCLYZ21,DBLP:conf/esop/LiSW24} but
because of the different nature of the definition of subtyping, such work are loosely related ours.

In~\cite{DBLP:conf/cav/LiSWZ23,DBLP:conf/ecoop/Stutz23}, the authors revisited
the theory of MPST projections through the lenses of implementability (a slight variant of realisability),
and step by step managed to improve the complexity of the decision procedure for implementability from EXPSPACE to PSPACE and more recently NP~\cite{DBLP:journals/pacmpl/LiSWZ25}. Although implementability and deadlock-free realisability are not the same problem, it would be interesting to explore how some of the ideas behind the optimisations they introduced can be adapted to either global types that are not sender-driven or communication models beyond $\ppmodel$. 
For instance, two criteria, called \emph{send validity} and {receive validity}, where introduced in~\cite{DBLP:conf/ecoop/Stutz23}; the second criteria, receive validity,
seems strongly justified by the choice of $\ppmodel$ as the communication model, and is unsound for
other communication models, but it might be interesting to identify variants of these criteria for other communication models beyond $\ppmodel$.
 
Deadlock-free realisability was introduced by Alur et al.~\cite{DBLP:journals/tcs/AlurEY05} under the name of safe realisability. We picked a slightly different name
as the two notions only coincide for causally-closed global types. Indeed, Alur et al. defined deadlock-freedom as an MSC property, whereas we defined it at the level of executions (see Definition~\ref{def:deadlock-free}) and later showed that it can be lifted to MSCs for causally-closed communication models (see Proposition~\ref{prop:deadlock-free-as-a-property-on-mscs-for-p2p-and-synch}). It is worth noting that, for non-causally-closed communication models (for instance $\mbmodel$), there are some systems $\cfsms$ that are deadlocked according to Definition~\ref{def:deadlock-free} but for which 
$
\msclanguageof{\acceptcompletion{\cfsms}}{\mbmodel}\subseteq\prefixclosureof{\msclanguageof{\cfsms}{\synchmodel}}$, i.e., a partial MSC of the system can be extended to an accepting MSC, possibly reordering some unmatched sends.

Alur et al.~\cite{DBLP:journals/tcs/AlurEY05} showed that weak realisability (i.e., condition~1 of Def~\ref{def:realisability}) is undecidable for HMSCs, and 
Lohrey~\cite{DBLP:journals/tcs/Lohrey03} extended their result to safe realisability. The HMSCs they use in their proof are global types (they only describe synchronous MSC languages), which shows that deadlock-free realisability is undecidable for global types in $\ppmodel$. Even more, Lohrey's proof suggests that deadlock-free realisability might be undecidable even in $\synchmodel$ for 5 processes. It is remarkable that, despite these results are more than 20 years old, some corners of shadow remain on the exact number of processes needed for undecidability, or on the extent in which they depend on the asynchrony of the communication model. 

Mazurkiewicz traces~\cite{DBLP:books/ws/95/DR1995} are words over partially commutative alphabets. Synchronous MSCs (but not asynchronous ones) can be seen as Mazurkiewicz traces over the alphabet $\Arrows$ of arrows. Regular trace languages are the ones recognized by commutation-closed finite automata. Zielonka's asynchronous automata are a model of a distributed implementation of a regular trace language. It differs from realizability of MPSTs and HMSCs in several way: first, processes in this model synchronise by rendez-vous (that may involve more than two processes), second processes may exchange additionnal informations during rendez-vous (not just the message label specified by the choreography), and third, the choreography is commutation-closed. As a consequence, all regular trace languages are implementable, but the projection operation is highly non-trivial in this setting. Although quite
far from MPSTs at first sight, some old works on Mazurkiewicz traces could possibly shade some light on a few questions that arise from our work, and which, to the best of our knowledge, remain open, like the decidability of the complementability of a global type, or concise complementation procedures for non commutation-closed global types beyond sender-driven choices.

\bibliographystyle{alphaurl}

\bibliography{refs}

%
% If your work has an appendix, this is the place to put it.

%\appendix
%\input{content/appendix-main}

%\input{ongoing}

%\input{monotonicity}
\end{document}